%% file: main.tex
\def\QuantumSubmissionVersion{1}
\documentclass[a4paper,twocolumn,10pt,unpublished]{quantumarticle}
\pdfoutput=1
\usepackage[T1]{fontenc}
\usepackage[utf8]{inputenc}
\usepackage[english]{babel}
\usepackage{amsmath}
\usepackage{microtype}
\usepackage{amssymb,amsthm,mathtools}
\usepackage{bm}
\usepackage{booktabs}
\usepackage{multirow}
\usepackage{array}
\usepackage{longtable}
\usepackage{tabularx}
\usepackage{enumitem}
\usepackage{graphicx}
\usepackage{xcolor}
\usepackage[numbers]{natbib}
\usepackage{hyperref}
\usepackage[nameinlink,noabbrev]{cleveref}
\hypersetup{hidelinks}

\makeatletter
\let\switch@array\relax
\makeatother

\newtheorem{definition}{Definition}
\newtheorem{theorem}{Theorem}
\newtheorem{corollary}{Corollary}
\newtheorem{proposition}{Proposition}
\newtheorem{lemma}{Lemma}
\newtheorem{remark}{Remark}
\newtheorem{assumption}{Assumption}

\newcommand{\pcharge}{\mathsf{prec}}
\newcommand{\Ftwo}{\mathbb{F}_2}
\DeclareMathOperator{\lcmop}{lcm}
\newcommand{\DataNumber}[2]{#2}
\input{generated/paper_numbers}

\newcommand{\pkResultsSummary}{All five paths completed at $m=4$; semantic
UCC, \texttt{qiskit opt3}, staq rotation folding, and the phase-polynomial
reference also completed at $m=8,16,32,64$, while rebased TKET PauliSimp
returned \texttt{predicate\_error} at those four sizes because its output
fell outside the Hadamard-free symbolic-certificate predicate. Semantic UCC
and the phase-polynomial reference produced the same linear resource sequence,
from $25/21/8$ gates/depth/\texttt{cx} at $m=4$ to $385/321/128$ at $m=64$.
Among the completed external rows, \texttt{qiskit opt3} and staq produced
$17m$ and $10m$ output gates, respectively, as the $r m$ token stream grew.
Status outcomes are recorded separately and omitted from numerical curves.}

\begin{document}

\title{Representation-Dependent Recoverability in Quantum Compilation}

\author{Jinze Yang}
\affiliation{School of Physics, Xidian University, Xi'an 710071, China}

\author{Yangyang Li}
\email{yyli@xidian.edu.cn}
\affiliation{Key Laboratory of Intelligent Perception and Image Understanding of Ministry of Education of China,
Collaborative Innovation Center of Quantum Information of Shaanxi Province, Institute of Interdisciplinary Quantum Science and Technology (IIQST),
School of Artificial Intelligence, Xidian University, Xi'an 710071, China}

\author{Xiu-Hao Deng}
\email{dengxiuhao@iqasz.cn}
\affiliation{Shenzhen International Quantum Academy, Shenzhen 518048, China}
\affiliation{Shenzhen Branch, Hefei National Laboratory, Shenzhen, 518048, China}

\maketitle

\begin{abstract}
Fault-tolerant compilation can disperse high-level structure.
Time-dependent and variational schedules split an accumulated phase
across many rounds, gate synthesis approximates an angle by a
Clifford+$T$ word whose individual gates do not encode that angle, and
randomized compiling represents logical rotations through
sign-randomized fragments and associated corrections.  These
transformations preserve the intended computation within their stated
equivalence and error conventions, but they change how cheaply a
downstream compiler can recover the aggregate phase data.

We formalize this representation-dependent recoverability for a
downstream compiler---one positioned after scheduling, randomization,
and lowering, such as a cloud transpiler or a control-stack stream
optimizer.  The model charges two information channels: output
irreversibly committed before the suffix arrives and a complete restart
state crossing the cut, serialized into control, store, live-window,
parameter, and intermediate-representation fields.  For an $r$-round
modular accumulation of a commuting layer with $m$ generators, use an
$\epsilon$-calibrated grid
$Q_\epsilon=K_\epsilon+1=\Theta(1/\epsilon)$ with exact packing
separation $2\sin(\pi/(2Q_\epsilon))$.  Any compiler that is correct on
every valid stream, makes $p$ forward passes, and fails with probability
at most $\delta$ obeys
$\overline A_p+(2p-1)S\geq
(1-\delta)m\log_2K_\epsilon-h_2(\delta)$, where $\overline A_p$ is
committed output and $S$ is the serialized crossing-state cap.
Within this scheduled family, sign normalization is a bijection of the
residue alphabet, so the folded (compile-time) presentation of
randomized compiling inherits the bound; no claim depends on tracking
frames.  Under the additional, explicitly disclosed frame-tracked
presentation, the Pauli frame is a sub-register of $S$.  A
memory-capped block compiler attains the envelope within a constant
factor, whereas direct semantic aggregation attains the compact-output
corner with logarithmic state.

We then measure the predicted channels rather than infer them.
Instrumented reference compilers expose the state, IR, and commitment
coordinates directly.
The two information channels are not, however, always interchangeable
at par, and we characterize when they are not.  A compiler whose committed prefix is
executed at the cut---applied before the suffix is parsed, and so
obliged to name the coordinates it acts on---pays a toll the state
channel never pays: the entropy $\log_2\binom{m}{k}$ of its own
mask specifying which $k$ coordinates are dispersed and committed early.  A compiler that
defers, or that merely buffers what it has written, is handed that
decision for free by the stream that follows and pays none of it.  We
prove both directions and measure both under a preregistered protocol
frozen before the run.  The measured per-generator toll matches
$\log_2\binom{m}{k}/k$, stays approximately constant in $m$ at fixed
$k/m$, and vanishes for a public mask and for the buffered control.

One downstream experiment fixes the physical stakes without promoting
them to a representation-independent theorem.  Under one disclosed
fixed-total-error Clifford+$T$ synthesis and surface-code model, a
materialize-first pipeline that retains separate rotations costs orders
of magnitude more logical $T$ states and spacetime volume than the
semantic-first pipeline; external tools that reconstruct the aggregate
before synthesis are classified on the semantic-capable side.  The
operational design rule is therefore to preserve scheduling and phase
semantics until aggregation whenever the downstream interface and
latency budget permit it.
\end{abstract}

\section{Introduction}
\label{sec:introduction}

A fault-tolerant compiler is not only a translator; some of its standard
stages also disperse structure by design.  Time-dependent schedules and
variational ansatz layers split a single
accumulated phase across many
rounds~\cite{farhi2014qaoa,peruzzo2014vqe,cerezo2021variational,childs2021trotter}.
Gate synthesis dissolves an angle
into a Clifford+$T$ word whose individual letters do not encode the angle
on their own~\cite{dawson2006solovay,kliuchnikov2013fast,selinger2015efficient,ross2016optimal}.
Randomized compiling replaces each rotation
by a sign-randomized fragment so that coherent errors average
away~\cite{wallman2016noise}.  Each of these transformations
is correct, standard, and desirable for a distinct
reason---yet each takes a quantity that was once written down in
one place and scatters it across an instruction stream.

This paper asks what it costs to reconstruct such a quantity.  Within
the downstream model defined below, the answer is an
information-theoretic tradeoff with two charged channels.  A
compiler positioned downstream of the dispersal---a cloud transpiler
receiving a flat program, or a stream optimizer inside a control
stack---must either carry the scattered information in internal state as
it goes, or commit output before the information has arrived and pay
in output length.  All restart channels, including a reconstructed
phase table or graph, are included in the state term.  The resulting
frontier depends on the accuracy $\epsilon$ and the number of phase
terms $m$; a separate synthesis/QRE experiment evaluates how one
materialized-output path converts that difference into physical cost.

\begin{figure}[tbp]
\centering
\includegraphics[width=\columnwidth]{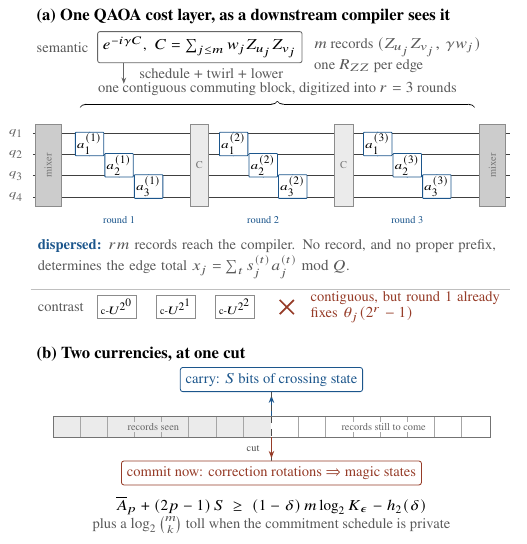}
\caption{Dispersal, and the two currencies that undo it.  (a) A QAOA
cost layer $e^{-i\gamma C}$ on $m$ edges is one commuting object
carrying $m$ records.  Lowered under a digitized schedule it reaches a
downstream compiler as $r$ contiguous rounds of grid residues,
optionally sign-randomized and separated by Clifford records; the edge
total $x_j$ is then a sum that no single record and no proper prefix
reveals.  The mixer layers bound the block rather than sitting inside
it: at generic angles they are non-Clifford $R_X$ rotations
(\cref{rem:round-contiguity}).  The controlled-power ladder beneath is
the complementary case---its rounds are contiguous, but round one
already fixes the aggregate $\theta_j(2^r-1)$, so it supplies
contiguity without dispersal (\cref{rem:dispersal}).  (b) At any cut
the scattered information is either carried across in compiler state or
committed to output.  When those committed rotations are synthesized
separately, the disclosed QRE model converts the extra multiplicity into
additional magic-state demand.
\Cref{thm:w1-unified-frontier} states that the two cannot both be
small, and \cref{prop:commitment-toll} charges the committed channel a
further $\log_2\binom{m}{k}$ when the dispersal mask is private
\emph{and} the prefix is executed at the cut rather than buffered;
a buffered prefix pays none of it.}
\label{fig:intro-dispersal}
\end{figure}

A concrete instance fixes the scale (\cref{fig:intro-dispersal}).  Take
a single QAOA~\cite{farhi2014qaoa} cost layer on a $3$-regular MaxCut instance over $22$
vertices, so $m=33$ edges, compiled to projective accuracy
$\epsilon=10^{-3}$.  The synthesis grid at that tolerance has
$L_\epsilon=\lfloor\pi/(8\arcsin(2\epsilon))\rfloor=196$ and hence
$K_\epsilon=197$ distinguishable settings per edge
(\cref{prop:epsilon-packing}), so each edge total is worth
$\log_2 197=7.62$ bits and the layer as a whole carries $252$ bits, or
about $31$ bytes, of aggregate information.  A deferred compact-output
compiler must carry at least this order of information across the
relevant cut.  In the materialize-first pipeline evaluated here, the
alternative is to retain separate rotations and synthesize them
individually; a generic non-Clifford rotation at this tolerance has a
typical Ross--Selinger cost near
$3\log_2(1/\epsilon)\approx30$ $T$ gates
(\cref{rem:rs-grid}), supplied by distilled magic states rather than by
the code itself~\cite{bravyi2005magic,litinski2019magic}.  This example
motivates, but does not itself prove, the physical conversion.
\Cref{sec:ft-consequences} evaluates it under a common error budget and
finds a \QREMaterializeSpacetimeRatio$\times$ spacetime-volume ratio
between the aggregate and materialized paths in the representative
slice.

The same figure also delineates the scope of the statement, which is
worth clarifying early because a figure is easier to overinterpret than a theorem.
Two hypotheses are required, and different families supply them.  The
rounds must be contiguous in the stream, and they must carry
coefficients that vary freely from round to round.  A variational cost
layer supplies the second within one block of consecutive commuting
rounds, the mixer closing the block rather than lying inside it;
controlled-power ladders and phase estimation supply the first and not
the second, since a ladder's aggregate is already determined by its
first round.  Clifford records---routing, \textsc{swap}, and the
Pauli-frame insertions of randomized compiling itself---are interior
and normalize away when they permute the declared generator set up to
sign (\cref{lem:clifford-interleaving}).  The boundaries
are stated in full in \cref{rem:round-contiguity,rem:clifford-scope}.

Underlying all three dispersal mechanisms is a single representation
problem.  A high-level circuit describing
a commuting diagonal phase layer with $m$ generators can be compiled in at
least two qualitatively different ways. A \emph{semantic-first} pipeline
preserves the generator--coefficient structure, aggregates repeated
coefficients, and lowers the aggregate representative once. A
\emph{materialize-first} pipeline expands the same layer into a flat
target-basis instruction stream and then attempts to recover aggregate
structure from that stream, whether by local rewrites or by constructing
a new global IR.  Both pipelines produce correct unitaries,
but they expose very different resource profiles to downstream synthesis
and fault-tolerant estimation.

This representation dependence is well known anecdotally. Phase-polynomial
and Pauli-gadget
methods~\cite{amy2014tdepth,nam2018automated,cowtan2020phase,amy2019cnot,beaudrap2020techniques}
exploit algebraic structure before returning to a gate circuit.
ZX-calculus
methods~\cite{duncan2020graph,kissinger2020reducing,kissinger2020pyzx,vandewetering2020zx}
provide a global representation in which phases can propagate and merge
beyond local syntax. Architecture-aware
compilers~\cite{martiel2022architecture,meijer2023architecture} and full-stack
toolchains~\cite{watkins2024surfacecode} make representation choices that
propagate into physical resource
estimates~\cite{beverland2022assessing}.  We develop a
quantitative account of \emph{how recoverable} an aggregate is from the
representation in which it arrives---one that isolates representation from
the other variables named in the model, proves lower bounds for the
declared stream families, and measures the charged channels.  Physical resource estimation
is then a consequence of such an account rather than a parallel goal:
after the information tradeoff is established in bits, the observed
rotation multiplicity of particular pipelines can be propagated through
synthesis and error correction under a disclosed model.

We address this gap by formalizing \emph{approximate recoverability}:
given a target unitary $U$ and a fixed approximation tolerance $\epsilon$
in projective operator norm, what compiler resources are required to emit
an output within $\epsilon$ of $U$ from a lowered representation? The
contributions of this paper are:
\begin{enumerate}[leftmargin=1.5em]
    \item An approximate recoverability framework with a budget-explicit
    compiler model charging a unified bit-level cut budget
    $B_{\mathrm{cut}}=B_{\mathrm{control}}+B_{\mathrm{store}}+
    B_{\mathrm{window}}+B_{\mathrm{parameter}}+B_{\mathrm{IR}}$, together
    with a self-delimiting growing-width token codec and explicit charges
    for RAM/DAG/cache layout, dynamic passes, random seeds, output length,
    and approximation error
    (\cref{sec:framework}).
    \item A cyclic $K_\epsilon$-ary packing with exact wrap margin,
    $K_\epsilon=\Theta(1/\epsilon)$, and an $r$-round modular update
    representation in which no token determines the aggregate.  For
    failure probability $\delta$, separate output-code bounds
    $\overline B_{\mathrm{rel}}^{\mathrm{out}},
    \overline B_{\mathrm{self}}^{\mathrm{out}}
    \geq(1-\delta)m\log_2K_\epsilon-h_2(\delta)$, and the explicit
    forward-pass frontier
    $\overline A_p+(2p-1)S\geq
    (1-\delta)m\log_2K_\epsilon-h_2(\delta)$.  A memory-capped block
    compiler reaches compact relative output with
    $S=O((m/p)\log(K_\epsilon+1))$, giving a constant-factor upper
    envelope; the older two-round $K=2,Q=3$ result retains its stronger
    zero-error one-pass constant
    (\cref{sec:packing,sec:reduction,sec:w1-unified-frontier}).
    \item A characterization of when the two payment channels are
    \emph{not} interchangeable at par, stated in both directions.  A
    compiler whose committed prefix is executed at the cut---applied
    before the suffix is parsed, and so obliged to name the coordinates
    it acts on---pays $\log_2\binom{m}{k}$ bits beyond the payload, the
    entropy of the dispersal mask specifying which $k$ coordinates are committed
    early.  A compiler that defers, or that merely buffers what it has
    written, pays none of it, because the stream that follows discloses
    that decision for free.  We prove both clauses and measure both on
    an instrumented device under a preregistered protocol frozen before
    the run; the per-generator toll approaches
    $h_2(k/m)/(k/m)$ rather than being independent of the mask density,
    and the buffered control, which pays zero, is what fixes
    the boundary between the two regimes
    (\cref{def:self-identifying,prop:commitment-toll,fig:commit-toll}).
    \item A representation-aware recovery pipeline with bounded recognition,
    aggregation certificates, and conservative fallback, evaluated on a
    complete nine-representation-by-five-compiler matrix and a separately
    versioned campaign of local, semantic-capable, and authors' reference
    baselines.  Bridge loss and non-completion statuses are reported separately
    (\cref{sec:compiler,sec:validation}).
    \item Direct event-level measurement of committed output, all five
    cross-cut state fields, serialized IR, RSS, and time for echo, aggregation,
    memory-capped hybrid, and materialized-IR reference compilers.
    A separate fixed-total-error synthesis/QRE campaign then propagates
    the observed rotation multiplicities of three certified pipelines
    into physical units: logical $T$, data and factory qubits, cycles, factories,
    runtime, and spacetime volume under two
    versioned QEC models (\cref{sec:cut-tracing,sec:ft-consequences}).
    \item A natural-workload campaign covering structured algorithm families,
    two negative controls, and the complete six-factor ablation, together with
    a normalized frozen-data registry, numeric provenance map, and
    independent-directory reconstruction of theorem, benchmark, and QRE
    subsets (\cref{sec:workloads,sec:reproducibility}).
\end{enumerate}

The prior version of this work established a
representation-dependent recoverability boundary for exact semantics on
fixed-width witness families. Those results---the final-output lower
bound for constant-budget finite-window transducers, the four-qubit
irrational-angle independent-Pauli witness, and the rational CP witness
with period $T=131040$---are retained here as an exact-semantics
corollary (\cref{cor:exact-semantics-corollary}). The primary asymptotic
axes are now the number of independent phase terms $m$ and the requested
approximation tolerance $\epsilon$; the output lower bound scales as
$m\log(1/\epsilon)$.

\section{Operational Setup}
\label{sec:setup}

We define the operational pipeline in four stages:
\emph{semantic input}, \emph{basis lowering}, \emph{recovery}, and
\emph{synthesis/resource estimation}.

\begin{definition}[Semantic input]
\label{def:semantic-input}
A semantic input is a circuit description $C_{\mathrm{sem}}$ in a
representation that preserves generator--coefficient structure. For a
commuting diagonal phase layer with $m$ independent generators
$P_1,\ldots,P_m$ and coefficients $\theta_1,\ldots,\theta_m$, the semantic
input has $O(m)$ records of the form $(P_j, \theta_j)$.
\end{definition}

\begin{definition}[Basis lowering]
\label{def:lowering}
Let $\mathcal{L}_B(\cdot)$ denote exact lowering into a fixed target basis
$B$. The lowered circuit $\mathcal{L}_B(C_{\mathrm{sem}})$ is a flat
instruction stream over $B$ implementing the same unitary up to global
phase. Lowering preserves unitary semantics but need not preserve the
generator--coefficient structure.
\end{definition}

\begin{definition}[Recovery]
\label{def:recovery}
A recovery procedure $\mathcal{R}$ takes either a semantic input or a
lowered circuit $\mathcal{L}_B(C_{\mathrm{sem}})$ and produces an output circuit
$C_{\mathrm{out}}$ such that
$d_{\mathrm{proj}}(U(C_{\mathrm{out}}), U(C_{\mathrm{sem}})) \leq \epsilon$
under the projective operator norm $d_{\mathrm{proj}}$. The recovery is
\emph{semantic-first} if the aggregate structure remains available before
basis lowering; it is \emph{materialize-first} if recovery begins only from
the lowered stream.  A materialize-first procedure may still reconstruct a
global phase-polynomial, Pauli, or ZX representation; the serialized state
needed for that reconstruction is charged by \cref{def:compiler}.
\end{definition}

\begin{definition}[Synthesis and resource estimation]
\label{def:synthesis}
A fault-tolerant synthesis procedure maps $C_{\mathrm{out}}$ to a
discrete-gate circuit over a fault-tolerant gate set (e.g., Clifford+$T$)
subject to a total error budget $\epsilon_{\mathrm{total}}$. A quantum
resource estimator (QRE) then reports logical $T$ count, magic-state
demand, physical qubits, cycles, and spacetime volume under specified
architecture assumptions.
\end{definition}

The central question is: for a fixed target unitary $U$ and fixed
$\epsilon$, how do the compiler resources required by $\mathcal{R}$
depend on whether $\mathcal{R}$ receives $C_{\mathrm{sem}}$ or
$\mathcal{L}_B(C_{\mathrm{sem}})$?

\section{Approximate Recoverability Framework}
\label{sec:framework}

\subsection{Approximate Distance and Recoverability}

\begin{definition}[Projective operator norm distance]
\label{def:proj-norm}
For unitaries $U, V$ on $n$ qubits,
\begin{equation}
\label{eq:proj-norm}
  d_{\mathrm{proj}}(U, V) = \inf_{\phi \in \mathbb{R}} \|U - e^{i\phi} V\|_{\mathrm{op}}.
\end{equation}
  
The infimum is attained because the circle $\{e^{i\phi}\}$ is compact and
$\|U - e^{i\phi}V\|_{\mathrm{op}}$ is continuous in $\phi$.
\end{definition}

The projective distance directly eliminates global phase and serves as
the primary correctness criterion below; see \eqref{eq:proj-norm}.

\begin{lemma}[Unitary-channel diamond comparison]
\label{lem:diamond-comparison}
For the unitary channels
$\mathcal U(X)=UXU^\dagger$ and $\mathcal V(X)=VXV^\dagger$, define the
\emph{unnormalized} diamond distance
$d_\diamond(\mathcal U,\mathcal V)
:=\|\mathcal U-\mathcal V\|_\diamond\in[0,2]$.  Then
\begin{equation*}
 d_{\mathrm{proj}}(U,V)
 \leq d_\diamond(\mathcal U,\mathcal V)
 \leq 2d_{\mathrm{proj}}(U,V).
\end{equation*}
If the normalized convention
$\widetilde d_\diamond:=\tfrac12\|\mathcal U-\mathcal V\|_\diamond$ is used,
the corresponding explicit relation is
$\tfrac12d_{\mathrm{proj}}\leq\widetilde d_\diamond\leq
d_{\mathrm{proj}}$.
\end{lemma}

\begin{proof}
Let $\Theta\in[0,2\pi]$ be the length of a shortest closed arc containing
the spectrum of $U^\dagger V$.  Centering that arc by a global phase gives
\begin{equation*}
 d_{\mathrm{proj}}(U,V)=2\sin(\Theta/4).
\end{equation*}
The numerical-range formula for two unitary channels gives
\begin{equation*}
 d_\diamond(\mathcal U,\mathcal V)=
 \begin{cases}
  2\sin(\Theta/2),&0\leq\Theta\leq\pi,\\
  2,&\pi\leq\Theta\leq2\pi.
 \end{cases}
\end{equation*}
For $\Theta\leq\pi$, the ratio is $2\cos(\Theta/4)\in[\sqrt2,2]$;
for $\Theta\geq\pi$, one has
$d_{\mathrm{proj}}\in[\sqrt2,2]$ and $d_\diamond=2$.  The two inequalities
follow.
\end{proof}

Consequently, a projective packing with separation $s$ is also a diamond
packing with separation at least $s$.  Every lower bound below therefore
holds with the \emph{same numerical accuracy threshold} when correctness is
instead required as
$d_\diamond(\mathcal U(C_{\mathrm{out}}),\mathcal U)\leq\epsilon$.
Under the normalized convention, the same argument uses packing separation
$s/2$ and correctness threshold $\epsilon/2$.
Conversely, a projective-error upper bound $\epsilon$ implies diamond error
at most $2\epsilon$; the exact constructive upper bounds used here have zero
error in either metric.  No diamond packing conclusion is inferred from the
upper inequality alone.

\begin{definition}[$\epsilon$-approximate recoverability]
\label{def:approx-recoverability}
Fix $\epsilon > 0$. Given a target unitary $U$ presented as a lowered
circuit $\mathcal{L}_B(C_{\mathrm{sem}})$, the $\epsilon$-approximate
recoverability problem asks for an output $C_{\mathrm{out}}$ with
$d_{\mathrm{proj}}(U(C_{\mathrm{out}}), U) \leq \epsilon$, using compiler
resources charged by the model below.
\end{definition}

\subsection{Budget-Explicit Compiler Model}

We retain the fixed-width finite-window model only for the exact corollary
and use the following more general randomized, random-access model for the
growing-width approximate results.

\begin{definition}[Sparse formal-parameter serialization]
\label{def:parmodule}
Fix formal angle symbols $\theta_1,\ldots,\theta_M$.  A parameter is stored
in normalized sparse form
\begin{equation*}
\begin{gathered}
 \lambda=\displaystyle\sum_{\ell=1}^{s}c_\ell\theta_{j_\ell}
 +\dfrac{a}{d}\pi,\\
 j_1<\cdots<j_s,\quad c_\ell,a\in\mathbb Z,\quad d\geq1,
\end{gathered}
\end{equation*}
with zero coefficients omitted and $\gcd(a,d)=1$.  Let
$\mathsf{u}(z)$ and $\mathsf{s}(z)$ be fixed self-delimiting codes for
nonnegative and signed integers.  The executable parameter serialization is
\begin{equation}
\label{eq:precision-charge}
\begin{aligned}
 \mathsf{ser}_{\mathrm{par}}(\lambda)
 &=\mathsf{u}(s)\prod_{\ell=1}^{s}
   \bigl(\mathsf{u}(j_\ell)\mathsf{s}(c_\ell)\bigr)\\
 &\quad\cdot\mathsf{s}(a)\mathsf{u}(d),\\
 \pcharge(\lambda)&:=|\mathsf{ser}_{\mathrm{par}}(\lambda)|.
\end{aligned}
\end{equation}
Thus one nonzero formal coefficient costs $O(\log M)$ rather than
$\Omega(M)$ bits.  Qubit and generator locations are not part of this
parameter payload; they are serialized by the complete-token codec below.
\end{definition}

\begin{definition}[Width-aware complete-token codec]
\label{def:token-codec}
Fix a self-delimiting integer code $\mathsf u$, a signed variant
$\mathsf s$, a two-bit Pauli code $\mathsf p$, and a reserved record
terminator $\mathsf d$.  A token $\tau$ with gate-type ID $g$, arity $r$,
ordered operand qubits $q_1,\ldots,q_r$, sparse Pauli support
$\{(z_\ell,\sigma_\ell)\}_{\ell=1}^{s}$, and $k$ parameter slots is encoded
as
\begin{equation*}
\begin{aligned}
\mathsf{ser}_{\mathrm{tok}}(\tau)
={}&\mathsf u(g)\mathsf u(r)
  \prod_{i=1}^{r}\mathsf u(q_i)\;\mathsf u(s)
  \prod_{\ell=1}^{s}
    \bigl(\mathsf u(z_\ell)\mathsf p(\sigma_\ell)\bigr)\\
 &\cdot\mathsf u(k)
  \prod_{a=1}^{k}
    \bigl(\mathsf u(\mathrm{id}_a)
           \mathsf{ser}_{\mathrm{par}}(\lambda_a)\bigr)\mathsf d .
\end{aligned}
\end{equation*}
The zero values of $r$, $s$, or $k$ are explicitly encoded, so the record is
executable and uniquely parseable without out-of-band arity, support, or
delimiter information.  Update, IR-node, and output records use the same
field order, with an opcode specifying which fields are active.

At growing width $n$, an arbitrary qubit or generator choice alone requires
$\lceil\log_2 n\rceil$ bits in any injective codec; the concrete codec above
uses $O(\log n)$ bits per listed address and $O(s\log n)$ bits for an
$s$-sparse Pauli support.  Consequently there is no width-independent
``structural cost per token''.  The literal bit strings
$\mathsf{ser}_{\mathrm{tok}}$ are charged wherever a token is live or
emitted.  Only when $n=n_0$ and a complete-token precision bound $q$ are
both fixed is the token set $\Gamma_{B,n_0}^{\leq q}$ finite; the
fixed-width exact-semantics statements abbreviate it as $\Gamma_B(q)$.
\end{definition}

\begin{definition}[Randomized dynamic-pass compiler and serialized cut budget]
\label{def:compiler}
A compiler $\mathcal A$ may use random-access RAM, stacks, hash tables,
caches, DAG/graph IRs, and intermediate tapes; it may also select a finite
number $p=p(I,\rho)$ of scans or passes dynamically from the input $I$ and
a finite random seed/tape $\rho$.  The external input is append-only on its
first exposure.  Any consumed prefix that remains available for a later
scan---whether through a file, memory mapping, spool, or cache---is live
compiler storage and is charged below.  The final output is write-once.
At every execution cut $c$, serialize a complete restart snapshot in five
disjoint fields:
\begin{itemize}[leftmargin=1.5em]
    \item $\mathsf{ser}_{\mathrm{control}}(c)$: format version, pass and
    finite-control IDs, dynamic pass count and termination/scheduling state,
    all input/output heads and emitted-prefix lengths, end flags, framing
    lengths, and the random seed/tape with its cursor;
    \item $\mathsf{ser}_{\mathrm{store}}(c)$: the persistent and ephemeral
    RAM, stacks, hash-table keys, bucket layout, cache contents and metadata,
    with parameter payloads and IR objects replaced by location tags;
    \item $\mathsf{ser}_{\mathrm{window}}(c)$: the ordered live read-window
    token structures, including the structural fields of
    \cref{def:token-codec} but excluding parameter payloads;
    \item $\mathsf{ser}_{\mathrm{parameter}}(c)$: the location-tagged
    concatenation of every live parameter payload in control, store, window,
    and IR, each encoded by \eqref{eq:precision-charge};
    \item $\mathsf{ser}_{\mathrm{IR}}(c)$: every other live intermediate
    tape/DAG/table object and every rereadable consumed-input block that can
    affect future execution, serialized by node opcode, the structural
    fields of \cref{def:token-codec}, edge endpoints, multiplicity, and
    object boundaries, but with parameter payloads omitted.
\end{itemize}
The corresponding bit-level budget is
\begin{equation}
\label{eq:info-budget}
\boxed{
\begin{aligned}
B_{\mathrm{cut}}(c)
&=B_{\mathrm{control}}(c)+B_{\mathrm{store}}(c)\\
&\quad+B_{\mathrm{window}}(c)+B_{\mathrm{parameter}}(c)
      +B_{\mathrm{IR}}(c).
\end{aligned}}
\end{equation}
where each summand is the literal bit length of the named serialization.
For a fixed run, $B_{\mathrm{cut}}$ is the maximum over its cuts.  For
randomized compilers we use the worst-input expected budget
\begin{equation}
\label{eq:expected-cut}
 \overline B_{\mathrm{cut}}(\mathcal A,\mathcal I)
 :=\max_{I\in\mathcal I}\mathbb E_\rho
       \bigl[\max_c B_{\mathrm{cut}}(c;I,\rho)\bigr].
\end{equation}
The quantities introduced above are therefore
unambiguous bit measures:
$Q:=\max_c B_{\mathrm{parameter}}(c)$ and
$V:=\max_c B_{\mathrm{IR}}(c)$; no theorem below substitutes either one
for the total $B_{\mathrm{cut}}$.  The final output is write-once and is
charged separately by the output codes below.
\end{definition}

\begin{remark}[Model scope and escape channels]
\label{rem:escape-channels}
Equation~\eqref{eq:info-budget} charges all restart information, not only
the finite-control state.  In particular, a ZX graph, phase-polynomial
table, spilled update log, hash table, cache, random seed, or materialized
intermediate tape is serialized in the indicated field, and all coefficients
are serialized in $B_{\mathrm{parameter}}$.  Dynamic rescanning is allowed,
but the pass schedule, heads, and any rereadable old input are part of the
snapshot.  The five fields are disjoint by construction, so parameter bits
are not counted again as store, window, or IR structure.  A compiler may
escape a RAM-only workspace lower bound by materializing a large IR, but it
then pays the same bits in $B_{\mathrm{cut}}$.
\end{remark}

\begin{definition}[Three noninterchangeable output measures]
\label{def:info-budget-bits}
Fix two prefix-free, streaming binary codes.  The self-contained code
$E_{\mathrm{self}}$ serializes a target-basis circuit using gate opcodes,
all qubit/generator addresses, record delimiters, and sparse parameters
using \cref{def:token-codec}; it assumes no family dictionary.  For a public
packing-family dictionary
$\mathcal F_m=(m,P_1,\ldots,P_m,K_{\mathrm{pack}},\Delta,B)$, the conditional code
$E_{\mathrm{rel}}(\cdot\mid\mathcal F_m)$ has four modes: (i) an aggregate mode
consisting of a self-delimiting header and a packed
base-$K_{\mathrm{pack}}$ block of $m$
digits; (ii) a literal mode containing complete target-basis records;
(iii) an ordered-share mode whose header declares a split point and whose
fixed public schedule permits generator addresses to be omitted from its
residue slots, while retaining round-boundary and packed-block delimiters;
and (iv) a $p$-block hybrid mode whose header declares the pass blocks and
passthrough coordinates and whose data section contains packed aggregate
blocks and round-major residue blocks in that public order.  Define
\begin{equation}\label{eq:Bout-def}
\begin{aligned}
 B_{\mathrm{rel}}(C\mid\mathcal F_m)
 &:=|E_{\mathrm{rel}}(C\mid\mathcal F_m)|,\\
 B_{\mathrm{self}}(C)&:=|E_{\mathrm{self}}(C)|,\\
 G(C)&:=\text{number of target-basis gates in }C.
\end{aligned}
\end{equation}
For a compiler, a superscript ``out'' and a maximum over valid inputs denote
the worst-case output value of the corresponding measure; a barred
superscript ``out'' denotes the worst-input expectation over $\rho$.  These quantities
are not converted into one another: conditional description bits,
self-contained circuit bits, and gate count are reported and bounded
separately.
\end{definition}

\input{theory/model_implementation_mapping}

\begin{remark}[Physical scope of the bounded-memory streaming model]
\label{rem:streaming-grounding}
The charged streaming model of \cref{def:compiler}---bounded live
state, explicit pass accounting, append-only first exposure---is the
operating regime of a real-time fault-tolerant control stack, not a
convenience abstraction.  Between rounds, the classical electronics
adjacent to the quantum device holds decoder state, Pauli-frame
records, and instruction buffers under hard latency and buffer
budgets: corrections must be resolved before the next non-Clifford
operation, and frame tracking exists precisely because unbounded
buffering is unavailable at that
boundary~\cite{riesebos2017pauli,watkins2024surfacecode}.  Distributed
and modular architectures motivate the same scope from the
interconnect side: a cross-module interaction splits the
\emph{support} of a gate across modules while its coefficient stays
within one module, so modularity bounds the state that crosses a
partition without by itself dispersing any coefficient.  These
architectures therefore ground the model's memory and pass
constraints only; no claim is made that modular partitioning
instantiates the dispersed streams of \cref{sec:physical-dispersed},
whose instances come from scheduling and twirling alone.
\end{remark}

\subsection{Exact-Semantics Corollary from Prior Work}

We record the exact-semantics ($\epsilon = 0$) results
as corollaries; the full transducer proofs are
restated in \cref{app:sec:full-proofs}. The witness family is
fixed-width.

\begin{remark}[Separate fixed-width and growing-width encoding regimes]
\label{rem:encoding-regimes}
The exact corollary in this subsection fixes $n=n_0$, a precision bound
$q$, a deterministic constant pass count $p_0$, and one finite alphabet
$\Gamma_{B,n_0}^{\leq q}$; its unit is therefore a token over that fixed
alphabet.  No such constant-cost token assumption is made in the
growing-width results.  There $n=m+1$ (or $n\geq m$), the pass count may be
$p=p(m,I,\rho)$, and every input, output, RAM, and IR record is charged in
bits by \cref{def:token-codec} and \eqref{eq:info-budget}.
\end{remark}

\begin{definition}[Fixed-width witness family]
\label{def:witness-family}
Let $H_\Lambda$ be a fixed invertible boundary layer (a Hadamard layer in
the experiments) and let
$D = \prod_{j=1}^{m_0} \exp(-i\theta_j P_j/2)$
be a fixed commuting diagonal layer with Pauli-$Z$ characters
$P_j = Z^{a_j}$, $a_j \in \Ftwo^{n}$, and angles $\theta_j$. The witness
at repetition count $r$ is $C_r = H_\Lambda D^{r} H_\Lambda^{-1}$, with
exact lowering $\mathcal{L}_B(C_r) = \Pi M^r \Sigma$ for fixed token
strings $\Pi, M, \Sigma$.
\end{definition}

Two instances are certified and reused here. The
\emph{rational CP witness} ($n = 4$; ten $rz$/$cp$ terms;
$\operatorname{rank}_{\Ftwo}(A) = 4$) satisfies finite-range noncollision
with exact period $T = 131040$, certified by the self-contained rational
angle, single-term order, relative-phase, and LCM derivation in
\cref{app:sec:rational-period}. The
\emph{irrational independent-Pauli witness} ($n = 4$; supports
$1100, 0110, 0011, 1000$; $\theta_1/\pi = \sqrt{2}/11$) satisfies
unbounded noncollision by the following uniqueness lemma.

\begin{lemma}[Coefficient uniqueness for independent characters]
\label{lem:coeff-uniqueness}
If $a_1,\ldots,a_{m_0} \in \Ftwo^{n}$ are linearly independent and two
diagonal layers with coefficient vectors $\alpha, \alpha'$ satisfy
$D(\alpha) = e^{i\varphi} D(\alpha')$, then
$\alpha_j \equiv \alpha'_j \pmod{2\pi}$ for every $j$. In particular, if
some $\theta_j/\pi$ is irrational, the powers $D^{r}$ are pairwise
distinct up to global phase for all $r \geq 1$.
\end{lemma}

\begin{proof}
Independence makes
$z \mapsto (\chi_{a_1}(z),\ldots,\chi_{a_{m_0}}(z))$ surjective onto
$\{\pm 1\}^{m_0}$. Comparing the phase constraint at two sign patterns
that differ only in coordinate $j$ gives
$2(\alpha_j - \alpha'_j) \equiv 0 \pmod{4\pi}$. For the second claim,
$(r-r')\theta_j \in 2\pi\mathbb{Z}$ with $\theta_j/\pi$ irrational forces
$r = r'$.
\end{proof}

\begin{assumption}[Finite-range noncollision]
\label{ass:noncollision}
On a range $1 \leq r \leq R$, the powers $D^1,\ldots,D^R$ are pairwise
distinct up to global phase.
\end{assumption}

\begin{assumption}[Unbounded noncollision]
\label{ass:unbounded-noncollision}
The powers $D^r$, $r \geq 1$, are pairwise distinct up to global phase.
\end{assumption}

\begin{corollary}[Exact-semantics final-output lower bound]
\label{cor:exact-semantics-corollary}
\label{thm:final-output}
Assume exact lowering ($\epsilon = 0$), a fixed-width witness
$C_r = H_\Lambda D^r H_\Lambda^{-1}$ with
$\mathcal{L}_B(C_r) = \Pi M^r \Sigma$, and unbounded noncollision
(\cref{ass:unbounded-noncollision}). Let $\mathcal{A}$ be a uniform
deterministic finite-window, finite-burst transducer with constant budget
(constant $p,w,\mu$ and $B_{\mathrm{cut}}$) that is correct for every
$r \geq 1$. Then there exist constants
$h \geq 1$, $D \geq 1$ (depending on $\mathcal{A}$ and the witness but not
$r$) such that for every $r \geq h + D$,
\begin{equation}\label{eq:exact-output-bound}
  |T_p(r)| \geq \frac{r-h-D+1}{D} = \Omega_{\mathcal{A}}(r).
\end{equation}
\end{corollary}

\begin{proof}
This is \cref{app:thm:final-output}, proved in full in
\cref{app:sec:full-proofs}. Constant budget gives one finite
charged alphabet; per-tape pumpability propagation gives pumpability of the
final output; noncollision prevents collapse of distinct pump residues.
\end{proof}

\begin{corollary}[Finite-range tail form]
\label{cor:finite-tail}
Under finite-range noncollision (\cref{ass:noncollision}) with $N(R)=R$ on
$1 \leq r \leq R$, $|T_p(r)| \geq (r-h-D+1)/D$ for $h+D \leq r \leq R$.
For the rational CP witness, this applies only for $R < T = 131040$.
\end{corollary}

\begin{corollary}[Counting toll for short outputs]
\label{cor:counting-toll}
If finite-range noncollision supplies $N(R)$ distinct targets, then every
deterministic exact compiler satisfies
\begin{equation*}
 \max_{1\leq r\leq R}B_{\mathrm{self}}(T_p(r))
 \geq \left\lceil\log_2N(R)\right\rceil.
\end{equation*}
The same statement holds for $B_{\mathrm{rel}}(\cdot\mid\mathcal F)$
when one fixed public dictionary $\mathcal F$ is used on the entire range.
\end{corollary}

The irrational independent-Pauli witness instantiates
\cref{cor:exact-semantics-corollary} for all $r$ via
\cref{lem:coeff-uniqueness}; the rational CP witness instantiates
\cref{cor:finite-tail} on ranges $R < T$. These exact-semantics results
concern $\epsilon = 0$ and fixed width $n = 4$; the results below
generalize to growing $m$ and expose the dependence on $\epsilon>0$.

\section{Packing Construction}
\label{sec:packing}

\subsection{Growing-Width Family}

\begin{definition}[Packing family]
\label{def:packing-family}
For each $m$ and accuracy scale $\epsilon$, a packing family
$\mathcal F_{m,\epsilon}=\{U_x:x\in\mathcal X_{m,\epsilon}\}$ has
pairwise projective separation greater than $2\epsilon$, a semantic
description of $O(m)$ addressed records, and a lowered streaming
description whose aggregate information may be dispersed across updates.
The accuracy-dependent construction below has
$|\mathcal X_{m,\epsilon}|=K_\epsilon^m$ with
$K_\epsilon=\Theta(1/\epsilon)$.
\end{definition}

Every growing-width family in this paper---the codebook that carries the
separation, the streams that carry the lower bounds, the two-round
device that is executed in \cref{sec:cut-tracing}, and the schedules of
\cref{sec:physical-dispersed}---is one scheme under four parameters.  We
state it once and list the specializations in
\cref{tab:family-specializations}.

\begin{definition}[Phase-accumulation family]
\label{def:phase-family}
Let $P_1,\ldots,P_m$ be $\Ftwo$-independent Pauli-$Z$ characters on
$n\geq m$ qubits (i.e.\ $P_j=Z^{a_j}$ with
$\operatorname{rank}_{\Ftwo}(a_1,\ldots,a_m)=m$), and write
$R_{P_j}(\theta)=\exp(-i\theta P_j/2)$ and $[K]_0=\{0,\ldots,K-1\}$.  A
family is fixed by a step $\Delta>0$, an aggregate alphabet $[K]_0$, a
number of rounds $r\geq1$, and a sign set $\Sigma\subseteq\{\pm1\}$.  The
\emph{semantic input} is the addressed vector $x\in[K]_0^m$, realizing
\begin{equation}\label{eq:Ux-def}
 U_x=\exp\Bigl(-\tfrac{i\Delta}{2}\textstyle\sum_{j}x_jP_j\Bigr)
    =\prod_{j=1}^{m}R_{P_j}(\Delta x_j).
\end{equation}
The \emph{flat input} is the round-major record sequence
$\mathsf{Update}(t,j,a_j^{(t)},s_j^{(t)})$ with residues $a_j^{(t)}$ and
signs $s_j^{(t)}\in\Sigma$, realizing
\begin{equation}\label{eq:family-flat}
 W=\prod_{t=1}^{r}\prod_{j=1}^{m}R_{P_j}\bigl(\Delta\,s_j^{(t)}a_j^{(t)}\bigr),
\end{equation}
whose aggregate is $x_j=\sum_{t}s_j^{(t)}a_j^{(t)}$.  Two arithmetic
regimes are used, and no result mixes them:
\begin{itemize}[leftmargin=1.4em,topsep=2pt,itemsep=1pt]
\item \emph{no-wrap}: residues and aggregate live in $\mathbb Z$ and
$(K-1)\Delta\leq\pi$, so distinct aggregates are distinct angles;
\item \emph{cyclic}: $\Delta=2\pi/Q$ for a modulus $Q\geq3$, residues lie
in $\mathbb Z_Q$, the aggregate is taken $\bmod\,Q$, and $K=Q-1$ unless
stated otherwise.
\end{itemize}
Since the $P_j$ commute, $W=U_x$ up to global phase in either regime
(\cref{prop:family-aggregates}).
\end{definition}

\begin{proposition}[Flat inputs aggregate projectively]
\label{prop:family-aggregates}
Every valid flat input of \cref{def:phase-family} implements $U_x$ up to
global phase, with no promise about any individual round.
\end{proposition}

\begin{proof}
The $P_j$ commute, so $W=\prod_jR_{P_j}(\Delta\sum_ts_j^{(t)}a_j^{(t)})$.
In the no-wrap regime the exponent is $\Delta x_j$ by definition.  In the
cyclic regime write $\sum_ts_j^{(t)}a_j^{(t)}=x_j+Qk_j$ for integers
$k_j$; since $Q\Delta=2\pi$, $R_{P_j}(\Delta x_j+2\pi k_j)
=(-1)^{k_j}R_{P_j}(\Delta x_j)$, and multiplying over $j$ changes only a
global sign.
\end{proof}

\begin{table}[tbp]
\centering
\footnotesize
\setlength{\tabcolsep}{3.5pt}
\begin{tabular}{@{}clcccl@{}}
\toprule
Def. & Specialization & Regime & $r$ & $\Sigma$ & Role \\
\midrule
\ref{def:packing-construction} & no-wrap packing & no-wrap & $1$ & $\{+\}$ & codebook \\
\ref{def:w1-dispersed-stream} & dispersed update & cyclic & $\geq2$ & $\{+\}$ & lower bounds \\
\ref{def:masked-share} & masked share & cyclic, $Q{=}3$ & $2$ & $\{+\}$ & executed \\
\ref{def:scheduled-stream} & scheduled stream & cyclic & $\geq2$ & $\{\pm\}$ & physical \\
\ref{def:streaming-encoding} & balanced & no-wrap & $r$ & $\{+\}$ & diagnostic \\
\bottomrule
\end{tabular}
\caption{The growing-width families are specializations of
\cref{def:phase-family}.  The balanced encoding is listed for
completeness: its tokens contain $x_j/r$, so it proves no lower bound
(\cref{rem:dispersal}).  The fixed-width witness family of
\cref{def:witness-family} belongs to the separate exact-semantics regime
of \cref{rem:encoding-regimes} and is not an instance of this scheme.}
\label{tab:family-specializations}
\end{table}

\begin{definition}[$K$-ary no-wrap construction]
\label{def:packing-construction}
The no-wrap specialization of \cref{def:phase-family} with $r=1$,
$\Sigma=\{+1\}$, integer $K\geq2$, and $(K-1)\Delta\leq\pi$.  Its
codebook is $\{U_x:x\in[K]_0^m\}$ as in \eqref{eq:Ux-def}.
\end{definition}

\begin{lemma}[Packing separation]
\label{lem:packing-separation}
For the construction in \cref{def:packing-construction},
all $x\neq y\in[K]_0^m$ satisfy
\begin{equation}\label{eq:Kary-separation}
 d_{\mathrm{proj}}(U_x,U_y)\geq2\sin(\Delta/4).
\end{equation}
The constant is tight: pairs differing by one in a single coordinate
attain it.
\end{lemma}

\begin{proof}
Let $s=x-y$ and choose $j^*$ with $s_{j^*}\neq0$.
Since $P_j$ commute, $U_x U_y^\dagger = \exp(-\frac{i\Delta}{2}\sum_j s_j P_j)$
is diagonal in the computational basis. Its eigenvalue at basis state $z$ is
$\lambda(z) = \exp(-\frac{i\Delta}{2}\sum_j s_j \chi_{P_j}(z))$,
where $\chi_{P_j}(z) = (-1)^{a_j \cdot z} \in \{\pm 1\}$.

Because $a_1,\ldots,a_m$ are $\Ftwo$-independent, the map
$z \mapsto (\chi_{P_1}(z),\ldots,\chi_{P_m}(z))$ is surjective onto
$\{\pm 1\}^m$. Hence two basis states can be chosen whose character
vectors agree everywhere except at $j^*$.  Their eigenvalue arguments
differ by $\Delta s_{j^*}$ modulo $2\pi$.  The no-wrap condition gives
\begin{equation*}
 \Delta\leq |\Delta s_{j^*}|\leq(K-1)\Delta\leq\pi .
\end{equation*}

Now include the global-phase infimum required by
\cref{def:proj-norm}. For every $\phi \in \mathbb{R}$,
\begin{equation}\label{eq:packing-sep-norm}
\|U_x - e^{i\phi} U_y\|_{\mathrm{op}}
 = \max_z \bigl|\lambda(z)e^{-i\phi} - 1\bigr|
 \geq 2\sin(\Delta/4),
\end{equation}
because two points on the unit circle whose arguments differ by
$|\Delta s_{j^*}|\in[\Delta,\pi]$ cannot both lie within circular distance
less than $\Delta/2$ of $1$ after any common rotation, and
$|e^{i\gamma} - 1| = 2|\sin(\gamma/2)|$. Taking the infimum over $\phi$
preserves the bound. A one-coordinate difference of one has exactly two
spectral values separated by $\Delta$, and the centered choice of $\phi$
attains the lower bound.
\end{proof}

\begin{proposition}[$\epsilon$-tuned $K$-ary packing]
\label{prop:epsilon-packing}
For $0<\epsilon\leq1/8$, set
\begin{equation}\label{eq:epsilon-codebook}
\begin{aligned}
 L_\epsilon&=\left\lfloor
  \frac{\pi}{8\arcsin(2\epsilon)}\right\rfloor,&
 M_\epsilon&=4L_\epsilon,\\
 K_\epsilon&=L_\epsilon+1,&
 \Delta_\epsilon&=\frac{2\pi}{M_\epsilon}.
\end{aligned}
\end{equation}
Then $K_\epsilon=\Theta(1/\epsilon)$,
$(K_\epsilon-1)\Delta_\epsilon=\pi/2$ (so no coefficient wraps around
$2\pi$), and the $K_\epsilon^m$ targets in
\cref{def:packing-construction} obey
\begin{equation}\label{eq:epsilon-separation}
 d_{\mathrm{proj}}(U_x,U_y)
 \geq2\sin(\Delta_\epsilon/4)\geq4\epsilon
 \qquad(x\neq y).
\end{equation}
\end{proposition}

\begin{proof}
The restriction $\epsilon\leq1/8$ gives $L_\epsilon\geq1$.
Since $\arcsin(2\epsilon)=\Theta(\epsilon)$, the asserted order of
$K_\epsilon$ follows.  The equality in the no-wrap statement is direct.
Finally $L_\epsilon\leq\pi/(8\arcsin(2\epsilon))$, hence
$\Delta_\epsilon/4=\pi/(8L_\epsilon)\geq\arcsin(2\epsilon)$.
Apply \cref{lem:packing-separation} and monotonicity of sine on this range.
\end{proof}

The next lemma shows that the alphabet size is intrinsic rather than
designer-chosen: up to
an explicit factor paid for the no-wrap property, $K_\epsilon$ is the
metric-entropy count of pairwise distinguishable $Z$-rotations at the
working separation, an intrinsic function of the tolerance rather than a
construction parameter.

\begin{lemma}[Synthesis-grid calibration]
\label{lem:grid-calibration}
For $0<\epsilon\leq1/8$, let $N(\epsilon)$ be the maximum cardinality of
a set of single-axis rotations $\{R_P(\theta):\theta\in\mathbb
R/2\pi\mathbb Z\}$ that is pairwise $4\epsilon$-separated in
$d_{\mathrm{proj}}$, and let $g_\epsilon:=4\arcsin(2\epsilon)$.  Then
\begin{equation}\label{eq:grid-calibration}
 N(\epsilon)=\left\lfloor\frac{2\pi}{g_\epsilon}\right\rfloor,
 \qquad
 \frac{1}{2\epsilon}-1\;<\;N(\epsilon)\;\leq\;\frac{\pi}{4\epsilon},
\end{equation}
and the alphabet of \cref{prop:epsilon-packing} satisfies
\begin{equation}\label{eq:grid-calibration-K}
 \frac{N(\epsilon)}{4}\;<\;K_\epsilon\;<\;\frac{N(\epsilon)}{4}+\frac32 .
\end{equation}
Moreover the grid step is separation-tight,
$g_\epsilon\leq\Delta_\epsilon\leq g_\epsilon(1+O(\epsilon))$, so the
factor $4$ in \eqref{eq:grid-calibration-K} is exactly the restriction of
the coefficient range to the quarter period $[0,\pi/2]$ imposed by the
no-wrap condition.
\end{lemma}

\begin{proof}
For a fixed non-identity Pauli $P$, both characters $\pm1$ are realized,
so $R_P(\theta)R_P(\theta')^\dagger$ has eigenvalue arguments
$\mp\delta/2$ with $\delta=\theta-\theta'$; projectively $\theta$ is
defined modulo $2\pi$.  Writing $g\in[0,\pi]$ for the circular distance
between $\theta$ and $\theta'$ in $\mathbb R/2\pi\mathbb Z$, the
global-phase infimum of \cref{def:proj-norm} is attained at the centered
rotation, giving $d_{\mathrm{proj}}=2\sin(g/4)$, increasing in $g$; the
computation is the one in the proof of \cref{lem:packing-separation}.
Hence pairwise $4\epsilon$-separation is equivalent to pairwise circular
distance at least $g_\epsilon$, because $2\sin(g_\epsilon/4)=4\epsilon$
exactly.  The maximum number of points of a circle of circumference
$2\pi$ with pairwise circular distance at least $g_\epsilon$ is
$\lfloor2\pi/g_\epsilon\rfloor$, attained by equal spacing.  The
two-sided bound follows from $x\leq\arcsin x\leq(\pi/2)x$ on $[0,1]$.
For \eqref{eq:grid-calibration-K}, put $y=\pi/(2\arcsin(2\epsilon))$, so
that $N(\epsilon)=\lfloor y\rfloor$ and, by
\eqref{eq:epsilon-codebook}, $K_\epsilon=\lfloor y/4\rfloor+1$.  Then
$K_\epsilon>y/4\geq N(\epsilon)/4$ and
$K_\epsilon\leq y/4+1<(N(\epsilon)+1)/4+1=N(\epsilon)/4+5/4$.  Finally
$\Delta_\epsilon=\pi/(2L_\epsilon)$ with
$\pi/(8\arcsin(2\epsilon))-1\leq L_\epsilon
\leq\pi/(8\arcsin(2\epsilon))$ gives the step bounds.
\end{proof}

\begin{remark}[Clifford+$T$ synthesis calibration of the grid]
\label{rem:rs-grid}
The packing grid is chosen by the target accuracy; it is not itself a
Ross--Selinger synthesis lattice.  Ancilla-free
Clifford+$T$ approximation of an arbitrary $z$-rotation to within
$\epsilon$---in the up-to-phase metric matched to $d_{\mathrm{proj}}$
\cite[Def.~9.1, Cor.~9.5, Alg.~9.8]{ross2016optimal}---is achieved by
the algorithm of \cite[Alg.~7.6]{ross2016optimal} with $T$-count
$3\log_2(1/\epsilon)+O(\log\log(1/\epsilon))$ in the typical case and
$c+4\log_2(1/\epsilon)$ for a constant $c\approx10$ in the worst case
\cite[Prop.~8.8 and \S8.3]{ross2016optimal}; given a factoring oracle
the returned $T$-count is exactly optimal per instance
\cite[Prop.~8.2]{ross2016optimal}.  Generic non-Clifford angles therefore
have synthesis cost $\Theta(\log(1/\epsilon))$, but special grid points
such as the identity or Clifford angles can be cheaper.  The
fixed-total-error campaign of \cref{sec:ft-consequences} does not assume a
uniform per-angle cost: it synthesizes every distinct angle with staq's
\texttt{grid\_synth} implementation and records the checked result.
Consequently the information term $m\log_2K_\epsilon$ and the generic
synthesis scale are both $\Theta(m\log(1/\epsilon))$, while no
angle-by-angle $T$-count equality is claimed.  The constant in the
packing entropy is pinned by \eqref{eq:grid-calibration-K}; both directions
of the calibration are stated as two-sided bounds.  For the cyclic grid used by the dispersed
streams (\cref{cor:w1-epsilon-cyclic}) the identification is exact:
$Q_\epsilon=\lfloor\pi/(2\arcsin(2\epsilon))\rfloor=N(\epsilon)$, so the
sharing modulus equals the metric-entropy count itself, and the factor
$4$ above concerns only the no-wrap subfamily.
\end{remark}

\subsection{Streaming Encoding}

\begin{definition}[Balanced streaming encoding]
\label{def:streaming-encoding}
The no-wrap specialization of \cref{def:phase-family} in which every round
carries the same residue $x_j/r$, emitted on the round-robin schedule
$\sigma(i)=((i-1)\bmod m)+1$ over $N=rm$ slots.  Each token therefore
contains $x_j/r$; \cref{rem:dispersal} records why this encoding proves no
lower bound.
\end{definition}

\begin{definition}[Two-round masked-share lowering]
\label{def:masked-share}
The cyclic specialization of \cref{def:phase-family} with $Q=3$,
$\alpha:=\Delta=2\pi/3$, $r=2$, and $\Sigma=\{+1\}$, in which the
aggregate alphabet is restricted to the binary packing
$K_{\mathrm{pack}}=2$, so $x\in\{0,1\}^m$ while shares range over
$\mathbb Z_3$.  A record carries the round tag, the full generator
address $j$, and one residue in $\mathbb Z_3$; it does not carry $x_j$,
since for every fixed first-round residue both $x_j=0$ and $x_j=1$ admit
a compatible second-round residue.  Each residue has $O(1)$ parameter
payload, while a worst-case complete token has $\Theta(\log m)$
structural address bits as required by \cref{def:token-codec}.
\end{definition}

\begin{remark}[Roles of the two flat schedules]
\label{rem:dispersal}
The balanced encoding of \cref{def:streaming-encoding} is retained only as
the configured experimental input in \cref{sec:validation}; because each
of its tokens contains $x_j/r$, it is not used to prove a representation
lower bound.  The theorem uses the masked-share stream instead.  Before the
second round arrives, every first-round update has two valid binary
aggregate completions, whereas the semantic representation supplies $x$
itself.  The separation below therefore charges recovery of an aggregate
hidden across rounds, not a voluntary delay after a full coefficient was
already visible.
\end{remark}

\subsection{Physically Realized Dispersed Representations}
\label{sec:physical-dispersed}

The masked-share and $r$-round streams above are stated adversarially:
shares are chosen so that no token determines the aggregate.  This
subsection replaces the adversary by three standard fault-tolerant
mechanisms.  Time-dependent scheduling supplies the additive
decomposition, randomized compiling supplies sign symmetry, and the
synthesis grid of \cref{lem:grid-calibration} supplies the modulus.  The
hypotheses of the lower bounds proved in
\cref{sec:reduction,sec:w1-unified-frontier} below are then met by
streams that a fault-tolerant stack emits anyway, not by a constructed
secret sharing.

\begin{definition}[Scheduled phase-accumulation stream]
\label{def:scheduled-stream}
The cyclic specialization of \cref{def:phase-family} with modulus
$Q\geq3$, step $\Delta=2\pi/Q$, $r\geq2$ rounds, and the full sign set
$\Sigma=\{\pm1\}$.  The residue $a_j^{(t)}\in\mathbb Z_Q$ is a round-$t$
physical coefficient rounded to the grid and the sign $s_j^{(t)}$ is a
twirl variable ($s_j^{(t)}\equiv+1$ for untwirled streams), so that the
aggregate of \cref{def:phase-family} reads
\begin{equation}\label{eq:scheduled-aggregate}
 x_j=\sum_{t=1}^r s_j^{(t)}a_j^{(t)}\bmod Q .
\end{equation}
A twirled stream presents its signs in one of two modes.  In
\emph{folded} mode each sign is absorbed at emission,
$a\mapsto sa\bmod Q$, into the residue and its dressing Clifford
records, as in compile-time randomized
compiling~\cite{wallman2016noise,hashim2021randomized}.  In
\emph{tracked} mode the signs accumulate in a classical Pauli-frame
register carried by the control stack and are settled downstream of the
emission point~\cite{knill2005realistically,riesebos2017pauli}.
\end{definition}

\begin{lemma}[Clifford interleaving]
\label{lem:clifford-interleaving}
Let the $r$ rounds of a scheduled stream be separated by Clifford
records, and let $D_{t-1}$ denote the accumulated Clifford emitted
before round $t$ (so $D_0=I$).  Suppose each $D_{t-1}$ carries the
generating set into itself up to sign and permutation,
\begin{equation}\label{eq:clifford-normalizes}
 D_{t-1}^\dagger P_jD_{t-1}=\varepsilon_j^{(t)}P_{\pi_t(j)},
 \qquad \varepsilon_j^{(t)}\in\{\pm1\},
\end{equation}
with $\pi_t$ a permutation of $[m]$.  Then commuting every Clifford
record to the end of the stream yields a scheduled stream of
\cref{def:scheduled-stream} on the same generators, in which the
round-$t$ record of generator $j$ contributes residue $a_j^{(t)}$ with
sign $\varepsilon_j^{(t)}s_j^{(t)}$ to coordinate $\pi_t(j)$, followed
by the fixed Clifford suffix $D_r$.  Both $\pi_t$ and
$\varepsilon^{(t)}$ are functions of the prefix, so the normalization is
prefix-measurable, and every statement about
\cref{def:scheduled-stream} applies to the normalized image.
\end{lemma}

\begin{proof}
Write the emitted unitary as $D_r'L_r\cdots D_1'L_1$ with $L_t$ the
round-$t$ layer and $D_t'$ the Clifford record following it, so
$D_t=D_t'\cdots D_1'$.  Moving each Clifford leftwards replaces $L_t$ by
$D_{t-1}^\dagger L_tD_{t-1}$ and leaves $D_r$ as a suffix.  Since
$L_t=\prod_jR_{P_j}(\Delta s_j^{(t)}a_j^{(t)})$ and conjugation is an
automorphism, \eqref{eq:clifford-normalizes} gives
$D_{t-1}^\dagger L_tD_{t-1}
=\prod_jR_{P_{\pi_t(j)}}(\Delta\varepsilon_j^{(t)}s_j^{(t)}a_j^{(t)})$,
which is a round of \cref{def:scheduled-stream} after the stated
re-indexing.  The Clifford records lie in the stream, so $D_{t-1}$, and
hence $\pi_t$ and $\varepsilon^{(t)}$, are computable from the prefix.
Finally $D_r$ is common to every member of the family and unitary, so it
preserves projective distance and the separation of
\cref{thm:w1-exact-separation} is unaffected.
\end{proof}

\begin{remark}[What the interleaving lemma does and does not cover]
\label{rem:clifford-scope}
Condition \eqref{eq:clifford-normalizes} is a real restriction, not a
formality.  It holds for Pauli-frame layers ($\pi_t=\mathrm{id}$,
$\varepsilon$ the frame signs), hence for the twirl insertions of
randomized compiling itself; and for routing or SWAP layers
whose induced qubit relabeling permutes the generating set, such as the
cyclic shifts of the chain characters $a_j=e_j+e_{j+1}$.  It fails for a
general Clifford: if $D_{t-1}$ normalizes the diagonal group but sends
some $P_j$ to a product of generators rather than to a generator, the
rounds accumulate in different bases and the coordinatewise structure of
\cref{def:scheduled-stream} is lost.  It also does not rescue the
non-commuting interleaving of \cref{rem:round-contiguity}(i): the
layers separating the diagonal rounds of a generic Trotter step are
$R_X$ rotations, which are not Clifford at generic angles.
\end{remark}

\begin{remark}[Downstream compiler scope]
\label{rem:downstream-scope}
The bounds obtained from scheduled streams constrain compilers
positioned \emph{downstream} of the schedule, the twirl, and the
lowering: the input is the emitted record stream, and the upstream
context---the pre-randomization circuit, the schedule generator, and
the twirl seed---is not part of it.  This position is not hypothetical;
it is where much practical compilation runs: (1) a cloud transpiler
receiving a flat QASM/QIR program without access to the user's
high-level intent~\cite{javadi2024qiskit}; (2) a real-time stream
optimizer inside the classical control
stack~\cite{riesebos2017pauli}; (3) a vendor-side pass that
post-processes an already emitted instruction stream.  If the compiler
can see the circuit before randomization and lowering, the theorems do
not apply, and this boundary is itself informative.  The fooling
families range over every valid schedule and sign pattern; a compiler
coupled to the upstream generator---for example, one holding the
$\lambda$-bit seed that pseudorandomly generated the twirl
signs---faces only $2^\lambda$ of the $2^m$ sign patterns, and the
frame lower bound of \cref{prop:frame-cross-cut}(ii) degrades from $m$
to $\lambda$ accordingly.  The theorems therefore quantify exactly what
is forfeited by discarding upstream context before compilation: they
are the price attached to the design rule \emph{randomize and lower
late}.
\end{remark}

\begin{proposition}[Physical instantiation]
\label{prop:physical-instantiation}
Each of the following mechanisms can emit scheduled streams of
\cref{def:scheduled-stream}.  For the families considered here, assume
that the grid-rounded coefficients are free to vary independently across
the $r$ rounds.  Under this coefficient-freedom hypothesis, the resulting
stream family satisfies the hypotheses of
\cref{def:w1-dispersed-stream}: every aggregate
$x\in[K]_0^m$, $K=Q-1$, is realized by some valid stream, and no individual token
or proper prefix determines it.
\begin{enumerate}[leftmargin=1.7em,label=(\alph*)]
\item \emph{Time-dependent Trotterization.}  For a commuting
time-dependent Hamiltonian $H(t)=\sum_jc_j(t)P_j$ with step $\tau$,
round $t$ applies $R_{P_j}(2c_j(t_t)\tau)$; grid rounding gives
$a_j^{(t)}$ and $s\equiv+1$.
\item \emph{Layerwise-parameterized phase layers.}  A depth-$r$ ansatz
whose round-$t$ commuting layer carries free coefficients
$\theta_j^{(t)}$ gives $a_j^{(t)}$ by grid rounding of $\theta_j^{(t)}$.
\item \emph{Randomized compiling of (a) or (b).}  Twirling multiplies
each residue by $s_j^{(t)}$, presented in folded or tracked mode.  If
the twirl bits are independent and uniform, then under the canonical
sharing distribution of \cref{prop:w1-stream-equivalence} every signed
residue is uniform on $\mathbb Z_Q$ and independent of $x_j$.
\end{enumerate}
In contrast, the constant schedule $a_j^{(t)}\equiv a_j$ is degenerate:
its first round determines the aggregate.  This is the balanced-encoding
degeneracy of \cref{rem:dispersal}.
\end{proposition}

\begin{proof}
In (a) and (b) the stated coefficient-freedom hypothesis makes the round
coefficients unconstrained across rounds,
so every residue sequence---in particular every valid share
decomposition of \cref{def:w1-dispersed-stream}---arises from some
schedule, and \cref{prop:w1-stream-equivalence} supplies realizability
and token privacy verbatim.  For prefixes: after any proper prefix, a
final-round residue choice reaches every aggregate in $\mathbb Z_Q^m$,
hence every $x\in[K]_0^m$.  In (c), sign normalization $a\mapsto sa$ is
a bijection of $\mathbb Z_Q$ for either sign, so normalizing every
record maps the signed family onto the unsigned family
record-for-record, and the same properties transfer; the distributional
clause is the one proved in \cref{prop:w1-stream-equivalence}, applied
after normalization.  For the constant schedule, $x_j=ra_j\bmod Q$ is a
function of the first round.  An exhaustive enumeration at
$(m,r,Q)=(2,2,5)$, frozen at
\path{data/frozen/fooling_set_m2_r2_Q5.json}, independently
confirms all three behaviors: every time-varying and every signed
prefix class reaches all $Q^m$ aggregates, and each constant-schedule
prefix class reaches exactly one.
\end{proof}

\begin{remark}[Contiguity of the rounds, and Clifford interleaving]
\label{rem:round-contiguity}
\Cref{def:scheduled-stream} is round-major over \emph{one} commuting
layer: the $r$ rounds that accumulate into $x_j$ reach the compiler as
records of a single layer.  Two consequences should be stated plainly,
since they delimit \cref{prop:physical-instantiation}.

(i) \emph{Non-commuting interleaving is outside the scope, and is not
claimed.}  In a split-operator product formula such as
$H=H_{ZZ}+H_X$, consecutive diagonal rounds are separated by
non-commuting layers.  There the rounds of the diagonal part may not be
aggregated at all: no compact output exists, the premise of the
recoverability question fails, and neither the hypotheses nor the
conclusions of \cref{thm:w1-one-pass-mask,thm:w1-unified-frontier}
are asserted.  \Cref{prop:physical-instantiation}(a) is a statement
about layers whose $r$ rounds are genuinely contiguous---a
time-dependent or shaped schedule of one commuting layer, a
commuting-block Trotterization, a contiguous diagonal ladder---and not
about the interleaved structure of a generic Trotter step.  The same
qualification applies to (b): variational cost layers separated by
mixers are aggregatable only within a block of consecutive commuting
rounds, and it is such a block that (b) describes.  Note that for a
\emph{fully} commuting $H$ the Trotter error vanishes and the round
structure is a scheduling choice rather than an accuracy requirement;
what the bounds charge is precisely that choice, made upstream and
visible to the downstream compiler only as $r$ separate records.

(ii) \emph{Clifford interleaving is inside the scope.}  Rounds separated
by Clifford records normalize back to \cref{def:scheduled-stream}
whenever the accumulated Clifford permutes the generating set up to
sign; this is \cref{lem:clifford-interleaving}, and
\cref{rem:clifford-scope} delimits it.  Contiguity in
\cref{def:scheduled-stream} should therefore be read modulo such
Clifford records.

Contiguity and coefficient freedom are separate hypotheses, and
different families supply them.  Contiguous diagonal rounds arise
structurally in controlled-power ladders and phase estimation, where
$\mathrm{c}\text{-}U^{2^k}$ contributes a sequence of diagonal layers with
no intervening non-commuting structure (the
\texttt{qpe\_controlled\_powers} and \texttt{controlled\_power\_ladder}
families of \cref{sec:workloads}).  Round coefficients that vary freely
across rounds---what the fooling families of
\cref{prop:physical-instantiation} require---are supplied by
time-dependent schedules and by layerwise-parameterized ansatz blocks.
\end{remark}

\begin{proposition}[Frame as cross-cut state]
\label{prop:frame-cross-cut}
Let a compiler in the charged model of \cref{def:compiler} be correct
within projective error $\epsilon<\sin(\pi/(2Q))$ on every valid twirled
stream of \cref{def:scheduled-stream}, and fix any inter-round cut $c$
with committed prefix output of length $\ell_{\mathrm{pre}}$ and
five-field snapshot budget $B_{\mathrm{cut}}(c)$.
\begin{enumerate}[leftmargin=1.9em,label=(\roman*)]
\item In both presentation modes the normalized stream is a valid
dispersed stream of \cref{def:w1-dispersed-stream}, so
\cref{thm:w1-one-pass-mask,thm:w1-unified-frontier} apply verbatim to
the physical stream.
\item In tracked mode,
\begin{equation}\label{eq:frame-lower}
 \ell_{\mathrm{pre}}+B_{\mathrm{cut}}(c)\geq m
\end{equation}
already from the pending-frame classes alone, and the frame register is
live classical restart state, serialized in
$\mathsf{ser}_{\mathrm{store}}$ (or $\mathsf{ser}_{\mathrm{control}}$)
by \cref{def:compiler}.  For a deferred compiler with
$\ell_{\mathrm{pre}}=O(\log m)$, the crossing snapshot therefore
carries $m-O(\log m)$ bits that determine the frame action on the $m$
generators: the Pauli frame is a sub-register of the cross-cut state.
\item In folded mode the flush map is a prefix-measurable sign
normalization, injective on sign components; the same $m$ bits appear
as the sign components of the signed residues, serialized in
$\mathsf{ser}_{\mathrm{parameter}}$ or $\mathsf{ser}_{\mathrm{store}}$,
and the bounds of (i) are unchanged.
\end{enumerate}
\end{proposition}

\begin{proof}
(i) Sign normalization is a bijection of $\mathbb Z_Q$, so normalizing
every record maps the twirled family bijectively onto the family of
\cref{def:w1-dispersed-stream} while preserving aggregates, record
boundaries, and cut positions; the cited theorems apply to the image
and pull back verbatim.

(ii) Take prefixes whose residue fields emit the residue $1$ exactly
once on each generator and are otherwise zero, and whose sign fields
realize the frame class $\sigma\in\{\pm1\}^m$; the residue fields are
identical across classes.  Choose one common suffix adding a residue
$a_j$ on generator $j$ with both $a_j+1\bmod Q$ and $a_j-1\bmod Q$ in
$[K]_0$; such $a_j$ exists for every $Q\geq3$ ($a_j=2$ for $Q=3$,
$a_j=1$ otherwise).  The completed aggregates
$x_j=\sigma_j+a_j\bmod Q$ then differ between two classes in every
coordinate where the frames differ, and both completions are valid
family instances.  By \cref{thm:w1-exact-separation} the corresponding
targets are separated by $2\sin(\pi/(2Q))>2\epsilon$, so a correct
compiler must emit different final outputs; but a deterministic suffix
computation is a function of the snapshot and the common suffix, so two
runs with equal committed prefix and equal snapshot emit identical
final outputs.  The pair (prefix output, snapshot) is therefore
injective on the $2^m$ frame classes; prefix-freeness of the declared
serializers and Kraft's inequality give \eqref{eq:frame-lower}.  The
frame register is classical restart state, so its bits lie in the named
fields; subtracting the deferred-output cap bounds the snapshot share.

(iii) In folded mode the sign is consumed at emission, so the crossing
information consists of the signed residues themselves; injectivity on
sign components is the bijection of (i) restricted to fixed unsigned
residues, and \eqref{eq:frame-lower} holds with the same fooling family
presented in folded records.
\end{proof}

\begin{corollary}[Physically realized lower bounds]
\label{cor:physical-lower-bounds}
For the streams of \cref{prop:physical-instantiation}(a)--(c) with the
tuned modulus $Q_\epsilon$ of \cref{cor:w1-epsilon-cyclic}, the
hypotheses of \cref{def:w1-dispersed-stream,def:w1-pass-accounting} are
met, so \cref{thm:w1-one-pass-mask,thm:w1-unified-frontier} (and, for
the two-round binary case, \cref{thm:main-tradeoff}) hold verbatim on
the physical stream: any compiler positioned downstream of the schedule
and of the twirl obeys
$\overline A_p+(2p-1)S\geq(1-\delta)m\log_2K_\epsilon-h_2(\delta)$, and
under tracked randomized compiling the crossing state contains the
Pauli frame in the sense of \cref{prop:frame-cross-cut}(ii).
\end{corollary}

\begin{proof}
Immediate from
\cref{prop:physical-instantiation,prop:frame-cross-cut}; no separate
argument is required.
\end{proof}

A single knob interpolates between the two regimes, and the
interpolation is itself a theorem: the fraction of coordinates that
reach the compiler in dispersed rather than pre-aggregated form trades
off linearly against the compiler state that survives the cut.  We call
that fraction the \emph{dispersal fraction} $\eta$.  The name is
deliberate: $\eta$ is not the twirl strength of randomized compiling
(the probability with which a random Pauli is applied), and the proof
below never uses the twirl---it conditions on the dispersal mask and
applies \cref{thm:w1-unified-frontier} to the dispersed coordinates.
What $\eta$ measures is the fraction of phase content presented upstream
in late-aggregated form, which is exactly the axis of the
design rule.

\begin{proposition}[Dispersal--memory Pareto bound]
\label{prop:eta-pareto}
Fix a dispersal fraction $\eta\in[0,1]$, $Q\geq3$, $K=Q-1$, and
$0\leq\epsilon<\sin(\pi/(2Q))$.  Let each generator independently with
probability $\eta$ enter the stream in dispersed twirled form (canonical
sharing of \cref{prop:w1-stream-equivalence} with uniform signs), and
otherwise in aggregate-exposed form.  Let $\mathcal A$ be correct within
$\epsilon$ with probability at least $1-\delta$ on every valid stream
and mask, making $p=p(m)$ forward passes in the charged model of
\cref{def:w1-pass-accounting}.  Then, in expectation over the dispersal
mask $T$,
\begin{equation}\label{eq:eta-pareto}
 \mathbb E_T\bigl[\overline A_p+(2p-1)S\bigr]
 \geq(1-\delta)\,\eta m\log_2K-h_2(\delta).
\end{equation}
Under tracked presentation, each twirled generator additionally
contributes its frame bit as in \cref{prop:frame-cross-cut}(ii).
\end{proposition}

\begin{proof}
Condition on the mask $T$.  Fixing the exposed coordinates and
dispersing the coordinates in $T$ yields a subfamily satisfying the
hypotheses of \cref{def:w1-dispersed-stream} with $|T|$ effective
generators, so \cref{thm:w1-unified-frontier} applies verbatim and
gives $\overline A_p+(2p-1)S\geq(1-\delta)|T|\log_2K-h_2(\delta)$ for
every fixed $T$.  Take the expectation over the Bernoulli mask and use
$\mathbb E|T|=\eta m$.
\end{proof}

\begin{figure*}[t]
\centering
\includegraphics[width=0.78\textwidth]{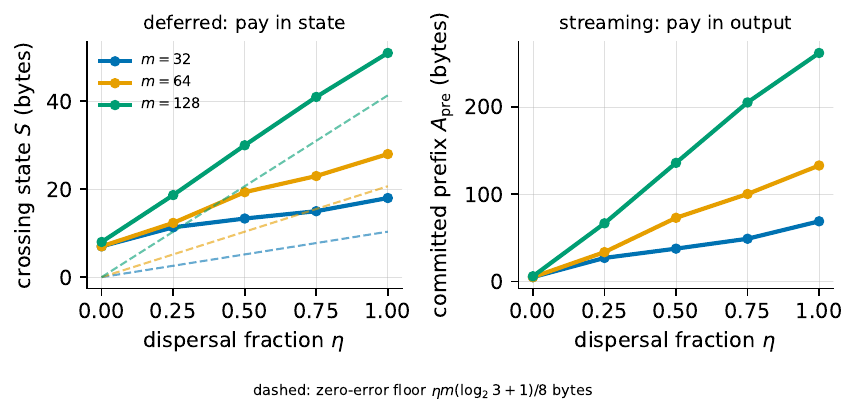}
\caption{Measured dispersal--memory tradeoff on the executable two-round
masked-share device ($Q=3$, binary aggregates, tracked signs; every
point restart-verified: the serialized snapshot alone plus the
round-two suffix reproduces the exact output).  Left: a deferred
compact-output compiler pays the dispersal in crossing state; measured
slopes $10.3/21.1/43.3$ bytes per unit $\eta$ at $m=32/64/128$ sit on
the zero-error floor $\eta m(\log_2 3+1)/8$ bytes (dashed; ratios
$0.99$--$1.05$), the residue-plus-frame cost of
\cref{thm:w1-one-pass-mask} and \cref{prop:frame-cross-cut}.  Right:
a streaming compiler pays the same dispersal in committed prefix output
($A_{\mathrm{pre}}$ slopes $60/129/260$ bytes per unit $\eta$, larger
by the public addresses) while its crossing state stays flat at
$\leq8$ bytes.  The horizontal axis is the dispersal fraction $\eta$ of
\cref{prop:eta-pareto}.  Frozen artifact:
\texttt{data/frozen/eta\_sweep\_pareto.json}, generated by
\texttt{PaperDraft/scripts/run\_eta\_sweep.py}.}
\label{fig:eta-sweep}
\end{figure*}

\Cref{fig:eta-sweep} measures \eqref{eq:eta-pareto} on the executable
masked-share device: the two admissible payment channels of the bound
appear as the two panels, the state channel sits on the predicted floor
to within $5\%$, and the exchange between them is linear in the
dispersal fraction with slope proportional to $m$.

The two panels also expose an asymmetry that the frontier itself does
not predict.  Per dispersed generator the state channel costs about
$2.6$ bits---the floor $\log_2 3+1=2.585$ of residue plus sign---while
the committed channel costs about $16$: an exchange rate near $6{:}1$
between two currencies the frontier charges at par.  The next
proposition isolates the part of that gap that is not an artifact of
the serializer.  This requires one additional hypothesis beyond
\cref{def:compiler}, and the hypothesis is substantive rather than
technical: without it, the gap is not merely unprovable but zero, and we
state that converse alongside the bound.

\begin{definition}[Executed and buffered commitment]
\label{def:self-identifying}
Let $c$ be the inter-round cut of a two-round partially dispersed stream
as in \cref{prop:eta-pareto}.  Let $T\subseteq[m]$ be its dispersal mask:
coordinates in $T$ carry two signed shares, whereas coordinates outside
$T$ arrive in aggregate-exposed form.  Suppose the compiler has
irrevocably written the round-one content
$W:=(a^{(1)}_j,s^{(1)}_j)_{j\in T}$ before $c$.  The commitment is
\emph{executed} if the prefix $O_{\mathrm{pre}}$ is applied at $c$,
before any round-two record is parsed; since the device must be told
which generators to rotate and by how much, an executed prefix
determines the pair $(T,W)$ on its own, without the suffix and without
the cut snapshot.  The commitment is \emph{buffered} if
$O_{\mathrm{pre}}$ is written and never rewritten, but is applied only
after later records have been read, in which case it need not determine
$(T,W)$.  Call $\mathcal A$ \emph{mask-uniform} if a single algorithm,
with no serializer chosen as a function of $T$, is correct within
$\epsilon$ on every valid partially dispersed stream for every mask.
\end{definition}

\begin{proposition}[Commitment toll: a characterization]
\label{prop:commitment-toll}
Fix $Q\geq3$, $K=Q-1$, $\epsilon<\sin(\pi/(2Q))$, $r=2$, and
$1\leq k\leq m$.  Let $\mathcal A$ be mask-uniform and correct within
$\epsilon$ on every valid two-round partially dispersed stream described
in \cref{def:self-identifying}, with a prefix-free committed serializer.
Draw the dispersal mask $T$ uniformly from the
$\binom{m}{k}$ subsets of size $k$, and the committed content $W$
uniformly on $\mathbb Z_Q^k\times\{\pm1\}^k$ independently of $T$.
\begin{enumerate}[label=(\roman*),leftmargin=1.6em]
\item \emph{(Executed commitment pays the mask entropy.)}  If
$\mathcal A$ commits in executed form at the inter-round cut, then
\begin{equation}\label{eq:commitment-toll}
 \mathbb E_{T,W}\bigl[A_{\mathrm{pre}}\bigr]
 \;\geq\;\underbrace{k\log_2(2Q)}_{\text{payload}}
 \;+\;\underbrace{\log_2\tbinom{m}{k}}_{\text{toll}},
\end{equation}
and \emph{a fortiori} $\max_{T,W}A_{\mathrm{pre}}$ obeys the same
bound.  A deferred compiler attains $A_{\mathrm{pre}}=0$ with
$S\leq k\log_2(2Q)+O(\log m)$, so the identical content crosses the cut
for $\log_2\binom{m}{k}$ fewer bits in the state channel.  That
difference is the toll.
\item \emph{(Buffered commitment pays none of it.)}  The hypothesis in
(i) cannot be removed.  The mask-deferred compiler that writes $W$
irrevocably in positional order and recovers $T$ from the round-two
records is admissible in \cref{def:compiler} and correct, commits in
buffered form, and satisfies
$\mathbb E_{T,W}[A_{\mathrm{pre}}]\leq k\log_2(2Q)+O(\log m)$: it pays
the payload and the header and no mask entropy whatever.
\end{enumerate}
The toll is therefore a property of \emph{executed} commitment, not of
commitment as such.  Where it is paid it equals the entropy of the
dispersal mask and is bounded by $h_2(k/m)\big/(k/m)$ bits per
committed generator.  A fixed public mask has zero entropy; within the
uniform size-$k$ ensemble above, the toll therefore vanishes at $k=m$
(and trivially at $k=0$, which the proposition excludes).
\end{proposition}

\begin{proof}
(i) By \cref{def:self-identifying} an executed prefix determines
$(T,W)$, so $H(O_{\mathrm{pre}})\geq H(T,W)$.  Mask uniformity forbids
a serializer selected after $T$ is known, so the committed prefixes
form one prefix-free code over the joint source $(T,W)$, and the
expected length of a prefix-free code is at least the entropy of its
source; hence
$\mathbb E[A_{\mathrm{pre}}]=\mathbb E|O_{\mathrm{pre}}|\geq H(T,W)$.
By construction, $T$ and $W$ are independent and uniform, so
$H(T,W)=\log_2\binom{m}{k}+k\log_2(2Q)$, which is
\eqref{eq:commitment-toll}, and a maximum dominates a mean.  For the
deferred compiler the round-major order is public, so its snapshot holds
the $k$ residues and signs in positional order while the round-two record
types distinguish dispersed updates from aggregate-exposed coordinates
and thereby disclose the mask at no charge;
the $O(\log m)$ term is the header.

(ii) The mask-deferred compiler never rewrites an emitted record, so it
meets \cref{def:compiler}, and it is correct because the round-two
records supply $T$ before the final output is written.  Its prefix
holds $W$ in positional order at $k\log_2(2Q)$ bits plus a header, and
reveals no function of $T$ beyond the value of $k$ already implied by
its length, so no $\log_2\binom{m}{k}$ term arises.  There is no
conflict with (i): this prefix does not determine $(T,W)$, hence is not
executable at the cut, and the compiler has bought its exemption by
giving up the ability to apply what it has written.
\end{proof}

\begin{remark}[Scope of \cref{prop:commitment-toll}, and two degeneracies]
\label{rem:self-id-scope}
Three limits of the statement are worth naming.  First,
\cref{def:self-identifying} is an assumption about the compiler and not
a consequence of \cref{def:compiler}: append-only first exposure fixes
the \emph{semantics} of an emitted record---it is never
reinterpreted---but does not require the record to name its own
targets, and clause (ii) is precisely a compiler that exploits the
difference.  What the definition models is the real-time control stack
of \cref{rem:streaming-grounding}, where the prefix reaches the device
before the suffix is parsed.  Second, the bound is on an ensemble.  For
a single mask fixed in advance and known to the serializer there is
nothing to pay, and the $\log_2\binom{m}{k}$ term is then not merely
unprovable but false; mask uniformity is what makes the schedule a
source rather than a constant.  Third, the payload term is charged
against the declared record content $(a^{(1)},s^{(1)})$ and not against
the round-one unitary, which determines strictly less: a coordinate
carrying residue $0$ acts as identity and is invisible, and
$(s,a)$ and $(-s,-a\bmod Q)$ induce the same rotation, so the applied
unitary alone fixes $(T,W)$ only up to those two degeneracies.  Both
are degeneracies of the payload; neither touches the
$\log_2\binom{m}{k}$ toll, which is precisely the quantity isolated by clause (ii)
and \cref{fig:commit-toll} measures.
\end{remark}

The two channels are therefore \emph{not} at par once commitment is
executed rather than buffered: applying what one has written requires
the prefix to identify the dispersal mask, while deferring---or merely
buffering---obtains that same mask from the stream that follows.  Because
$\log_2\binom{m}{k}/k\to h_2(\eta)/\eta$, the toll,
where it is paid at all, is bounded by a constant---about $2$ bits per
generator at $\eta=\tfrac12$ and $0$ at $\eta=1$---and does not grow
with $m$.

\Cref{fig:commit-toll} tests both clauses on the same device under a
preregistered protocol frozen before the run
(\path{PaperDraft/research/E8_preregistration.md}), with every cell
restart-verified.  Four serializers are compared at
$m\in\{128,512,2048\}$ (and $m$ up to $32768$ for two of them): the
deferred snapshot, the address-per-record serializer of
\cref{fig:eta-sweep}, an identity-optimal executed prefix that ranks
the mask combinadically, and---instantiating
\cref{prop:commitment-toll}(ii)---a mask-deferred control whose prefix
omits the mask and recovers it from the suffix.  Against clause (i),
the measured toll matches $\log_2\binom{m}{k}/k$ within $0.17$ bits at
the worst cell and within $0.02$ bits typically; it is flat in $m$ at
fixed $\eta$ ($2.11/1.99/1.97$ bits at $\eta=\tfrac12$ across a
$16\times$ range of $m$); and it is exactly zero at $\eta=1$.  Against
clause (ii), it is exactly zero at every $\eta$ in the mask-deferred
control.  The control is the informative arm: it pays the same payload
into the same write-once output and rewrites nothing, so what separates
it from the executed encoder is not irreversibility and not addressing
but whether the prefix can be applied before the suffix is read.  The
experiment therefore measures where the class boundary of
\cref{def:self-identifying} lies, rather than confirming a toll that
commitment alone would impose.

This also explains the $6{:}1$ ratio in \cref{fig:eta-sweep}.  Under
the identity-optimal serializer the ratio is $1.00$ at $\eta=1$ and at
most $2.64$ anywhere in the measured range; the remainder is codec
slack, and it is the slack---not the toll---that grows with address
width, from $5.4{:}1$ at $m=128$ to $10.8{:}1$ at $m=32768$.  The
preregistration recorded this as a prediction \emph{against} the
natural reading of \cref{fig:eta-sweep}, and we report it as such: there
is no address-width law here.

Two attempts to carry the toll onto compilers we do not control were
preregistered, and \textbf{both failed}; we report them rather than
repair them.  The first asked only whether real toolchains spend
committed bytes on addresses, and its design placed every rotation on
the low-index qubits, so it never varied the address width it was meant
to vary (a corrected rerun, labelled non-preregistered in the frozen
artifact, finds $2.95$ bytes of growth per rotation over a $1000\times$
range of register width, itself below the preregistered threshold).
The second tested the mechanism directly: two circuits identical in $m$,
in $k$, in rotation count and in angle multiset, differing only in
whether the dispersed subset is an arithmetic progression or a random
$k$-subset, emitted through Qiskit and \textsc{tket} and compared by
compressed size.  Its textual control passed---the variants differ by
under $0.05\%$ of raw bytes, so the comparison is not confounded by text
length---and on \textsc{tket} the excess reaches $86\%$ of
$\log_2\binom{m}{k}$ at $k=m/2$, comfortably above the seed noise.  But
the excess peaks at $k=m/4$ rather than $k=m/2$, and on Qiskit it is
$8\%$ of prediction and smaller than the seed-to-seed noise.  Under the
frozen rule this is a failure, and we make no third-party claim: a
general-purpose compressor applied to a $100$\,kB instruction text is
evidently not a sharp enough entropy estimator for a few hundred bytes
of schedule information.  \Cref{prop:commitment-toll} rests on the
instrumented device of \cref{fig:commit-toll} alone, where every cell is
restart-verified; whether the toll is measurable at a third-party API
boundary remains open.

\begin{figure*}[t]
\centering
\includegraphics[width=0.78\textwidth]{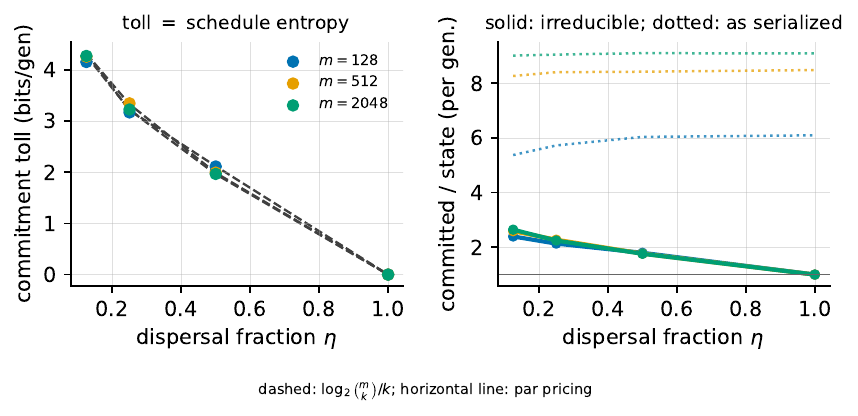}
\caption{The commitment toll, preregistered and measured on the
two-round masked-share device; every cell restart-verified.  Left:
bits per committed generator paid by an executed committed
prefix in excess of the deferred snapshot, against the predicted
mask entropy $\log_2\binom{m}{k}/k$ (dashed).  The toll is flat in
$m$ and vanishes at $\eta=1$; for the buffered mask-deferred control of
\cref{prop:commitment-toll}(ii) it is zero at every $\eta$.  Right: the per-generator exchange rate
between the committed and state channels.  Solid: irreducible, under
an identity-optimal serializer---at most $2.64$, and exactly $1$ when
the dispersal mask is public.  Dotted: the same rate as actually
serialized by the address-per-record codec of \cref{fig:eta-sweep},
which is where the $6{:}1$ comes from and which is the only part that
grows with address width.  Frozen artifact:
\texttt{generated/e8\_commit\_toll.json}, generated by
\texttt{PaperDraft/scripts/run\_commit\_toll.py}.}
\label{fig:commit-toll}
\end{figure*}

\section{Information Lower Bounds}
\label{sec:reduction}

We prove two lower bounds using one-way message arguments~\cite{kushilevitz1997communication}.
The first uses the $K_\epsilon^m$-point no-wrap packing and Fano's inequality
to obtain an explicit accuracy--description tradeoff, including randomized
compilers.  The second injects the $3^m$ possible first-round masks into the
message crossing the inter-round cut.  Its deterministic zero-error bound is
matched by a memory-capped hybrid compiler, and a conditional-Fano extension
covers randomized failure probability $\delta$.

\subsection{Packing Decoders}

\begin{lemma}[Nearest-packing decoder]
\label{lem:nearest-decoder}
Let $\mathcal F=\{U_x:x\in\mathcal X\}$ have pairwise distance at least
$4\epsilon$.  From any output circuit $C$, define
\begin{equation*}
 \widehat x(C)\in\arg\min_{y\in\mathcal X}
       d_{\mathrm{proj}}(U(C),U_y),
\end{equation*}
with a fixed tie rule.  If
$d_{\mathrm{proj}}(U(C),U_x)\leq\epsilon$, then $\widehat x(C)=x$.
\end{lemma}

\begin{proof}
The true target is at distance at most $\epsilon$.  Every other target is
at distance at least $4\epsilon-\epsilon=3\epsilon$ by the triangle
inequality, so it cannot minimize.
\end{proof}

The following binary-coordinate version will be used for the masked-share
tradeoff, where $K=3$ denotes the sharing modulus but the aggregate alphabet
remains $\{0,1\}$.

\begin{lemma}[Bit-class separation]
\label{lem:bit-class-separation}
For the binary subfamily of \cref{def:packing-construction}, put
$\epsilon_0=\sin(\Delta/4)$.  For each $i \in [m]$, partition the family into
$\mathcal{F}_i^{(0)} = \{U_x : x_i = 0\}$ and
$\mathcal{F}_i^{(1)} = \{U_x : x_i = 1\}$. Then the inter-class distance
satisfies
\begin{equation}\label{eq:bit-class-sep}
\operatorname{dist}\!\bigl(\mathcal{F}_i^{(0)},\,\mathcal{F}_i^{(1)}\bigr)
:= \inf_{\substack{x_i=0 \\ y_i=1}} d_{\mathrm{proj}}(U_x, U_y)
\;\geq\; 2\epsilon_0.
\end{equation}
\end{lemma}

\begin{proof}
Any pair with $x_i \neq y_i$ has $x \neq y$, so
\cref{lem:packing-separation} gives $d_{\mathrm{proj}}(U_x, U_y) \geq 2\epsilon_0$.
Taking the infimum preserves the inequality.
\end{proof}

\begin{lemma}[Representation-invariant decoder]
\label{lem:decoder}
For each $i \in [m]$ and any circuit $C$, define
\begin{equation}\label{eq:decoder}
\delta_i(C) :=
\arg\min_{b \in \{0,1\}}\ \min_{y:\, y_i = b} d_{\mathrm{proj}}(U(C), U_y).
\end{equation}
(Bob's local computation is unbounded, so he may exhaustively search the
$2^{m-1}$ candidates of each class.) If
$d_{\mathrm{proj}}(U(C_x), U_x) \leq \epsilon < \epsilon_0$, then
$\delta_i(C_x) = x_i$.
\end{lemma}

\begin{proof}
The correct class $b = x_i$ contains $U_x$ itself, so its inner minimum is
at most $\epsilon$. Suppose the wrong class $b \neq x_i$ also attained an
inner minimum $\leq \epsilon$, realized by some $U_y$ with $y_i = b$. Then
the triangle inequality gives
$d_{\mathrm{proj}}(U_x, U_y) \leq 2\epsilon < 2\epsilon_0$, contradicting
\cref{lem:bit-class-separation}. Hence only the correct class attains the
minimum, i.e.\ $\delta_i(C_x) = x_i$.
\end{proof}

\subsection{Output-Description Lower Bound}

\begin{theorem}[Randomized $\epsilon$--description tradeoff]
\label{thm:output-counting}
Let $0<\epsilon\leq1/8$, let $\mathcal F_{m,\epsilon}$ be the
$K_\epsilon$-ary family of \cref{prop:epsilon-packing}, and let a possibly
randomized compiler $\mathcal A$ output $C_{x,\rho}$ satisfying, for every
$x$,
\begin{equation}\label{eq:random-success}
 \Pr_\rho\!\left[
 d_{\mathrm{proj}}(U(C_{x,\rho}),U_x)\leq\epsilon
 \right]\geq1-\delta,
 \qquad 0\leq\delta<1/2.
\end{equation}
The compiler may use random-access RAM/DAG IR, hash tables, caches, and a
dynamic pass count $p=p(m,x,\rho)$ as in \cref{def:compiler}.  Writing
$h_2(\delta)=-\delta\log_2\delta-(1-\delta)\log_2(1-\delta)$ with
$0\log 0:=0$, each output
code separately obeys
\begin{equation}\label{eq:output-counting}
\begin{aligned}
 \max_x\mathbb E_\rho
 B_{\mathrm{rel}}(C_{x,\rho}\mid\mathcal F_{m,\epsilon})
 &\geq(1-\delta)m\log_2K_\epsilon-h_2(\delta),\\
 \max_x\mathbb E_\rho B_{\mathrm{self}}(C_{x,\rho})
 &\geq(1-\delta)m\log_2K_\epsilon-h_2(\delta).
\end{aligned}
\end{equation}
In particular, for fixed $\delta<1/2$ both are
$\Omega(m\log(1/\epsilon))$.

At $\delta=0$, the family-relative bound is tight to lower-order framing
terms: a deterministic semantic compiler uses
\begin{equation}\label{eq:relative-epsilon-upper}
 B_{\mathrm{rel}}(C_x\mid\mathcal F_{m,\epsilon})
 \leq\lceil m\log_2K_\epsilon\rceil
       +O(\log m+\log K_\epsilon).
\end{equation}
Separately, for $P_j=Z_j$ and $RZ\in B$, the explicit self-contained
circuit satisfies
$B_{\mathrm{self}}(C_x)=O(m(\log m+\log K_\epsilon))$ and $G(C_x)\leq m$;
no self-contained-bit or gate-count tightness claim is made.
\end{theorem}

\begin{proof}
Let $X$ be uniform on $[K_\epsilon]_0^m$ and let $C=C_{X,\rho}$.
By \eqref{eq:epsilon-separation}, the nearest-packing decoder of
\cref{lem:nearest-decoder} has error probability at most $\delta$.  Fano's
inequality gives
\begin{equation*}
 H(X\mid C)\leq h_2(\delta)
       +\delta\log_2(K_\epsilon^m-1),
\end{equation*}
and hence
$I(X;C)\geq(1-\delta)m\log_2K_\epsilon-h_2(\delta)$.
For either prefix-free code $E$, source coding gives
$\mathbb E|E(C)|\geq H(C)\geq I(X;C)$.  The maximum conditional expected
length over $x$ is at least this uniform average, proving each line of
\eqref{eq:output-counting}.  This argument depends only on the output
distribution, so random access, dynamic passes, and internal IR do not
weaken it.

For \eqref{eq:relative-epsilon-upper}, the public dictionary supplies
$(P_1,\ldots,P_m,K_\epsilon,\Delta_\epsilon,B)$ and aggregate mode packs
the base-$K_\epsilon$ vector into one fixed-length integer block.  Its
header serializes $m,K_\epsilon$, the mode, and the block boundary.  In the
canonical self-contained construction, each nonzero $Z_j$ term has an
$O(\log m)$ address and an $O(\log K_\epsilon)$ rational-$\pi$ parameter
under \eqref{eq:precision-charge}; there are at most $m$ terms.
\end{proof}

\begin{corollary}[Deterministic fixed-grid specialization]
\label{cor:fixed-grid-output}
For any no-wrap $(K,\Delta)$ family of
\cref{def:packing-construction}, a deterministic compiler accurate to
$\epsilon<\sin(\Delta/4)$ has, separately,
\begin{equation*}
 B_{\mathrm{rel}}^{\mathrm{out}}\geq\left\lceil m\log_2K\right\rceil,
 \qquad
 B_{\mathrm{self}}^{\mathrm{out}}\geq\left\lceil m\log_2K\right\rceil.
\end{equation*}
In particular, the binary fixed-step family yields the linear-in-$m$ bound used in the
configured scaling panel.
\end{corollary}

\begin{proof}
Accuracy balls are disjoint by \cref{lem:packing-separation}, so the
deterministic output map is injective on $K^m$ inputs.  If $\ell_x$ are
their lengths under either prefix-free code, Kraft's inequality gives
$\sum_x2^{-\ell_x}\leq1$.  Thus
$K^m2^{-\max_x\ell_x}\leq1$, proving the two rounded bounds.
\end{proof}

\begin{remark}[What the counting bound does and does not say]
\label{rem:pass-through}
\Cref{eq:output-counting} holds regardless of internal resources, but the
two inequalities must not be cross-compared.  The packed $K_\epsilon$-ary
block in \eqref{eq:relative-epsilon-upper} is an upper bound only for
$B_{\mathrm{rel}}(\cdot\mid\mathcal F_{m,\epsilon})$; it is not a
self-contained target-basis circuit string.  Likewise $G(C)$ is never
multiplied by an implicit ``bits per gate'' constant.  The
representation-dependent question is addressed by the masked-share cut
below under one fixed family-relative output code.
\end{remark}

\subsection{The Masked-Share Compact-Output--Cut Tradeoff}

The semantic input supplies the aggregate vector $x$.  The flat input of
\cref{def:masked-share} instead supplies an arbitrary first share
$u\in\mathbb Z_3^m$ and only later a second share $v$ satisfying
$u+v=x\pmod 3$.  The cut between these rounds exposes information that is
specific to the representation: a correct compiler must preserve enough
about $u$ to handle every still-valid continuation $v$.

\begin{definition}[Masked-share cut and precommitted output]
\label{def:cuts}
Run $\mathcal A$ on a masked-share stream and let $c_{\mathrm{MS}}$ be the
arrival cut after all $m$ first-round updates have been delivered and before
any second-round update is available.  Before this cut the compiler may use
arbitrary RAM/DAG operations and a dynamic number of internal scans.  A
buffered or rereadable first-round record is included in
$B_{\mathrm{store}}$ or $B_{\mathrm{IR}}$; the only excluded prefix is one
irreversibly committed to the final write-once output.  Let
\begin{equation}\label{eq:Bcom-def}
\begin{aligned}
 B_{\mathrm{pre}}^{\mathrm{rel}}(\mathcal A,m\mid\mathcal F_m)
 &:=\max_{u\in\mathbb Z_3^m}
 \ell_{\mathrm{pre}}(u),\\
 B_{\mathrm{cut}}^{\mathrm{MS}}(\mathcal A,m)
 &:=\max_{u\in\mathbb Z_3^m}B_{\mathrm{cut}}(c_{\mathrm{MS}}(u)).
\end{aligned}
\end{equation}
Here $\ell_{\mathrm{pre}}(u)$ is the literal bit length already committed
under the prefix-monotone stream $E_{\mathrm{rel}}$; it is zero if final
output has not started.  Thus
$B_{\mathrm{pre}}^{\mathrm{rel}}\leq
B_{\mathrm{rel}}^{\mathrm{out}}$.
For randomized $\mathcal A$, barred quantities replace the maxima in
\eqref{eq:Bcom-def} by
$\max_u\mathbb E_\rho[\ell_{\mathrm{pre}}(u,\rho)]$ and
$\max_u\mathbb E_\rho[B_{\mathrm{cut}}(c_{\mathrm{MS}}(u,\rho))]$.
\end{definition}

\begin{theorem}[Representation-dependent masked-share tradeoff]
\label{thm:main-tradeoff}
Let the sharing modulus be $K_{\mathrm{share}}=3$, let $\alpha=2\pi/3$, and
$\epsilon_0=\sin(\pi/6)=1/2$.  Suppose first that a deterministic compiler
$\mathcal A$ is correct for \emph{every} valid two-round decomposition in
\cref{def:masked-share}:
\begin{equation*}
\begin{gathered}
 d_{\mathrm{proj}}(U(C(u,v)),U_{u+v\bmod 3})
 \leq\epsilon<\epsilon_0,\\
 u+v\bmod3\in\{0,1\}^m.
\end{gathered}
\end{equation*}
Then, under the single family-relative output code of
\cref{def:info-budget-bits},
\begin{equation}\label{eq:main-tradeoff}
 B_{\mathrm{pre}}^{\mathrm{rel}}(\mathcal A,m\mid\mathcal F_m)
 +B_{\mathrm{cut}}^{\mathrm{MS}}(\mathcal A,m)
 \geq \left\lceil m\log_2 3\right\rceil.
\end{equation}
Consequently, if the final output is family-relative compact,
\begin{equation*}
 B_{\mathrm{rel}}^{\mathrm{out}}(\mathcal A,m\mid\mathcal F_m)
 \leq m+c\log_2(m+2),
\end{equation*}
then
\begin{equation}\label{eq:compact-cut-floor}
 B_{\mathrm{cut}}^{\mathrm{MS}}(\mathcal A,m)
 \geq \left\lceil m\log_2 3\right\rceil
       -m-c\log_2(m+2)=\Omega(m).
\end{equation}
In contrast, a one-pass compiler receiving the semantic aggregate vector
$x$ streams the same aggregate-mode code with
$B_{\mathrm{rel}}^{\mathrm{out}}=m+O(\log m)$ and
$B_{\mathrm{cut}}=O(\log m)$.

More generally, let $0\leq\delta<1/2$ and let $\mathcal A$ be randomized
and satisfy the displayed
accuracy condition on every valid $(u,v)$ with probability at least
$1-\delta$.  It may use arbitrary random access and a dynamic
$p=p(m,u,v,\rho)$.  Then
\begin{equation}\label{eq:randomized-main-tradeoff}
 \overline B_{\mathrm{pre}}^{\mathrm{rel}}
 +\overline B_{\mathrm{cut}}^{\mathrm{MS}}
 \geq(1-\delta)m-h_2(\delta).
\end{equation}
In particular, for a deferred aggregate-mode compiler with
$\overline B_{\mathrm{pre}}^{\mathrm{rel}}=O(\log m)$ and fixed
$\delta<1/2$, the flat representation still requires
$\overline B_{\mathrm{cut}}^{\mathrm{MS}}=\Omega(m)$.
\end{theorem}

\begin{proof}
Let Alice hold an arbitrary first share $u\in\mathbb Z_3^m$.  She simulates
$\mathcal A$ through $c_{\mathrm{MS}}$ and sends the already committed
family-relative output prefix together with the complete five-field restart
serialization of \eqref{eq:info-budget}.  All head positions, lengths,
addresses, parameters, stores, windows, and materialized IR prefixes needed
to resume are present by definition, including any rereadable old input.
Bob, given a second share $v$, resumes the computation and runs its
dynamically chosen remaining operations locally.

The message must distinguish every pair $u\neq u'$.  To see this, choose a
single continuation $v$ coordinatewise.  If $u_j=u'_j$, set
$v_j=-u_j$.  If $u'_j=u_j+1\pmod3$, again set $v_j=-u_j$; the two aggregates
are then $0$ and $1$.  In the remaining case
$u_j=u'_j+1\pmod3$, set $v_j=-u'_j$; the aggregates are $1$ and $0$.
Thus both $(u,v)$ and $(u',v)$ are valid, binary aggregate instances, and
their aggregate vectors differ whenever $u\neq u'$.

If Alice's two messages were identical, determinism would make Bob emit the
same final circuit on these two valid completions.  That circuit cannot be
within $\epsilon$ of both targets, because
\cref{lem:packing-separation} gives target distance at least
$2\epsilon_0>2\epsilon$.  Hence the message map is injective on the
$3^m$ choices of $u$.  Serialize the restart snapshot first; its control
field gives the length of the committed prefix that follows, so the complete
messages are prefix-free.  Kraft's inequality therefore gives maximum
message length at least $\lceil m\log_2 3\rceil$.  The sum of the two
separate worst-case field lengths is at least that maximum, which proves
\eqref{eq:main-tradeoff}.  Prefix monotonicity and the assumed
upper bound on the final output give \eqref{eq:compact-cut-floor}.

For the semantic upper bound, the compiler writes the aggregate-mode header,
copies the arriving bits $x_j$ into the support-mask field, and retains only
a self-delimiting index and finite control.  It stores no coefficient table
or intermediate graph.  The sparse parameter payload is fixed and the live
address/index costs $O(\log m)$, so every cut has
$B_{\mathrm{cut}}=O(\log m)$ under the same serializer.

For the randomized extension, draw $U$ uniformly from $\mathbb Z_3^m$ and
$X$ uniformly from $\{0,1\}^m$, independently, and give Bob
$V=X-U\pmod3$.  Then $X$ and $V$ are independent and $H(X\mid V)=m$.
Alice's message $M$ consists of the committed prefix and the full restart
snapshot after processing $U$.  The snapshot includes the seed/tape and its
cursor, RAM, caches, old-input buffers, and IR, so Bob can resume the exact
randomized run upon receiving $V$.  The binary decoder of
\cref{lem:decoder} recovers $X$ with error at most $\delta$.  Conditional
Fano therefore gives
\begin{align*}
 H(X\mid M,V)&\leq h_2(\delta)+\delta m,\\
 I(X;M\mid V)&\geq(1-\delta)m-h_2(\delta).
\end{align*}
The two message fields are self-delimiting, so source coding lower-bounds
their expected total length by this conditional information, yielding
\eqref{eq:randomized-main-tradeoff}.  This is a conditional-Fano
(equivalently, distributional/Yao) argument and does not impose a fixed pass
count or sequential internal memory.
\end{proof}

\begin{proposition}[Memory-capped hybrid compilers]
\label{prop:hybrid-upper}
For every $\beta\in[0,1]$, let $h=\lfloor\beta m\rfloor$.  There is a
deterministic zero-error compiler for the masked-share stream, using the
ordered-share mode of $E_{\mathrm{rel}}$, such that at the inter-round cut
\begin{equation}\label{eq:hybrid-upper}
\begin{aligned}
 B_{\mathrm{cut}}^{\mathrm{MS}}
 &\leq h\log_2 3+O(\log m),\\
 B_{\mathrm{pre}}^{\mathrm{rel}}
 &\leq(m-h)\log_2 3+O(\log m).
\end{aligned}
\end{equation}
Thus \eqref{eq:main-tradeoff} is matched to leading order for every memory
fraction $\beta$ under this single family-relative encoding.  This statement
does not assert a matching curve for $B_{\mathrm{self}}$ or $G$.
\end{proposition}

\begin{proof}
The ordered-share header records $h$.  During round one the compiler packs
$u_1,\ldots,u_h$ into a base-three RAM block and immediately commits
$u_{h+1},\ldots,u_m$ in their public generator order.  It needs only the
packed block, an index, and framing state at the cut.  During round two it
emits the aggregates $u_j+v_j\pmod3$ for $j\leq h$ and the separate residues
$v_j$ for $j>h$.  The latter combine with the already emitted first shares;
all rotations commute, and \cref{prop:family-aggregates} gives the target up
to global phase.  Fixed-order packing gives the two bounds in
\eqref{eq:hybrid-upper}; no qubit address is hidden, because the ordered-share
mode declares its public schedule and split in its header.
\end{proof}

\begin{corollary}[Compactness versus flat-share recovery]
\label{cor:output-inflation}
On the masked-share representation, any compiler with
$B_{\mathrm{cut}}^{\mathrm{MS}}=o(m)$ must have
\begin{equation*}
 B_{\mathrm{rel}}^{\mathrm{out}}(\mathcal A,m\mid\mathcal F_m)
 \geq m\log_2 3-o(m),
\end{equation*}
whereas family-relative optimal output $m+O(\log m)$ requires at least the linear
cut budget of \eqref{eq:compact-cut-floor}.  This is a representation
tradeoff: the semantic aggregate stream attains both compact output and
$O(\log m)$ cut budget.
\end{corollary}

\begin{proof}
Use $B_{\mathrm{pre}}^{\mathrm{rel}}\leq
B_{\mathrm{rel}}^{\mathrm{out}}$ in \eqref{eq:main-tradeoff}, then compare
with the semantic construction in \cref{thm:main-tradeoff}.
\end{proof}

\begin{remark}[Three points under one encoding]
\label{rem:corners}
All points in \eqref{eq:main-tradeoff} use
$E_{\mathrm{rel}}(\cdot\mid\mathcal F_m)$.  The semantic pass-through has
compact relative output and logarithmic cut budget.  A flat compiler that
materializes the first round in a table or intermediate graph pays its
full serialized structure and parameters in $B_{\mathrm{cut}}$.  A
one-pass flat compiler may instead commit update residues before seeing
their mates, but those records enlarge $B_{\mathrm{pre}}^{\mathrm{rel}}$
and hence the final relative code.  The hybrids of
\cref{prop:hybrid-upper} interpolate between these endpoints and attain the
deterministic lower envelope to leading order for every $\beta$.  No claim
about $B_{\mathrm{self}}$ or $G(C)$ is inferred from this curve.
\end{remark}

\begin{remark}[Relation to the exact-semantics results and model scope]
\label{rem:three-channels}
The exact-semantics corollary
(\cref{cor:exact-semantics-corollary}) bounds final output length under
exactness and noncollision, via pumpability.  \Cref{thm:output-counting}
bounds two explicitly named output-description codes via counting.
\Cref{thm:main-tradeoff} is different: it uses the two-round completion
adversary and the complete bit-level restart snapshot
$B_{\mathrm{cut}}$.  A global ZX graph, phase-polynomial table, spilled
store, large-precision parameter, addressed window, RAM/DAG, cache, hash
table, dynamic pass schedule, or random seed is therefore an accounted
recovery channel, not an exception.  The deterministic theorem allows
arbitrary $p=p(m,I)$ in this charged model, and
\eqref{eq:randomized-main-tradeoff} allows input-dependent randomized
$p=p(m,I,\rho)$.  The balanced experimental schedule is not claimed to
instantiate this lower bound.
\end{remark}

\section{Unified $K$-ary Randomized Multipass Frontier}
\label{sec:w1-unified-frontier}

The preceding results isolate output information and the binary two-round
cut.  The following package aligns the packing modulus with a genuinely
dispersed additive-update representation and makes the number of passes
and the full serialized state cap explicit.  It also separates the proved
constant-factor upper envelope from the remaining model-dependent gaps.

\input{theory/qary_packing}
\input{theory/dispersed_update_tradeoff}
\input{theory/randomized_multipass_extension}
\input{theory/matching_upper_bounds}

\section{Representation-Aware Compiler}
\label{sec:compiler}

\subsection{Bounded Supported-Pattern Recognition}

\begin{definition}[Supported-pattern region recognizer]
\label{def:region-discovery}
The implemented recognizer accepts only configured high-level forms:
(i) explicit Pauli rotations or phase-gadget records that already carry
their support, (ii) local phase-ladder/Fourier spans, (iii) expanded
mirrored or self-inverse composite blocks, and (iv) repeated-run or preset
semantic nodes.  Within one supported contiguous span it maintains the
declared generators $P_1,\ldots,P_k$, coefficients
$\theta_1,\ldots,\theta_k$, qubit support, and source records.  It may extend
the span while the next record belongs to the supported grammar and passes
the configured commutation check; an unsupported or certified
noncommuting record closes the candidate.  Rejected recognition leaves the
original circuit available to the conservative fallback.
\end{definition}

\begin{remark}[Recognition, maximality, soundness, and complexity scope]
\label{rem:discovery-scope}
The artifact does not implement a general binary-symplectic extractor for
arbitrary circuit DAGs and does not claim a maximal commuting region in
that sense.  Its only maximality notion is the maximal contiguous span
accepted by the fixed supported-pattern grammar under the chosen scan
order.  Records for distinct generators remain separate until aggregation.
Soundness of a selected rewrite comes from the independent exact or
symbolic equivalence certificate of \cref{def:certificate}; failure returns
the unchanged candidate.  No general worst-case linear-time claim is made:
the runtime is the sum of the bounded scans, signature/hash bookkeeping,
certificate work, and any configured backend probes, as detailed in
Appendix~B.
\end{remark}

\subsection{Aggregation and Certificate}

\begin{definition}[Aggregation of an accepted supported region]
\label{def:aggregation}
Given an accepted supported region with generators $P_1,\ldots,P_k$ and coefficients
$\theta_1,\ldots,\theta_k$, the aggregation emits a representative
circuit implementing $\exp(-i\sum_j \theta_j P_j / 2)$ using at most
$c_0 + c_1 k$ gates in the target basis.
\end{definition}

\input{theory/certificate_soundness}

Within the supported exact domain, the executable audit applies
\CertificateMutationCases\ adversarial mutations---angle, sign, support bit,
qubit index/permutation, Clifford frame, deletion, and duplication---and
observes \CertificateFalseAccepts\ false accepts.  It also performs
\CertificateCrosscheckCases\ dense/symbolic cross-check cases for $n\leq6$
with no checker disagreement.  The status contract distinguishes
\texttt{completed\_valid}, \texttt{completed\_invalid},
\texttt{unsupported}, \texttt{predicate\_error}, and
\texttt{numerical\_inconclusive}; only independently certified
\texttt{completed\_valid} rows enter quality means.

\subsection{Constructive Aggregation Proposition}

\begin{proposition}[Constructive aggregation size]
\label{prop:constructive-aggregation}
For a supported semantic input accepted as a commuting diagonal phase layer
with $m$ independent phase terms, the semantic aggregation path emits an aggregate representative with at
most $c_0 + c_1 m$ phase records before final lowering, for constants
$c_0, c_1$ fixed by the representation and lowering conventions. Hence
the constructive representative has size $O(m)$ and is independent of the
repetition count $r$ at fixed width.
\end{proposition}

\begin{proof}
By hypothesis the accepted layer is a product of $m$ supported commuting
diagonal terms. The constructive
path keys terms by their diagonal generator and adds coefficients across
repetitions. Repetition changes the coefficient on each key, not the
number of keys. The aggregate contains at most one record per phase term,
plus fixed boundary and lowering overhead.
\end{proof}

\Cref{prop:constructive-aggregation} shows that a semantic-first pipeline
achieves an $O(m)$-record representative on the packing family.  Under the
public family dictionary it has
$B_{\mathrm{rel}}=m+O(\log m)$, matching the
family-relative inequality of \cref{thm:output-counting}; its separately
reported self-contained size is $O(m\log m)$.  On direct semantic input,
the support mask can be streamed with $B_{\mathrm{cut}}=O(\log m)$, which
is the constructive side of the separation in
\cref{thm:main-tradeoff}.

\section{Controlled Validation}
\label{sec:validation}

\subsection{Matched-Representation Benchmark}

\begin{definition}[Matched-representation benchmark]
\label{def:matched-benchmark}
For one sampled coefficient vector $x$, a matched cell fixes the target
$U_x$, target basis, hardware model, total error budget, timeout, memory cap,
and seed, and then evaluates every compiler on nine representations:
(1) semantic aggregate, (2) contiguous flat, (3) round-robin flat,
(4) random commuting order, (5) masked-share/update stream, (6) locally
folded flat, and (7--9) QASM, TKET, and ZX bridge records.  Each representation
has a unique identifier, serialized payload, generation seed, and unitary
certificate.  A cell is numerical only after the common certificate policy
returns \texttt{completed\_valid}; unsupported predicates and malformed or
inconclusive records remain status outcomes.
\end{definition}

The executable campaign uses the no-wrap target
\begin{equation*}
 U_x=\prod_{j=1}^{m}\exp\!\left(-\frac{i\pi x_j}{8K}Z_jZ_{j+1}\right),
 \qquad x_j\in\{0,\ldots,K-1\}.
\end{equation*}
The active support is sampled at the requested density and each active
coefficient is uniform on $\{1,\ldots,K-1\}$.  The frozen campaign contains
\MatchedCells\ matched cells across the nine representations and five
compilers.  Of these, \MatchedCompleted\ are
\texttt{completed\_valid}; the remaining \MatchedPredicateErrors\ are retained
as \texttt{predicate\_error}, all from the Qiskit masked-update path, and are
not converted into numerical wins or losses.  Every raw row records the
generation parameters, sampled $x$, realized density, schedule, basis,
timeout, memory cap, runtime, peak RSS, gates, depth, CX, serialized input and
output/IR bytes, bridge sequence, and certificate kind.

\begin{figure*}[t]
\centering
\includegraphics[width=0.94\textwidth]{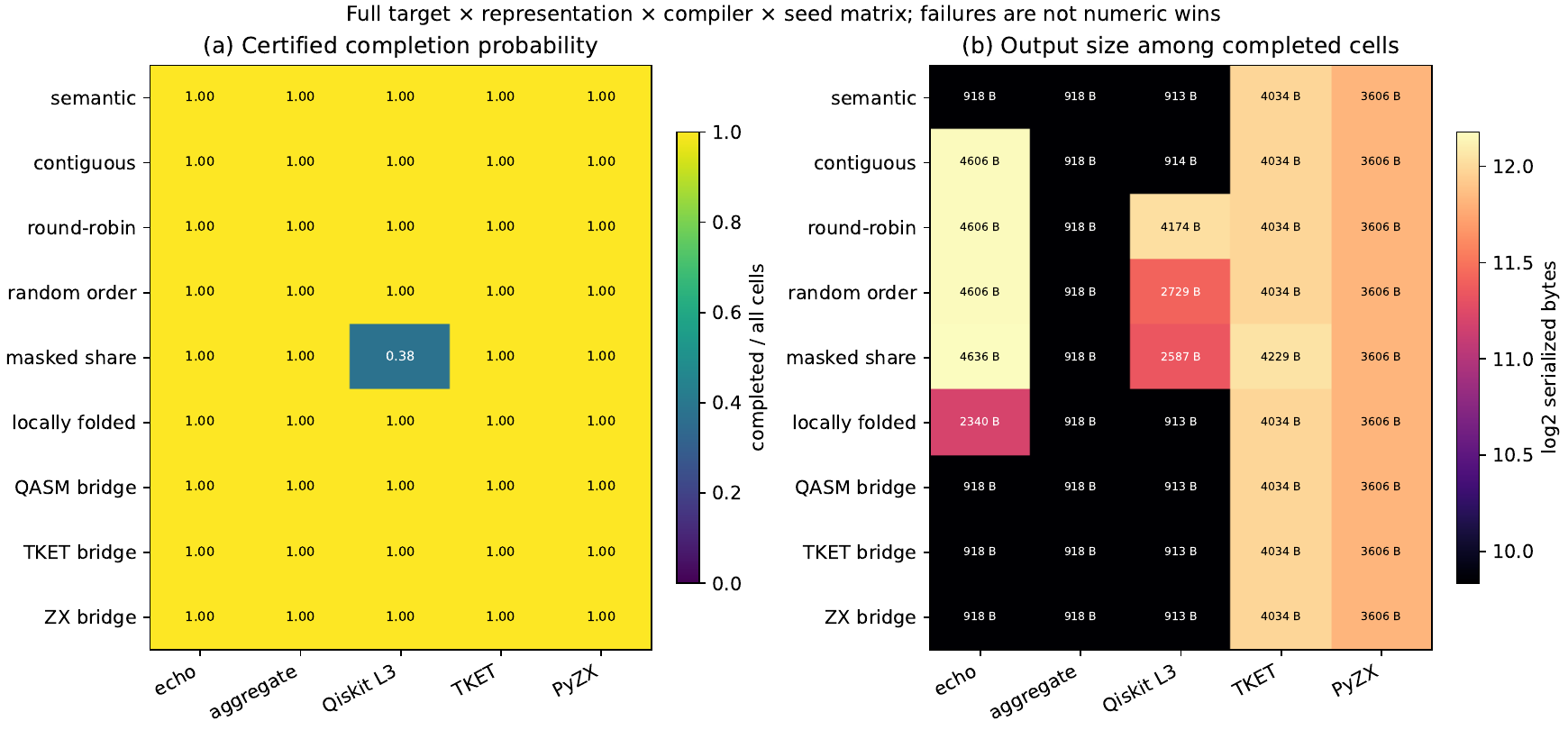}
\caption{Complete matched-representation interaction summary.  Panel (a)
reports certified completion probability for every representation--compiler
pair; panel (b) reports serialized output size only among completed cells.
Predicate errors remain status observations and never become numerical wins.}
\label{fig:matched-interaction}
\end{figure*}

The representation main effect, compiler main effect, their interaction,
completion probability, and scaling exponents are computed from the complete
matrix.  The interaction is substantive: a compiler that reconstructs a
phase polynomial or Pauli network before synthesis belongs on the
semantic-capable side even when its input arrived through a flat bridge.

\subsection{Semantic-Capable External Baselines and Bridge Loss}
\label{sec:external-baselines}

A separate \ExternalCells-cell campaign freezes recommended-native and
unified-target-basis configurations for the authors' exact phase-table
reference, local Qiskit optimization, TKET PauliSimp, and PyZX full-reduce.
It records tool version, available commit metadata, pass sequence, bridge,
seed, basis, timeout, memory cap, pre/post-bridge gate/depth/CX and serialized
IR, runtime, RSS, and certificate status.  \ExternalCompleted\ cells are
\texttt{completed\_valid}; \ExternalPredicateErrors\ are retained as
\texttt{predicate\_error}.  TKET and PyZX are classified as
semantic-capable, not as local baselines, whenever their configured global
IR recovers the compact structure.  Bridge loss is reported separately from
compiler quality, and unsupported is never conflated with incorrect.

The stance of this campaign is explanatory, not competitive.  TKET
PauliSimp and PyZX full-reduce already implement the semantic-first
principle on their global IRs; the contribution of this paper is the
bound that says what any such tool must pay to do so on a dispersed
stream (\cref{cor:physical-lower-bounds}).  A preregistered scan tested
this on a third-party compiler the authors do not control: with
predictions and decision rule committed before execution
(\path{PaperDraft/research/E7_preregistration.md}), TKET PauliSimp was
run on the masked-share family at $m\in\{8,16,32,64,128\}$, three seeds
per size, in the byte-identical pass sequence of this campaign.  It
achieved compact output in every cell ($508$ dispersed rotations
aggregated to $127$ non-Clifford $R_z$ at $m=128$), and its measured
pass-phase memory and serialized IR grew with log--log slope
$\approx1.7$ in $m$---at least the linear growth that the
$\Omega(m\log K_\epsilon)$ floor requires, and indeed more---while the
instrumented semantic-input control stayed flat ($196\to199$ bytes over
the same $16\times$ range of $m$).  The preregistered point prediction
of \emph{exactly} linear growth was thereby rejected under its own
frozen decision rule, and we report it as rejected; the one-sided lower
bound is consistent with, and slack against, the superlinear
measurement.  The frozen artifact is
\path{data/frozen/e7_tket_state_scaling.json}.  Two limits on what this
scan is allowed to establish should be stated explicitly.  First, the
quantity measured at a third-party API boundary is a pass-phase memory
and serialized-IR snapshot, which is a configured-pipeline diagnostic in
the sense of \cref{tab:evidence-classes}: it is not promoted to a
measurement of $B_{\mathrm{cut}}$ and is not commensurable with the $S$
of the frontier.  The scan therefore tests only the qualitative
dichotomy the theorem predicts---dispersed input forces growth in $m$,
semantic input does not---and tests neither the constant nor the
exponent.  Second, superlinear growth means the floor is not the binding
constraint on this tool: the bound establishes the floor, the observed
cost exceeds it, and the excess is implementation structure (the
pairwise commutation structure of a Pauli graph is itself
superlinear in the number of terms).  No claim that
the authors' implementation outperforms these tools is made or implied.

\begin{lemma}[Dimension-preserving instrumentation]
\label{lem:metric-conversion}
For every accepted version-one artifact $A$, its reported size is exactly
$8\operatorname{stat}(A).\mathtt{st\_size}$ bits.  Address growth is charged by
the literal varint bytes; in particular $127$ and $128$ have different address
lengths.  There is no representation-independent conversion from gate count,
depth, \texttt{cx}, token count, or IR-node count to bits.
\end{lemma}

\begin{proof}
The first statement is \cref{prop:file-backed-accounting}.  Gate or node count
does not determine parameter numerators, denominators, qubit-address widths,
edge endpoints, delimiters, or textual external payloads.  Thus equal-count
objects can have different artifact lengths.  Conversely, family-relative
dictionary codes can assign different self-contained circuits equal
conditional length.  Only the named artifact file supplies a bit value.
\end{proof}

\begin{table*}[!t]
\centering
\scriptsize
\setlength{\tabcolsep}{3pt}
\begin{tabularx}{\textwidth}{@{}l l X X@{}}
\toprule
Dimension & Recorded field & Executable serializer or definition & Permitted comparison \\
\midrule
bits & \texttt{serialized\_input\_bits} & verified \texttt{SemanticStream}, \texttt{FlatUpdateStream}, or external-input artifact file & same codec/mode only \\
bits & \texttt{B\_com\_bits} & verified literal committed-output artifact/prefix file & same codec only \\
bits & five \texttt{B\_*\_bits} fields & verified control/store/window/parameter/IR files & corresponding formal field only \\
records/tokens & source records & exact count before serialization & never converted by a constant \\
IR nodes & \texttt{ir\_nodes} & structural node count reported separately from \texttt{B\_ir\_bits} & no implicit bit or gate conversion \\
gates/depth/CX & structural output & Qiskit circuit operations, DAG depth, and \texttt{cx} count & like-for-like target basis \\
\bottomrule
\end{tabularx}
\caption{Instrumentation contract.  Figures and tables compare within rows of
one dimension.  Cross-dimensional discussion cites the explicit codec or is
stated only qualitatively.}
\label{tab:instrumentation-contract}
\end{table*}

\begin{table*}[!t]
\centering
\small
\setlength{\tabcolsep}{4pt}
\begin{tabularx}{\textwidth}{@{}l X X@{}}
\toprule
Evidence class & Included experiments & Authorized inference \\
\midrule
instrumented theorem-model implementation & W3 traced reference compilers and restart verifier & directly tests the serialized commitment--cut coordinates and reports theorem residuals \\
configured-pipeline diagnostic & matched matrix, external semantic baselines, exact-semantics witnesses, and natural workloads & compares observed tool behavior; never proves a lower bound \\
model-specific consequence & fixed-total-error synthesis and surface-code calculation & reports consequences under the disclosed QRE; not a hardware-independent forecast \\
\bottomrule
\end{tabularx}
\caption{Formal/evidentiary separation used throughout the paper.  The first
row is the instrumented reference implementation of the declared resource
model; configured external tools are not retroactively placed in that class.}
\label{tab:evidence-classes}
\end{table*}

\subsection{Packing-Family Scaling}
\label{sec:packing-scaling}

The complete multifactor campaign of \cref{fig:matched-interaction}
is the primary matched-representation experiment.  The earlier 5-by-5 sweep
below is retained as a separately keyed legacy diagnostic because it reaches
the same maximum width under a 600\,s protocol and preserves its historical
\texttt{predicate\_error} outcomes; it is not substituted for the new matrix.
This legacy experiment links the growing-$m$ theorems to configured-pipeline
behavior. Unlike the fixed-width witnesses below, which exercise the
exact-semantics corollary, the packing family grows the number of
independent phase terms $m$ at fixed tolerance.

\paragraph{Instance.}
We instantiate \cref{def:packing-construction} with the chain of two-body
characters $a_j = e_j + e_{j+1}$ on $n = m+1$ qubits, which are
$\Ftwo$-independent for every $m$ (certified by rerunning the exact rank
verifier), using the separate experimental setting $K=2$, $\Delta=\pi/4$.
Each token is lowered as
$RZ_i(\theta)\,RZ_j(\theta)\,CP_{ij}(-2\theta)$ under the convention
used here. The balanced round-robin encoding
(\cref{def:streaming-encoding}) with dispersal $r = 4$ interleaves the
terms across shared wires, so per-wire locality does not trivially reveal
the aggregate.  Because these balanced tokens contain $x_j/r$, this is a
pipeline diagnostic rather than an instance of the masked-share lower
bound. The sweep is $m \in \{4, 8, 16, 32, 64\}$ under the
standard 600\,s protocol of this section.  The exact rational-$\pi$ sidecar
records $\Delta/\pi=1/4$ and per-token angle $1/16$.  Correctness uses the
dense certificate for $n\leq6$ and the sound exact symbolic checker of
\cref{def:certificate,prop:symbolic-soundness} for $n>6$; any output outside
$\mathcal D_{\rm CPC}$ is status-only.

\paragraph{Results.}
\Cref{tab:packing-scaling} reports completed output gates as a function of $m$
for semantic UCC, \texttt{qiskit opt3}, rebased TKET PauliSimp, staq
rotation folding, and the phase-polynomial reference.
\pkResultsSummary

\begin{table*}[!t]
\centering
\small
\begin{tabular}{cccccc}
\toprule
$m$ & semantic UCC & \texttt{qiskit opt3} & TKET PauliSimp (rebased) & staq folding & phase-poly ref \\
\midrule
4  & $25/21/8$ & $68/38/32$ & $54/20/10$ & $40/27/32$ & $25/21/8$ \\
8  & $49/41/16$ & $136/54/64$ & \texttt{predicate\_error} & $80/39/64$ & $49/41/16$ \\
16 & $97/81/32$ & $272/86/128$ & \texttt{predicate\_error} & $160/63/128$ & $97/81/32$ \\
32 & $193/161/64$ & $544/150/256$ & \texttt{predicate\_error} & $320/111/256$ & $193/161/64$ \\
64 & $385/321/128$ & $1088/278/512$ & \texttt{predicate\_error} & $640/207/512$ & $385/321/128$ \\
\bottomrule
\end{tabular}
\caption{Packing-family scaling at growing $m$ (chain two-body
characters, $n = m+1$, dispersal $r = 4$, 600\,s budget). Completed cells
report output gates/depth/\texttt{cx}; status outcomes are recorded separately and are not
scored as structural wins. No external row is claimed to be a member of
the transducer class of \cref{def:compiler}.}
\label{tab:packing-scaling}
\end{table*}

The earlier two-panel figure is deliberately omitted.  Its gate panel is
redundant with \cref{tab:packing-scaling}, while its family-relative bit panel
was generated from a handwritten length formula and has no serialized
per-point files.  Under the present accounting contract such a curve may be
restored only by emitting the exact conditional-code artifacts (together with
the family-dictionary digest) and measuring their files.

\paragraph{Scope.}
The external-toolchain rows are configured-pipeline diagnostics in the
sense of this section: they neither prove nor test
\cref{thm:main-tradeoff}, whose serialized inter-round cut budget is not
observable from output gate counts, and no bit curve is inferred from these
gate-count rows.  The separate reference-compiler implementation below
measures theorem-model cut fields directly; it is not inferred from output
gates.

\subsection{Event-Level Commitment--Cut Tracing}
\label{sec:cut-tracing}

The instrumented suite contains echo/pass-through, full aggregation, 25\%,
50\%, and 75\% memory-capped hybrids, and an IR-materializing control.  Each
compiler follows the formal flat-update interface and records an event after
every input update and commitment.  The campaign scans $m,r,K,\epsilon,p,S$,
schedule, and seed, producing \TradeoffRuns\ complete run summaries plus the
compressed raw event log and content-addressed trace archive.  At each named
cut it records input position, pass, control/store/window/parameter and IR
bytes, committed-output bytes, current output length, RSS, and wall time.
The literal measured definitions are
\begin{equation*}
\begin{aligned}
 B_{\rm com}(c)&=8|A_{\rm com,c}|_{\rm byte},\\
 B_{\rm cross}(c)&=8\sum_{f\neq\mathrm{com}}|A_{f,c}|_{\rm byte}.
\end{aligned}
\end{equation*}
RSS is recorded as an implementation metric and is not substituted for the
serialized state cap $S$.  The analyzer recomputes the theorem residual from
the trace coordinates and retains every anomaly or failed run.  All frozen
points lie in the theorem-allowed region; no point was removed as an outlier.

\begin{figure*}[t]
\centering
\includegraphics[width=0.94\textwidth]{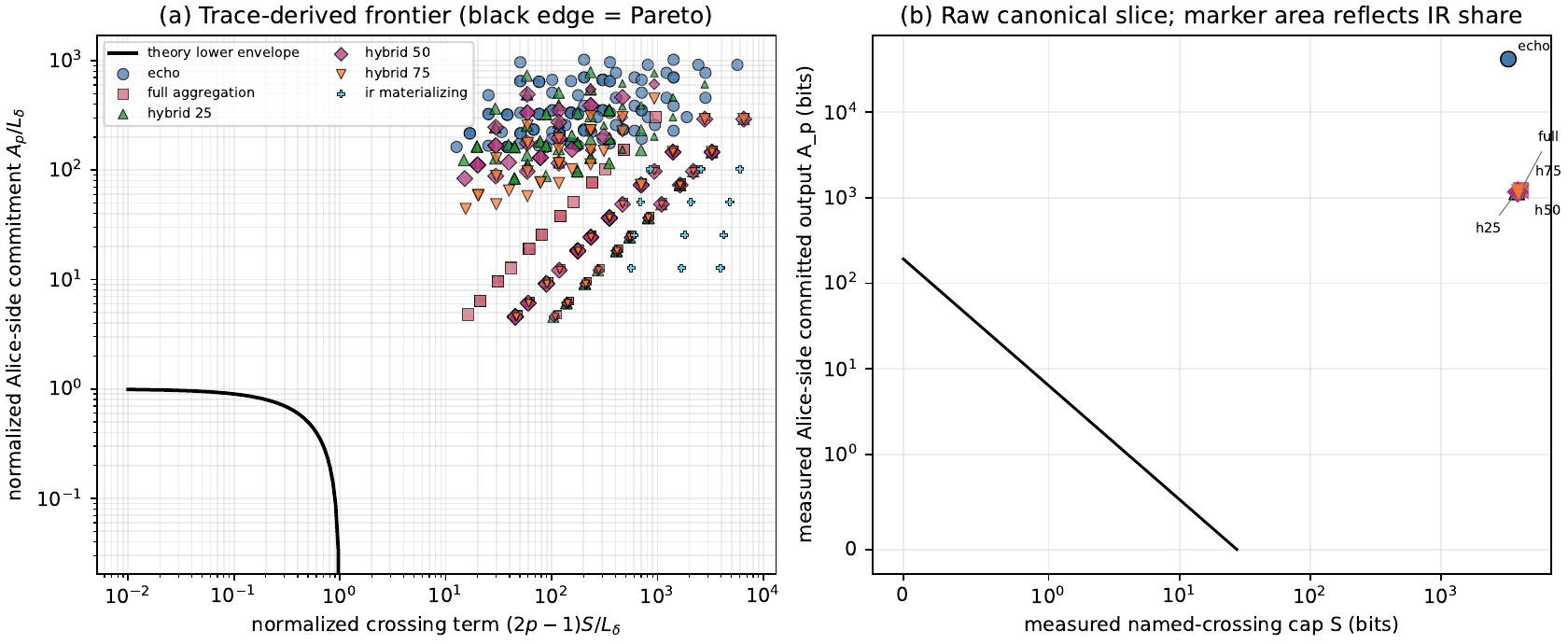}
\caption{Trace-derived commitment--crossing-state--IR surface.  The black
edge marks measured Pareto points and the theoretical lower envelope is
evaluated in the same serialized-bit coordinates.  Marker area in the raw
slice reflects IR share; RSS is recorded separately.}
\label{fig:measured-tradeoff}
\end{figure*}

Echo commits updates early with little retained state; full aggregation
occupies the state-heavy compact-output corner; and the three capped hybrids
interpolate through explicit block commitments.  Their scaling fits and
residuals are reported with the raw points, making this a direct measurement
of the theorem resources rather than an inference from output gate count.

\subsection{Width-Axis Scaling}

To check the constructive aggregation prediction
(\cref{prop:constructive-aggregation}), we run a width-axis sweep for the
\texttt{full\_pair\_cp} topology at $n = 4, 5, 6$.

\begin{table*}[htbp]
\centering
\small
\begin{tabular}{cccc}
\toprule
Width ($n$) & CP pairs ($m_{\rm CP}$) & Output Gates & CX Count \\
\midrule
4 & 6  & 42 & 12 \\
5 & 10 & 65 & 20 \\
6 & 15 & 93 & 30 \\
\bottomrule
\end{tabular}
\caption{Width-axis scaling for the semantic Fourier-layer witness. For
each fixed width $n$, the semantic output is constant across all tested
repetition counts $r$. Across widths, the output scales as $O(m_{\rm CP})$
where $m_{\rm CP} = n(n-1)/2$.}
\label{tab:width-axis-scaling}
\end{table*}

The output is $O(m_{\rm CP})$ and independent of $r$ on this finite
width sweep, consistent with the constructive proposition.

\subsection{Exact-Semantics Witness Results}

These experiments exercise the exact-semantics corollary
(\cref{cor:exact-semantics-corollary}), not the main resource tradeoff
(\cref{thm:main-tradeoff}). They verify the constructive
aggregation path on the tested fixed-width witnesses and record how configured
external pipelines behave predictably across the representation
boundary. Full data tables are in \cref{app:supplementary-experiments}.

The rational CP $H D^r H$ witness is the canonical fixed-width family
with a constant 42-gate output.  \Cref{fig:fourier-separation} shows that
semantic UCC remains at 42 gates while rebased TKET PauliSimp returns
52--73 gates across five scales; materialize-first baselines grow
linearly. The irrational independent-Pauli witness
(\cref{fig:unbounded-witness-scaling}) shows the same pattern at
27-gate constant output. Key observations:
\begin{itemize}[leftmargin=1.5em]
    \item Semantic UCC outputs a constant gate count (42 for rational,
    27 for irrational) at all five tested sizes.
    \item Rebased TKET PauliSimp recovers a small validated form at all
    five scales on both witnesses, demonstrating that a generic pipeline
    can cross the formal boundary by reconstructing an aggregate Pauli
    representation.
    \item Public staq rotation folding, \texttt{qiskit opt3}, and TKET
    GuidedPauliSimp grow linearly with input size.
    \item PyZX \texttt{full\_reduce} completes at small sizes but reaches
    time or correctness limits at larger scales.
\end{itemize}

\begin{figure*}[t]
\centering
\includegraphics[width=0.72\textwidth]{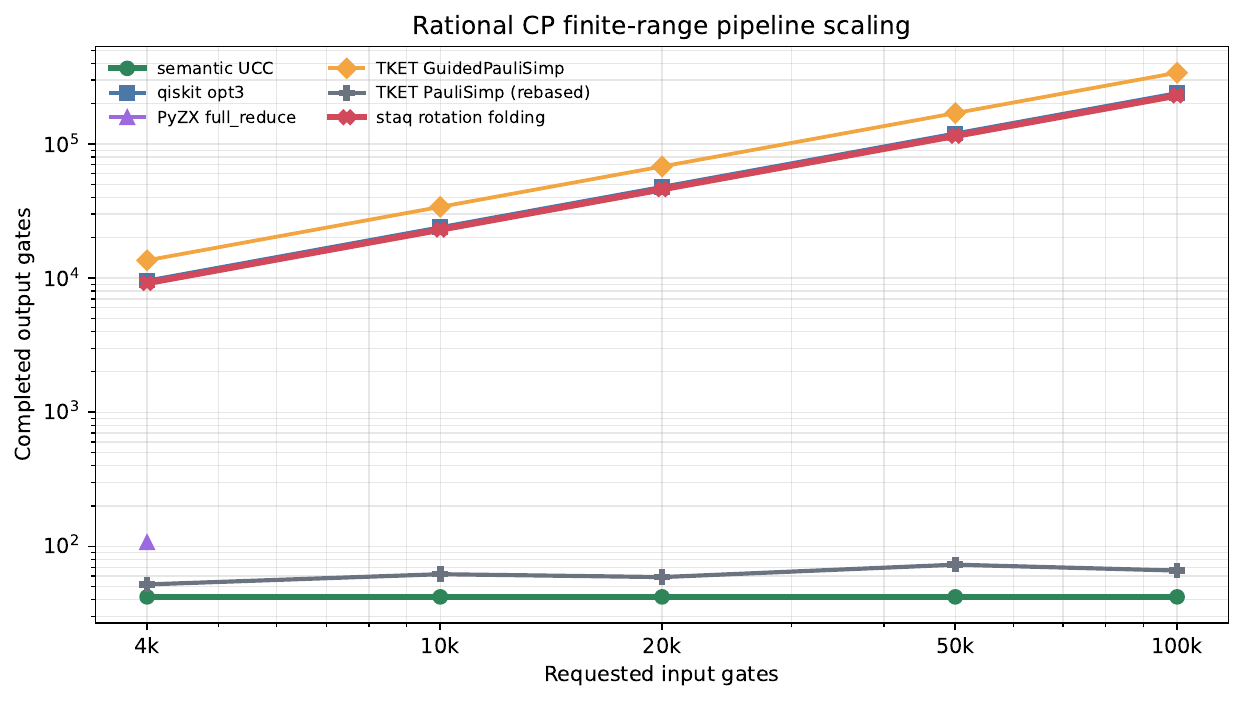}
\caption{Fixed-width non-inverse Fourier-layer scaling on the rational
$H D^r H$ witness. Semantic UCC remains at 42 gates and rebased TKET
PauliSimp returns 52--73 gates across the five scales. Public staq
rotation folding, \texttt{qiskit opt3}, and TKET GuidedPauliSimp grow
with input size; PyZX \texttt{full\_reduce} completes at 4k and reaches
the time budget at larger scales. Status-only cells are omitted from
numerical curves.}
\label{fig:fourier-separation}
\end{figure*}

\begin{figure*}[t]
\centering
\includegraphics[width=0.72\textwidth]{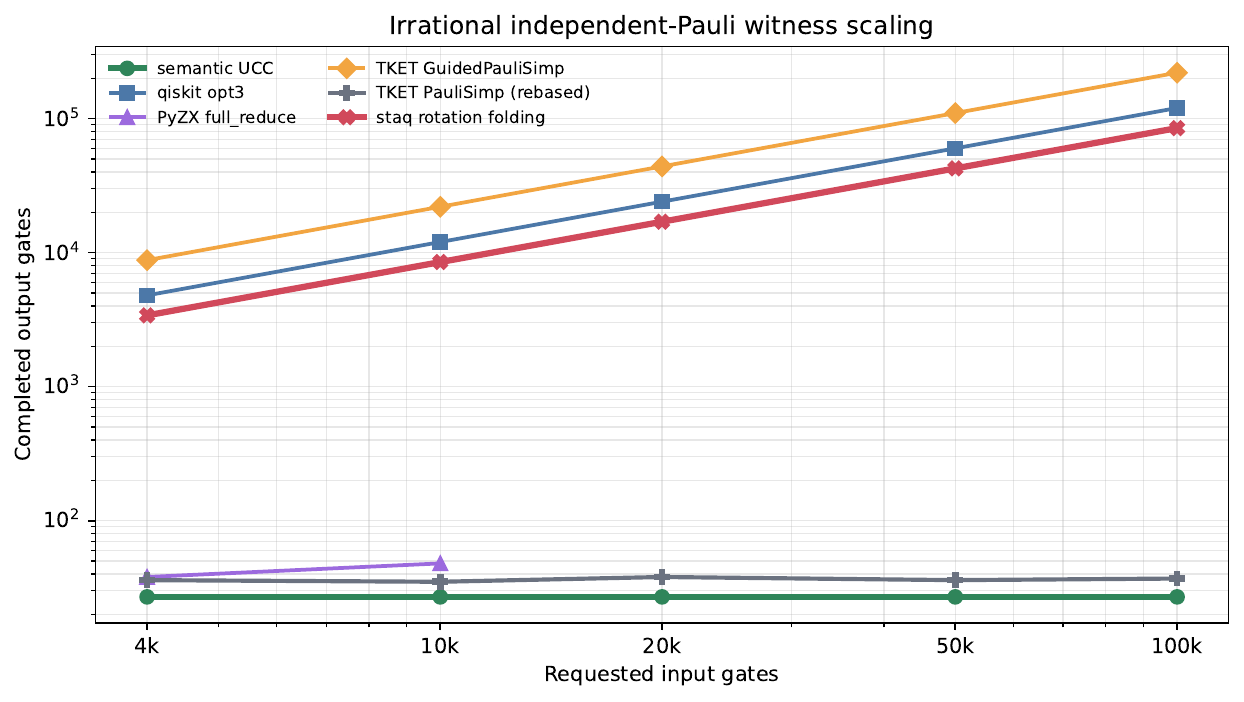}
\caption{Configured-pipeline scaling on the irrational independent-Pauli
witness. Semantic UCC remains at 27 gates. Rebased TKET PauliSimp
returns a small validated form at all five scales, demonstrating that a
generic pipeline can cross the formal boundary by reconstructing an
aggregate Pauli representation. Public staq rotation folding,
\texttt{qiskit opt3}, and TKET GuidedPauliSimp grow with input size.
PyZX \texttt{full\_reduce} completes at 4k and 10k, fails the numerical
tolerance at 20k, and reaches the time budget at larger scales.
Status-only cells are omitted from numerical curves.}
\label{fig:unbounded-witness-scaling}
\end{figure*}

These measurements exercise the constructive side of
\cref{cor:exact-semantics-corollary} and disclose the behavior of
configured tools; they do not prove the resource bound for any real
compiler. The rebased PauliSimp result is the classification's positive
prediction in action: once a generic pass sequence reconstructs a compact
Pauli representation, the output need not grow with the flat input.
Conversely, staq's growing output and PyZX's status outcomes are not
impossibility results---they are diagnostics of configured-pipeline
behavior on this specific witness family.

\subsection{Mechanism Ablation}

\begin{table*}[tbp]
\centering
\small
\setlength{\tabcolsep}{4pt}
\begin{tabularx}{\textwidth}{@{}r *{4}{>{\raggedright\arraybackslash}X}@{}}
\toprule
Requested & semantic UCC & artifact UCC (Fourier disabled) & \texttt{qiskit} opt3 & phase-poly ref \\
\midrule
4{,}000 & 42 / 2.9 s & 9{,}588 / 18.2 s & 9{,}588 / 1.0 s & 42 / 0.2 s \\
10{,}000 & 42 / 8.7 s & 23{,}988 / 94.2 s & 23{,}988 / 5.1 s & 42 / 0.3 s \\
20{,}000 & 42 / 28.5 s & 47{,}988 / 339.3 s & 47{,}988 / 19.5 s & 42 / 0.3 s \\
50{,}000 & 42 / 20.7 s & timeout & 119{,}988 / 120.9 s & 42 / 0.5 s \\
100{,}000 & 42 / 28.7 s & timeout & 239{,}988 / 485.8 s & 42 / 0.8 s \\
\bottomrule
\end{tabularx}
\caption{Fourier-layer semantic-path ablation on the fixed-width
$H D^r H$ witness under the 600\,s budget. Disabling the Fourier-layer
IR completes at 4k--20k but materializes large outputs; 50k--100k exceed
the budget. Fourier-enabled semantic UCC and the phase-polynomial
reference return 42 gates at all sizes.}
\label{tab:fourier-ablation}
\end{table*}

\begin{figure*}[t]
\centering
\includegraphics[width=0.72\textwidth]{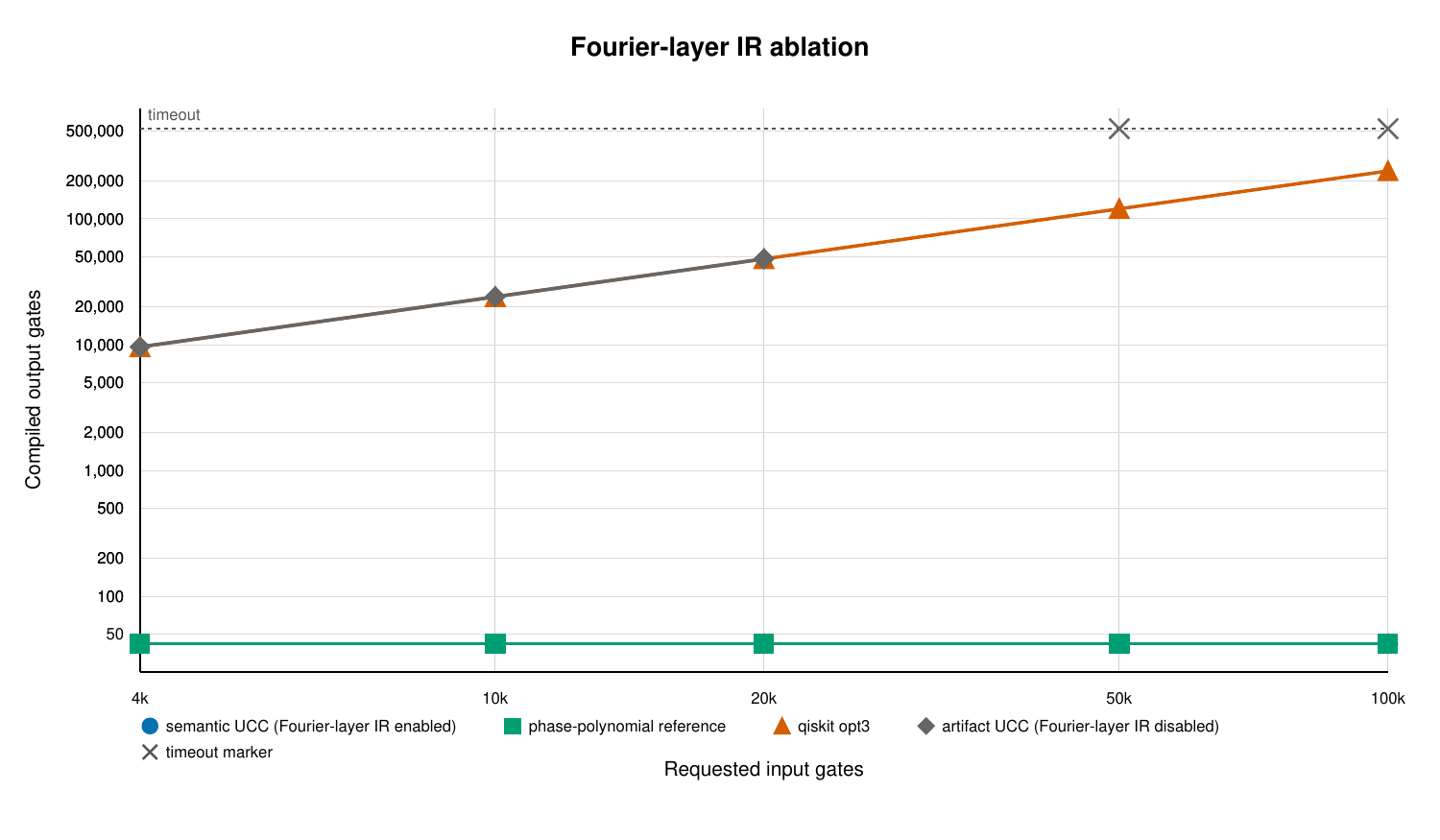}
\caption{Mechanism ablation for the fixed-width $H D^r H$ witness.
Disabling the Fourier-layer semantic path yields completed but linearly
large outputs through 20k and timeout rows at 50k and 100k, while
Fourier-enabled semantic UCC and the phase-polynomial reference recover
the canonical 42-gate circuit.}
\label{fig:fourier-ablation}
\end{figure*}

The ablation strengthens the mechanism interpretation. With Fourier-layer
IR disabled, the compiler completes at 4k, 10k, and 20k but materializes
the same large non-aggregated outputs as \texttt{qiskit opt3}; at 50k
and 100k the disabled-IR rows exceed the 600\,s budget
(\cref{fig:fourier-ablation}). In contrast, the Fourier-enabled semantic
UCC path and the phase-polynomial reference recover the same 42-gate
canonical circuit at every tested size. Thus the constant form is tied
to semantic Fourier-layer lifting and coefficient aggregation, not
merely to a favorable ordering of flat post-lowering passes.

\subsection{Fixed-Total-Error End-to-End Synthesis and QRE}

\begin{table*}[t]
\centering
\scriptsize
\begin{tabular}{l r r r r r r}
\toprule
pipeline & $N_{\rm rot}$ & $T$ & physical qubits & factories & cycles & qubit-s \\
\midrule
\input{generated/qre_main_slice}
\end{tabular}
\caption{Predeclared 20k representative slice of the fixed-total-error
campaign: $\epsilon_{\rm total}=10^{-6}$, equal synthesis allocation,
conservative $p_{\rm phys}=10^{-3}$ scenario, and 16 factories.  Each compiled
output has an independent certificate, and every unique-angle staq synthesis
returns \texttt{Check flag = 1}.  The frozen full table additionally records
code distance, data and factory qubits separately, and modeled quantum runtime.}
\label{tab:resource-consequence}
\end{table*}

The new campaign scans
$\epsilon_{\mathrm{total}}\in\{10^{-3},10^{-6},10^{-9}\}$ and assigns the
same additive budget fractions to every pipeline:
$0.10\epsilon_{\rm total}$ for the certified algorithmic stage,
$0.45\epsilon_{\rm total}$ for rotation synthesis, and
$0.45\epsilon_{\rm total}$ for logical failure.  For equal allocation and
$N_{\rm rot}$ output rotations, it chooses
$p=\lceil\log_{10}(N_{\rm rot}/\epsilon_{\mathrm{synth}})\rceil$ and
$\epsilon_i=10^{-p}$, hence
$\sum_i\epsilon_i\leq\epsilon_{\mathrm{synth}}$.  A second measured-cost
greedy allocation spends only this decimal slack; it is an upper construction,
not an optimality claim.  Each IEEE-serialized angle is actually synthesized
by the frozen deterministic GMP build of staq
\texttt{grid\_synth}~\cite{amy2020staq,ross2016optimal}.  Across
670 checked angle--precision pairs, the reported error never exceeds the
allocated radius; exact multiples of $\pi/4$ use staq's checked exact branch.

The full $2\times3\times3\times2\times2\times4=\QRECells$ matrix crosses two target
sizes, three pipelines, three total errors, two allocation policies, two QEC
profiles, and 1/4/16/32 factories.  All \QRECells\ rows are
\texttt{completed\_valid}, carry a content-addressed exact-symbolic or dense
$n\leq6$ certificate, and satisfy the requested total-error bound.  The two
surface-code profiles~\cite{fowler2012surface,horsman2012lattice,litinski2019game}
use $p_{\rm phys}=10^{-3}$ and $10^{-4}$, respectively,
with separately frozen data/factory coefficients, factory
periods~\cite{litinski2019magic}, and cycle
times.  They choose the least odd $d$ meeting the disclosed union bound.  These
are executable sensitivity scenarios rather than vendor estimates, and
follow the accounting conventions of published end-to-end
estimates~\cite{gidney2021factor,beverland2022assessing}.

\begin{figure*}[t]
\centering
\includegraphics[width=0.82\textwidth]{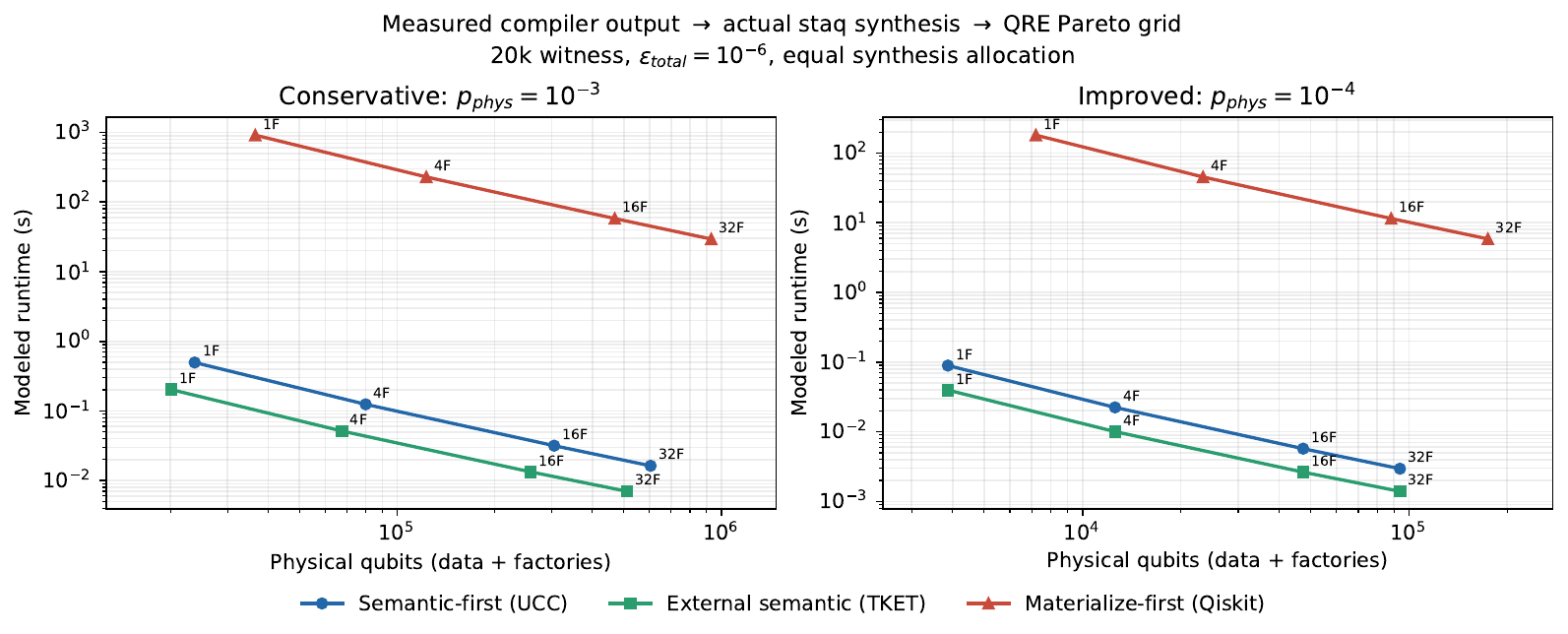}
\caption{Physical-qubit/runtime Pareto grid at the representative 20k,
$\epsilon_{\rm total}=10^{-6}$ slice.  Labels give factory counts.  The
external TKET path receives a materialized CX--$R_Z$--H bridge but reconstructs
a compact phase form before synthesis, so it is classified as
semantic-capable.}
\label{fig:qre-pareto}
\end{figure*}

The physically relevant point is the ordering of semantic aggregation relative to
synthesis.  In a fault-tolerant pipeline, arbitrary rotations are
eventually approximated by discrete resources such as Clifford+$T$
sequences. If the compiler first materializes $r$ separate rotations
and only later synthesizes them, the resource estimator sees $r$
independent approximation tasks. A resource-rich pipeline instead
combines the coefficients first and presents one aggregated rotation per
phase term. An FT toolchain that reconstructs the phase polynomial
before synthesis collapses this multiplicity itself and should be
classified on the semantic-capable side of the comparison
(see \cref{sec:ft-consequences} for the model-specific resource discussion).

\section{Application Workloads}
\label{sec:workloads}

\subsection{Workload Specification}

The application campaign contains seven structured families: Ising/MaxCut
cost evolution~\cite{farhi2014qaoa}, diagonal Hamiltonian simulation,
commuting-block Trotterization~\cite{childs2021trotter}, QFT arithmetic,
QPE controlled powers, controlled-power
ladders, and phase-polynomial blocks~\cite{amy2019cnot}.  Family sizes and
naming follow the low-level benchmark
conventions of~\cite{li2023qasmbench,quetschlich2023mqt}.  Random
Clifford+$R_Z$ and
routing-dominated~\cite{cowtan2019routing,li2019sabre} circuits are
negative controls.  The separate controlled
packing family
(\cref{def:packing-construction}) with growing $m$, fixed $\epsilon$,
is reported in \cref{sec:packing-scaling}.

\subsection{Strong Baseline Comparison}

The canonical comparison includes Qiskit
\texttt{optimization\_level=3}~\cite{javadi2024qiskit}, TKET
PauliSimp~\cite{sivarajah2021tket,cowtan2020phase}, PyZX
\texttt{full\_reduce}~\cite{kissinger2020pyzx}, and the authors' exact
phase-table reference.  Each external tool has a recommended-native and a
unified $\{\mathrm{CX},R_Z,H\}$ configuration, with pass sequence, bridge,
version, seed, timeout, memory cap, and certificate frozen.  TKET supplies the
Pauli-network/global-semantic class and PyZX the graph/global-IR class.  The
complete W4 interaction is \cref{fig:matched-interaction}; the independent W5
campaign and bridge-loss accounting are described in
\cref{sec:external-baselines}.  Historical staq, GuidedPauliSimp, and
phase-polynomial diagnostic rows remain in their original immutable campaigns
and are not silently mixed into the canonical matrix.  This comparison is
diagnostic, not a contest: the external tools are treated as independent
implementations of the semantic-first principle whose internal costs the
bounds of \cref{sec:physical-dispersed,sec:w1-unified-frontier} explain,
gate-count parity is reported as parity, and no headline ranks the
authors' pipeline above them.

\subsection{Real-Task Panel}

\begin{table*}[htbp]
\centering
\scriptsize
\begin{tabular}{lccccc}
\toprule
Instance & Method & Gates & Depth & CX & Runtime \\
\midrule
\multirow{4}{*}{\texttt{REAL24-phase-est.}}
  & translation only & 159{,}838 & 121{,}214 & 67{,}608 & 0.080 s \\
  & qiskit opt3 & 166{,}805 & 119{,}204 & 58{,}491 & 1.298 s \\
  & upstream UCC v0.4.12 default & 471{,}819 & 328{,}271 & 58{,}491 & 80.923 s \\
  & semantic UCC & 166{,}805 & 119{,}204 & 58{,}491 & 1.364 s \\
\midrule
\multirow{4}{*}{\texttt{REAL24-grover}}
  & translation only & 85{,}403 & 55{,}892 & 33{,}408 & 0.104 s \\
  & qiskit opt3 & 79{,}029 & 54{,}163 & 32{,}544 & 0.552 s \\
  & upstream UCC v0.4.12 default & 287{,}047 & 152{,}298 & 32{,}544 & 1.915 s \\
  & semantic UCC & 79{,}029 & 54{,}163 & 32{,}544 & 0.556 s \\
\midrule
\multirow{4}{*}{\texttt{REAL24-qaoa}}
  & translation only & 36{,}512 & 2{,}416 & 23{,}808 & 0.022 s \\
  & qiskit opt3 & 36{,}512 & 2{,}416 & 23{,}808 & 0.272 s \\
  & upstream UCC v0.4.12 default & 167{,}833 & 7{,}133 & 23{,}808 & 1.010 s \\
  & semantic UCC & 36{,}512 & 2{,}416 & 23{,}808 & 0.273 s \\
\bottomrule
\end{tabular}
\caption{Official real algorithm-family instances. The upstream UCC
v0.4.12 row is the frozen pre-semantic default defined in the provenance
map of Appendix~B. Parity with \texttt{qiskit opt3} reflects conservative
portfolio selection; this panel is an anti-regression check, not evidence
of semantic recovery.}
\label{tab:real-instances}
\end{table*}

\begin{figure*}[t]
\centering
\includegraphics[width=0.82\textwidth]{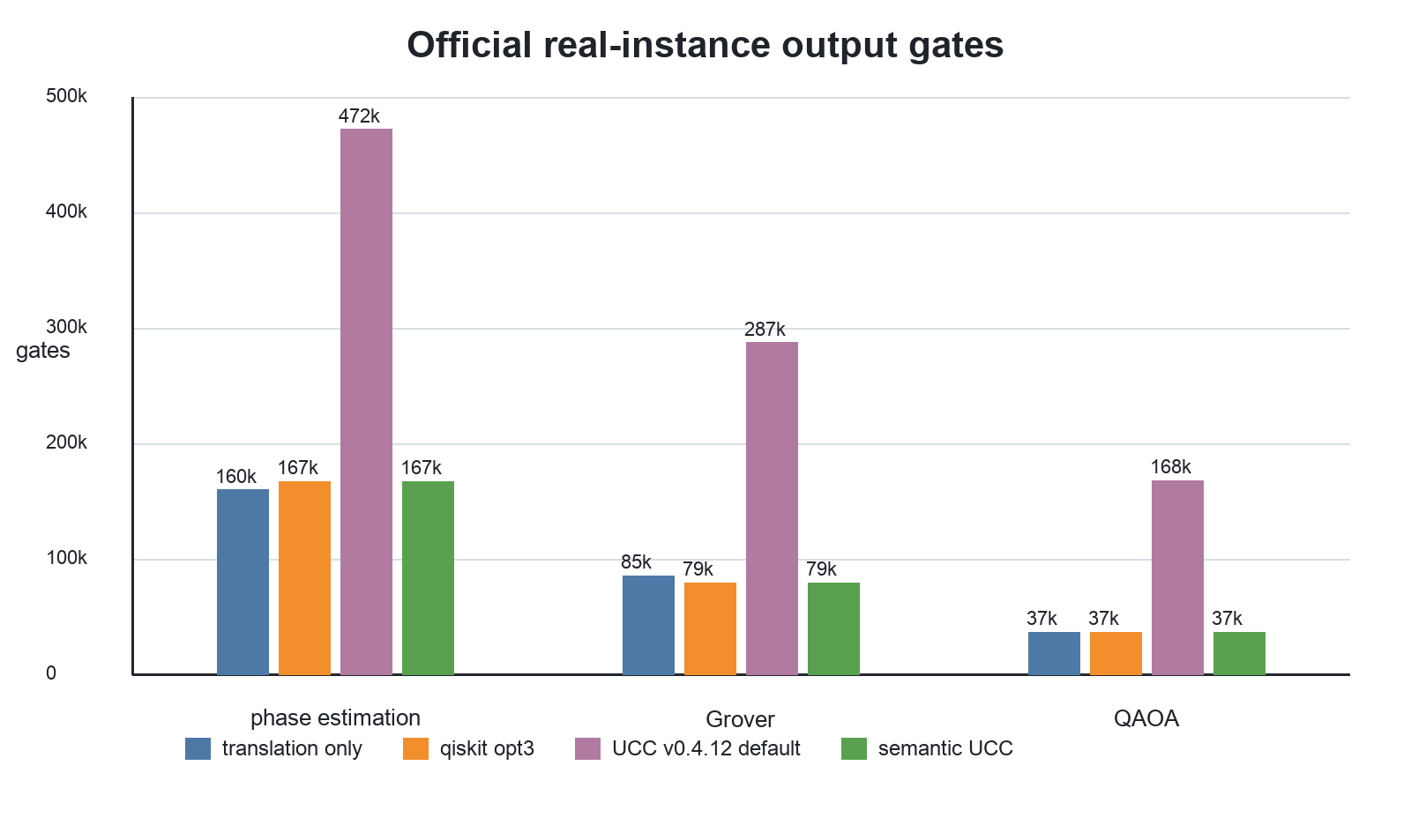}
\caption{Total output gate counts on the three official real instances.
The semantic UCC artifact matches \texttt{qiskit opt3}-level output
quality while substantially improving over the upstream UCC v0.4.12
pre-semantic default.}
\label{fig:real-instance-bars}
\end{figure*}

These results support the recoverability-parity regime
(\cref{fig:real-instance-bars}): the semantic artifact matches
\texttt{qiskit opt3}-level output quality while improving over the
upstream UCC v0.4.12 pre-semantic default of \cref{tab:real-instances} by
approximately 65\%, 73\%, and 78\% gate reduction on the three instances.
The selector cost in this panel is the implementation's conservative
structural score $\kappa = \text{gates} + \text{depth} +
10\cdot\text{multi-qubit}$, with candidate-control guards; the factor
$10$ is a conservative lexicographic-style heuristic favoring lower
two-qubit exposure, used only for portfolio selection. This explains the
\texttt{REAL24-phase-est.} row (source key
\texttt{phase\_estimation\_real}) in which translation-only has fewer
total gates than the selected candidate but more \texttt{cx} gates and a
higher score ($957{,}132$ versus $870{,}919$). This quality-recovery
pattern shows that semantic lifting behaves as an anti-regression
mechanism on workloads beyond the minimal Fourier construction.

\subsection{Natural-Kernel Factorial Ablation}
\label{sec:natural-factorial}

The frozen W8 campaign crosses nine workload families, scales 1/2/4, three
seeds, densities 0.5/1.0, and all-to-all/line hardware topologies.  Seven are
structured families: QAOA/Ising, diagonal-Hamiltonian simulation,
commuting-block Trotterization, QPE controlled powers, QFT arithmetic,
controlled-power ladders, and phase-polynomial blocks.  Random
Clifford+$R_Z$ and routing-dominated circuits are the two negative controls.
For every instance the authors' reference is evaluated on the complete
$2^6$ toggle matrix
\begin{equation*}
\begin{split}
&\text{semantic lift}\times\text{aggregation}\times\text{selector}
\times\text{cache/reuse}\\
&\qquad\times\text{preset recognizer}
\times\text{projected-block selection}.
\end{split}
\end{equation*}
The authors' factorial contributes 20,736 cells, and the strongest matched
external semantic baseline, TKET PauliSimp, contributes 324, for
\NaturalCells\ cells total.  Every cell is \texttt{completed\_valid}, passes
the independent dense projective certificate at six qubits, and carries the
same fixed-total-error W7 QRE fields.  Size increases repeat and densify
structure rather than relaxing the certificate width.

\begin{table*}[tbp]
\centering
\small
\begin{tabular}{lcc}
\toprule
Family & external-minus-authors gates [95\% CI] & strict success \\
\midrule
\input{generated/natural_families}
\end{tabular}
\caption{Complete structured-family summary against the matched external
semantic-capable baseline.  The predeclared strict criterion requires valid
outputs, no worse gate and CX counts at every primary cell, and lower measured
compiler runtime at every primary scale and seed.  Non-winning families are
retained.}
\label{tab:natural-factorial}
\end{table*}

\begin{figure*}[t]
\centering
\includegraphics[width=0.94\textwidth]{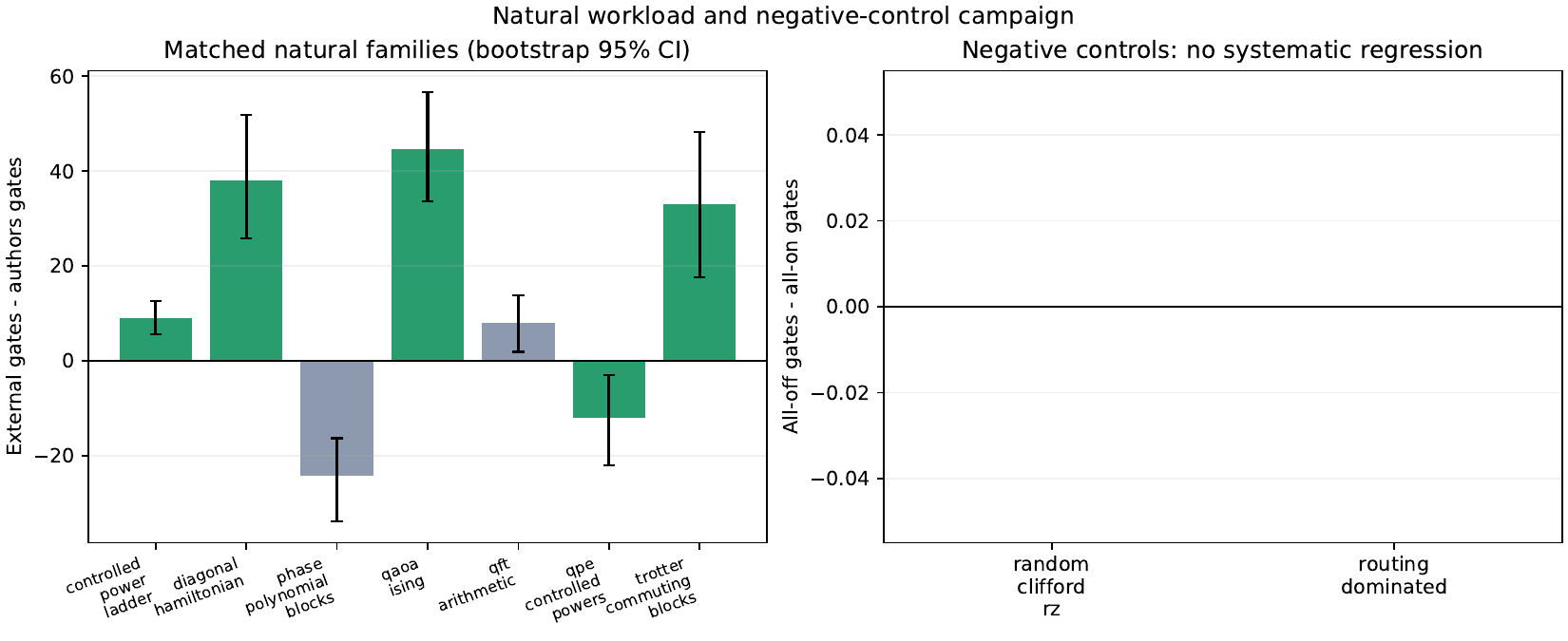}
\caption{Matched natural-family comparison and negative controls.  Error bars
are paired bootstrap 95\% intervals; all structured families and both controls
are retained.}
\label{fig:natural-summary}
\end{figure*}

\begin{figure*}[t]
\centering
\includegraphics[width=0.92\textwidth]{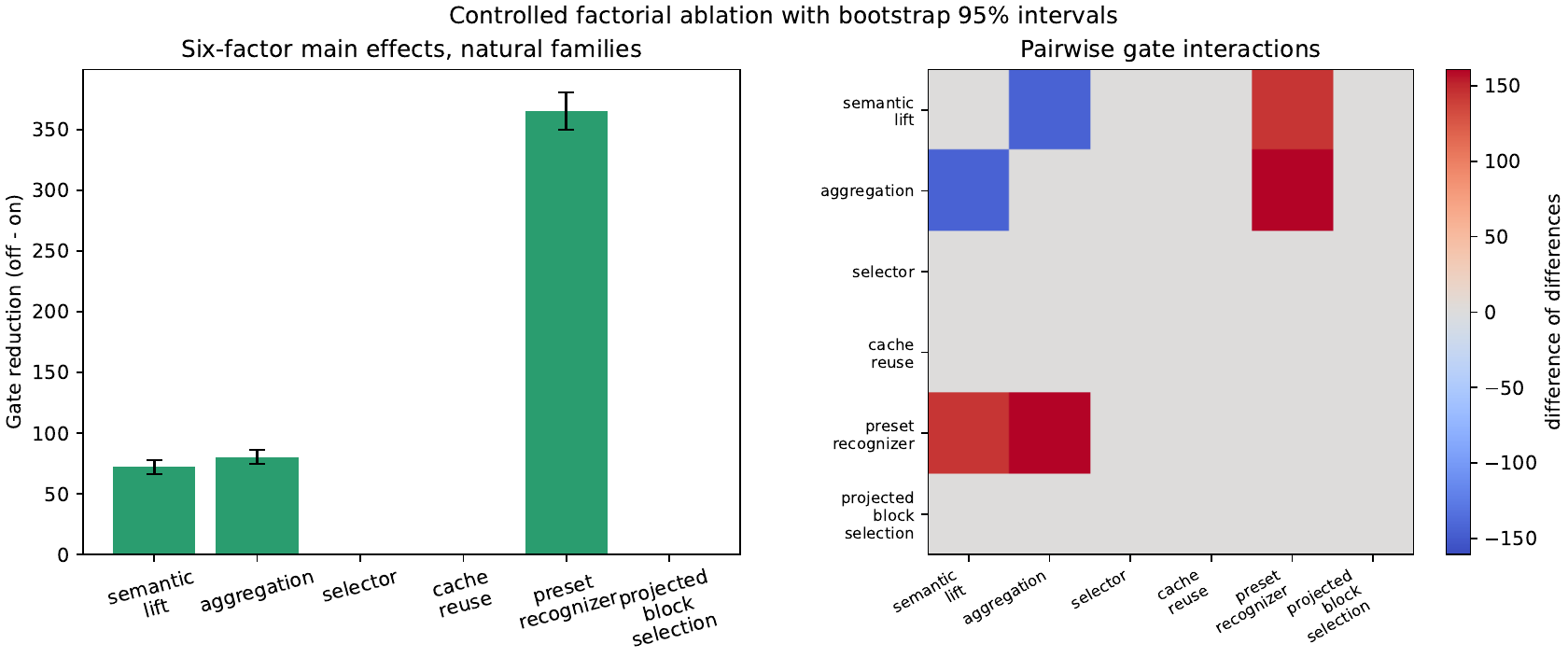}
\caption{All six factorial gate-count main effects and pairwise interactions.
Zero-effect factors are displayed rather than omitted.}
\label{fig:natural-factorial}
\end{figure*}

Five structured families satisfy the predeclared all-size/all-seed criterion;
QFT arithmetic and phase-polynomial blocks do not, and remain visible in
\cref{tab:natural-factorial,fig:natural-summary}.  The largest isolated
gate-count contribution is the preset recognizer, with mean reduction
\PresetGateEffect\ and bootstrap interval
[\PresetGateEffectLow,\PresetGateEffectHigh].  Aggregation and semantic lift
also have positive mean effects; selector, cache/reuse, and projected-block
selection have zero isolated gate effect in this controlled implementation.
Both negative controls have zero all-on-minus-all-off differences in gates,
depth, CX, logical $T$, and spacetime, hence
\NaturalNegativeRegressions\ systematic quality regressions in the frozen
grid.  These are finite reference measurements, not empirical proofs of the
streaming theorem or universal compiler superiority.

\section{Model-Specific Fault-Tolerant Resource Consequences}
\label{sec:ft-consequences}

The recoverability bounds of
\cref{sec:reduction,sec:w1-unified-frontier} establish an information
tradeoff, but they do not by themselves convert description bits into
gates or magic states.  The certified pipelines of
\cref{sec:validation} additionally exhibit different rotation
multiplicities.  This section propagates those measured circuit outputs
through Clifford+$T$ synthesis and surface-code scheduling under one
disclosed and versioned cost model.  The resulting physical numbers are
therefore a model-specific consequence of the measured pipelines, not a
representation-independent corollary of the bit lower bound.

The representation gap propagates into the fault-tolerant resources of the
explicit model in
\cref{tab:resource-consequence,fig:qre-pareto,fig:qre-sensitivity}.  Over all 96
matched target/error/allocation/QEC/factory settings, materialize-first uses
385--1{,}753 times as many logical $T$ states and 538--5{,}573 times the spacetime
volume of semantic-first.  No reversal occurs within the scanned error and
factory grid.  This is not a claim that the authors' implementation dominates
other semantic tools: external TKET \texttt{PauliSimp} reconstructs the phase
structure before synthesis and uses 0.443--0.536 times the semantic-first $T$
states on this witness.  Correctly crediting that recovery is part of the
comparison.

These values are outputs of actual checked angle synthesis followed by the
disclosed scheduling models, not universal hardware predictions.  The
absolute physical numbers still depend on routing, correlated-error, factory
layout, feed-forward, and code-distance assumptions outside this campaign.
The complete sensitivity plot and all preflight anomalies are retained in the
root QRE artifact package.

\begin{figure*}[t]
\centering
\includegraphics[width=0.98\textwidth]{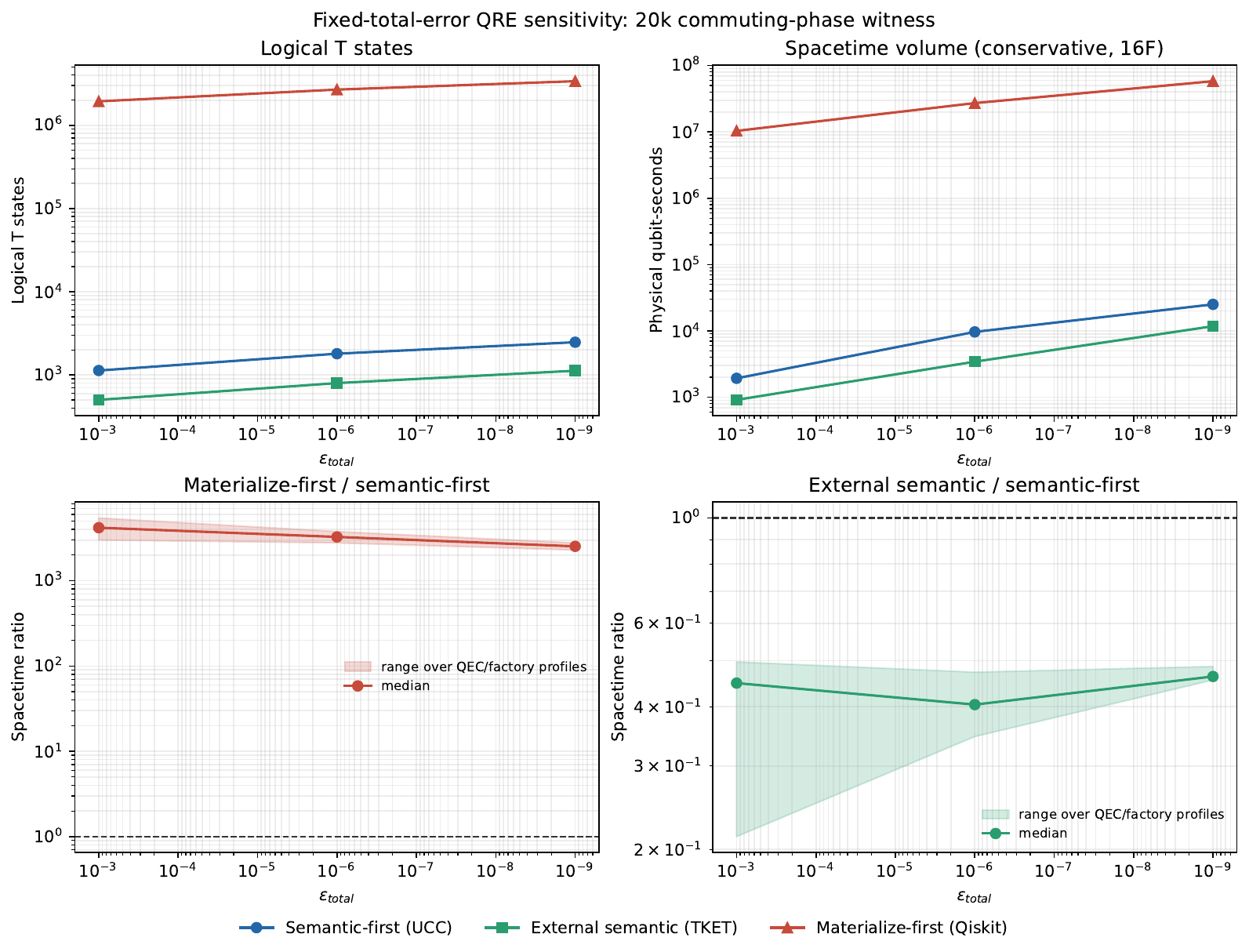}
\caption{Sensitivity to total error, QEC profile, and factory count on the 20k
witness.  Top: logical $T$ demand and the conservative-profile spacetime
volume.  Bottom: ratios to semantic-first, with bands spanning both QEC
profiles and all four factory counts.  The dashed line is equality.  No
materialize-first reversal occurs in the measured grid; the strong external
TKET path is reported on the semantic-capable side.}
\label{fig:qre-sensitivity}
\end{figure*}

\section{Related Work}
\label{sec:related-work}

\paragraph{Phase-polynomial, Pauli-gadget, and ZX optimization.}
Commuting-rotation aggregation is a standard optimization
target~\cite{amy2014tdepth,nam2018automated,ross2016optimal,cowtan2020phase,sivarajah2021tket,duncan2020graph,kissinger2020reducing,kissinger2020pyzx}.
The phase-polynomial line makes the generator--coefficient structure
explicit and then re-synthesizes the parity network, with or without a
connectivity
constraint~\cite{amy2019cnot,beaudrap2020techniques,meijer2023architecture};
the ZX line keeps a global diagrammatic normal form in which phases
migrate and fuse across syntactic
boundaries~\cite{vandewetering2020zx,kissinger2022simulating}.  All of
these are \emph{upper-bound} constructions: they exhibit a compiler that
recovers structure on the inputs where their rewrite system applies.
Our contribution is complementary and in the opposite direction: we
formalize a representation-dependent recoverability boundary with an
explicit $\epsilon$--resource tradeoff that no such pass can evade
without paying in the charged channels, connect it to communication
complexity, and measure model-specific fault-tolerant consequences under
one disclosed QRE.  Where these tools do recover the aggregate---as
rebased TKET \texttt{PauliSimp} does on our witnesses---the bound
charges the reconstruction state rather than treating it as free.
\Cref{sec:cut-tracing} measures that channel for the formal-model
reference compilers; the TKET API-boundary scan is reported separately
as a configured-pipeline diagnostic and is not identified with $S$.

\paragraph{Gate synthesis and $T$-count accounting.}
The packing entropy and the generic synthesis cost share the same
accuracy scale, but they are distinct statements.
Solovay--Kitaev decomposition~\cite{dawson2006solovay}, exact Clifford+$T$
synthesis~\cite{kliuchnikov2013fast}, and the number-theoretic
Ross--Selinger-type approximation
algorithms~\cite{ross2016optimal,selinger2015efficient} realize generic
$Z$-rotations with word length $\Theta(\log(1/\epsilon))$; multi-qubit $T$-count lower
bounds sharpen the same accounting at the unitary
level~\cite{gheorghiu2022tcount}.  The metric packing fixes
$K_\epsilon$ independently, and the QRE campaign then measures the
actual checked synthesis cost of each distinct output angle rather than
assigning one uniform repair cost to every coordinate.

\paragraph{Compiler IRs, verified compilation, and full-stack tooling.}
Quantum software stacks use intermediate representations to preserve
information across
lowering~\cite{steiger2018projectq,javadi2024qiskit,sivarajah2021tket,martiel2022architecture,watkins2024surfacecode,hietala2021verified,mills2021application,fellous2023optimizing}.
Recent designs are explicit about what the IR retains: OpenQASM~3 admits
timing and classical control in the instruction
stream~\cite{cross2022openqasm}, while SSA- and MLIR-based
representations expose dataflow so that passes can be composed and
verified~\cite{peduri2022qssa,mccaskey2021mlir}.  Our budget-explicit
model charges exactly this retained information---$B_{\mathrm{IR}}$ in
\cref{eq:info-budget}---rather than treating a richer IR as a free side
channel, which is what makes the semantic and materialized pipelines
comparable at all.

\paragraph{Routing, mapping, and benchmark suites.}
Architecture-directed passes reshape the stream that a downstream
compiler
sees~\cite{cowtan2019routing,li2019sabre,martiel2022architecture}, and
the platform-specific compilers for shuttling-based trapped
ions~\cite{kreppel2023shuttling} and dynamically field-programmable
neutral-atom arrays~\cite{tan2024dpqa} show how far that reshaping can
go before the stream reaches a downstream stage at all.  In
our model the \textsc{swap} and routing records they insert normalize
away only when the accumulated Clifford permutes the declared generator
set up to sign, exactly the hypothesis of
\cref{lem:clifford-interleaving}.  General routing that violates this
hypothesis is outside that corollary.  The
routing-dominated family is used as a negative control in
\cref{sec:workloads} for precisely this reason.  Benchmark suites
supply the workload conventions we
follow~\cite{quetschlich2023mqt,tomesh2022supermarq,mills2021application,li2023qasmbench}.

\paragraph{Fault-tolerant resource estimation.}
The downstream cost of a $T$ gate is set by distillation and by the
scheduling
model~\cite{bravyi2005magic,fowler2012surface,horsman2012lattice,litinski2019game,litinski2019magic}.
End-to-end estimates for concrete algorithms~\cite{gidney2021factor} and
the estimation frameworks that generalize
them~\cite{beverland2022assessing,watkins2024surfacecode} establish the
methodology used in \cref{sec:ft-consequences}: upstream rotation
multiplicity and synthesis accuracy propagate through distillation and
scheduling into spacetime cost, and the margin available to a quadratic speedup is
thin enough that they matter~\cite{babbush2021focus}.  Our QRE campaign
is an instance of this methodology under a disclosed model, not a new
estimator.

\paragraph{Noise tailoring and error mitigation.}
Randomized compiling and Pauli
twirling~\cite{wallman2016noise,hashim2021randomized,knill2005realistically}
are the mechanism that produces the dispersed stream we analyze, and
their sign randomization is exactly the transformation our
\cref{lem:clifford-interleaving} normalizes.  The broader mitigation
literature~\cite{cai2023mitigation}, including learned sparse noise
models used for probabilistic error
cancellation~\cite{vandenberg2023probabilistic}, shares the structural
feature that matters here: the physically executed circuit is a
randomized ensemble whose individual members do not display the
quantities the source program wrote down.

\paragraph{Classical lower-bound methodology.}
The proof imports fooling-set and communication-complexity reasoning into
a quantum-compilation representation
question~\cite{sipser2012introduction,kushilevitz1997communication,alur2010streaming}.
The pass-explicit form of the frontier follows the classical
space--passes tradition for one-way and multi-pass computation over a
stream~\cite{munro1980selection,alon1999space}, where the object charged
is the state serialized at a cut rather than total work.

\section{Discussion}
\label{sec:discussion}

\subsection{Representation as a Resource}

The central message is that representation is a computational resource in
quantum compilation.  A compiler receiving a semantic representation that
preserves generator--coefficient structure can aggregate in $O(m)$
records.  The cyclic packing gives
$\Omega(m\log(1/\epsilon))$ bits separately for family-relative and
self-contained output codes and is compatible with an $r$-round modular
share stream.  In the stated forward-scan model its pass-explicit frontier is
$\overline A_p+(2p-1)S\geq
(1-\delta)m\log_2K-h_2(\delta)$
(\cref{thm:w1-unified-frontier}).  Thus a deferred flat compiler needs
$S=\Omega(m\log K/p)$.  The block hybrid attains
$S=O((m/p)\log(K+1))$ with compact family-relative output, within the
proved constant factor, while the direct semantic vector uses only
$O(\log(mK))$ state
(\cref{prop:w1-block-hybrid,prop:w1-semantic-upper}).  The older binary
two-round theorem remains a sharper zero-error, one-pass special case.

This boundary is channel-explicit.  Equation~\eqref{eq:info-budget}
serializes control, store, window structure, every live parameter, every
live IR object, cache/hash layout, dynamic-pass state, and random seed.  A
compiler that reconstructs a global aggregate---as
rebased TKET PauliSimp does on the fixed-width witnesses---may succeed, but
its reconstruction table or graph is charged in $B_{\mathrm{cut}}$ rather
than treated as a free semantic side channel.

\subsection{Scope and Limitations}

\begin{enumerate}[leftmargin=1.5em]
    \item The results have different scopes.  The unified theorem
    (\cref{thm:w1-unified-frontier}) covers $r$-round additive shares,
    randomized failure, and $p=p(m)$ forward scans of the round-major
    stream, but charges all $2p-1$ serialized crossings and claims only a
    constant-factor state upper envelope.  Reverse scans, an uncharged
    random-access input oracle, or a different output code are outside that
    theorem.  The older representation-dependent tradeoff
    (\cref{thm:main-tradeoff}) is scoped to the two-round masked-share
    stream and gives the stronger deterministic one-pass mask count.  The
    exact-semantics corollary is scoped to fixed width, constant pass/cut
    budget, and certified noncollision.  The growing-width theorems instead
    use the complete bit codec and allow random access and dynamic
    $p=p(m,I,\rho)$.  The decoder is representation-invariant, but all
    syntax lengths retain their declared encoding semantics.
    \item The commitment toll (\cref{prop:commitment-toll}) is the most
    narrowly scoped result here, and deliberately so.  It is a
    two-sided statement about a compiler class, not about commitment in
    general: clause (i) needs the prefix to be \emph{executed} at the
    cut, so that it must name the coordinates it acts on, and clause
    (ii) exhibits an admissible buffered compiler that pays nothing.
    The bound is also on an ensemble of masks compiled by one
    mask-uniform algorithm; for a single mask known to the serializer
    there is no toll to pay.  The payload term is charged against the
    declared record content rather than the round-one unitary, which
    fixes it only up to zero residues and the
    $(s,a)\mapsto(-s,-a)$ degeneracy (\cref{rem:self-id-scope}).  Both
    clauses are measured only on the instrumented device of
    \cref{fig:commit-toll}; the two preregistered attempts to observe
    the toll at a third-party compiler boundary failed and are reported
    as failures.
    \item The tuned cyclic packing is stated for
    $0<\epsilon\leq1/8$ and chooses sharing modulus
    $Q_\epsilon=K_\epsilon+1$.  Its coefficient interval can cross $\pi$;
    \cref{thm:w1-exact-separation} handles circular distance and global
    phase exactly.  The prior no-wrap packing and the binary/$\mathbb Z_3$
    theorem remain separate valid specializations.
    The configured scaling experiment uses $K=2$, $\Delta=\pi/4$ and is not
    claimed to instantiate the masked-share theorem.
    \item The fixed-total-error experiment performs actual Clifford+$T$
    synthesis and executable surface-code calculations under two disclosed
    physical/QEC profiles, but these profiles share one phenomenological model
    family.  Its input is the measured rotation multiplicity of the three
    configured pipelines; the information lower bound alone does not imply a
    $T$-count or spacetime-volume ratio.  Absolute physical counts still require
    cross-validation under an independent estimator and architecture-specific
    routing/factory layouts.
    \item The packing-family scaling experiment
    (\cref{sec:packing-scaling}) links the growing-$m$ theorems to
    configured-pipeline behavior.  The instrumented reference compilers measure
    the masked-share cut fields directly (\cref{sec:cut-tracing}), but they are
    controlled formal-model objects, not instrumented versions of every
    external compiler.
    \item Dense independent certification fixes the W8 natural-workload width
    at six qubits; the size axis repeats and densifies structure rather than
    increasing Hilbert-space width.  Runtime and RSS are local-machine
    measurements.  The two negative controls and three zero gate-effect
    ablation factors are retained as finite null results, not safety theorems.
    \item The clean-room reconstruction creates an independent working
    directory and fresh virtual environment bound to locked project packages.
    It proves directory-level reconstruction of selected theorem, benchmark,
    and QRE results, not binary or container independence from the host.
\end{enumerate}

\section{Conclusion}
\label{sec:conclusion}

We have formalized approximate recoverability with a growing-width token
codec and a bit-level cut budget whose five executable serialization fields
include parameter, RAM/cache, random-seed, dynamic-pass, and IR channels.
The cyclic $K_\epsilon$-ary family has exact minimum separation
$2\sin(\pi/(2(K_\epsilon+1)))$, satisfies
$K_\epsilon=\Theta(1/\epsilon)$, and admits semantic and $r$-round flat
representations of the same $U_x$ even though no individual update
determines $x_j$.  For randomized failure $\delta$ and $p=p(m)$ forward
passes, each named output code, as well as the sum of the prefix-side
commitment and crossing transcript, requires at least
$(1-\delta)m\log_2K_\epsilon-h_2(\delta)$ expected bits.  Hence a deferred
compiler satisfies $(2p-1)S\geq
(1-\delta)m\log_2K_\epsilon-h_2(\delta)$.  A block hybrid reaches compact
family-relative output with
$S=O((m/p)\log(K_\epsilon+1))$, within the proved constant factor, and the
randomized relative-output bound is matched up to framing.  The semantic
vector achieves compact output with logarithmic state.  The two payment
channels exchange at par in the frontier but not everywhere beneath it:
commitment that is \emph{executed} at the cut additionally pays
$\log_2\binom{m}{k}$, the entropy of the dispersal mask, while deferred and
merely buffered commitment pay none of it, and we prove and measure both
directions rather than only the one that favors the thesis.  The
exact-semantics corollary retains the fixed-width final-output lower bounds
of the prior version.  The traced reference-compiler suite measures the
named commitment, state, and IR channels directly; the matched and external
campaigns show when global semantic tools recover compact structure; and the
natural-workload factorial records positive, null, and negative-control
outcomes without selection.  Fixed-total-error synthesis/QRE propagates the
measured rotation multiplicities of the configured pipelines through an
explicit physical model; the bit lower bound alone does not determine that
cost.  Every W3--W9
headline measurement is generated from the normalized frozen-data registry,
while retained legacy diagnostics are labelled historical.  Recoverability
should be treated as a controlled variable in compilation and resource
estimation.

\section*{Reproducibility}
\label{sec:reproducibility}

For scale, the headline campaigns comprise \TradeoffRuns\ traced
reference-compiler runs; a \MatchedCells-cell matched-representation
design and a separate \ExternalCells-cell external campaign containing
TKET PauliSimp and PyZX full-reduce, with completed metrics accepted
only when supported by the stated symbolic or dense certificate; \QRECells\
complete fixed-total-error fault-tolerant estimates under a common
algorithmic, synthesis, and logical-failure budget; and \NaturalCells\
natural-workload and negative-control cells crossing a six-factor
ablation, in which \NaturalStableFamilies\ structured families satisfy
the predeclared all-size/all-seed criterion and the two negative
controls show \NaturalNegativeRegressions\ systematic regressions.
Explicitly labelled historical diagnostics retain their original
artifact provenance.

The canonical registry is \path{data/manifest.yaml}; the normalized schema is
\path{schema/experiment.schema.json}.  Its complete key is
\begin{center}
\footnotesize\ttfamily
experiment\_id, instance\_id, representation\_id,\\
method\_id, seed, backend\_hash, tool\_version,\\
artifact\_commit, config\_hash.
\end{center}
The normalized Parquet index contains \UnifiedRows\ unique records and retains
\UnifiedPredicateErrors\ predicate-error rows outside quality means.  The data
audit checks CSV/Parquet mirrors, conflicting numeric cells, duplicate TeX
labels, source hashes, row counts, orphaned frozen files, and historical versus
canonical campaign disposition.  It finds no duplicate primary keys,
registered-file omissions, conflicting mirror cells, hash drift, or
unregistered frozen artifacts among the committed canonical sources.

All measurement macros, tables, and figures used by the W3--W9 sections are
generated from registered frozen sources by
\path{scripts/build_all_tables.py} and
\path{scripts/build_all_figures.py}; the operation and unique source row for
each printed number are recorded in \path{data/provenance_map.yaml} and tested
by \path{tests/test_paper_number_provenance.py}.  The correctness language,
symbolic checker, and mutation/dense audit are in \path{certificates/}; the
QRE schemas and two physical profiles are in \path{qre/}.  The approximately
700-MB raw W3--W8 run tree, including manifests, logs, certificates, seeds, and
package records, is versioned in the separate private
\path{OIerYangJZ/paper1-data} archive at the revision recorded in
\path{data/manifest.yaml}; its absence is reported but does not invalidate the
committed frozen-data rebuild.

The one-command entry point is
\path{reproducibility/run_submission.sh}.  From a fresh clone it creates the
locked project environment when needed, reruns frozen-data audits, regenerates
all registered tables and figures, executes the full test suite, builds the
canonical, arXiv, Quantum-submission, and supplement PDFs, and rejects
unresolved references.  This procedure reproduces the submitted analysis from
immutable frozen datasets; it does not rerun the complete measurement
campaigns.  The clean-room runner
copies a minimal artifact to an independent temporary directory, creates a
fresh environment bound to the locked project packages, reruns
\CleanTheoryTests\ headline-theorem tests, compiles and densely certifies a
matched QAOA/Ising instance through the authors and TKET paths, and exercises
\CleanGridPairs\ deterministic synthesis pairs in a fixed-total-error QRE
subset.  This is working-directory isolation, not a claim of host-binary
independence.  Complete audit and limitation records are
\path{DATA_AUDIT_REPORT.md}, \path{CLEAN_ROOM_REPORT.md},
\path{QRE_REPORT.md}, and \path{NATURAL_WORKLOAD_REPORT.md}.

\begin{acknowledgements}
J.Y.'s deepest thanks go to two teachers who are also coauthors of this
paper, Yangyang Li and Xiu-Hao Deng.  Their help was immense; without them
this paper would not exist.

J.Y. thanks Ke Lin of Xidian University for opening the door to research
and for the recommendation that led him to join Prof.~Li's group.

J.Y. thanks Brad Chase: the first ideas behind this paper arose while
fixing an issue in his \texttt{ucc} repository.

J.Y. also thanks the teachers who were willing to offer guidance along the
way, and the senior members of Prof.~Li's group, who have been generous
with their help throughout this work.

Thanks are due, finally, to everyone who supported this paper, too many to
name individually here.

Y.L. acknowledges support from the National Natural Science Foundation of
China under Grant No.~62476209, the Key Research and Development Program of
Shaanxi under Grant No.~2024CY2-GJHX-18, and the Fundamental Research Funds
for the Central Universities under Grant No.~QTZX26097.  X.-H.D.
acknowledges support from the Science, Technology and Innovation Commission of Shenzhen
Municipality (Grant Nos.~JCYJ20170412152620376 and KYTDPT20181011104202253),
the Innovation Program for Quantum Science and Technology (Grant
No.~2021ZD0301703), the Guangdong Major Project of Basic Research (Grant
No.~2025B0303000007), and the Shenzhen Science and Technology Program
(Grant No.~KQTD20200820113010023).

\ifdefined\QuantumSubmissionVersion
\else
This paper pays particular attention to provenance: every printed number is
traceable to a unique source row, every experiment to a frozen registry, and
every preregistered prediction to a criterion written down before the run.  That being
so---\emph{what one begins with, one must end with} (\emph{Zuo Zhuan},
Xuan 12)---it is fitting to record the provenance of the framing as well,
even though that provenance does not lie in the physics literature.

Two passages from the classical Chinese canon shaped parts of this work.
Their relevance is direct, and the acknowledgements provide a fitting
context in which to quote them.

The first is from the \emph{Xici} commentary to the \emph{Yijing}: ``What
is above form is called the Way; what is below form is called the vessel;
to transform and tailor it is called change; to carry it forward and put it
into practice is called continuity.''  What is useful here is not an
analogy but a set of distinctions.  The passage names four things
separately, and ranks none of them above another: a description that
carries structure, a description that carries only instructions, the act of
lowering the one into the other, and the act of execution.  A compiler is
precisely the engine that performs the third and hands off the fourth; and
the vocabulary of compilation itself---semantic level, lowering,
commitment---reproduces the same fourfold division without, so far as we
know, having borrowed it.  The lesson we draw from the passage is that
the level at which a computation is described is not a notational
convenience to be optimized away, but a variable that can be charged for.
This paper charges it, in bits.

The second is Chapter 64 of the \emph{Laozi}: ``What is at rest is easy to
hold; what has not yet shown a sign is easy to plan for; what is brittle is
easy to shatter; what is minute is easy to scatter.  Act on it before it
comes to be; order it before it falls into disorder.''  Every clause here
prices an action by the state of the object acted upon rather than by the
action itself---which is exactly this paper's claim about representation:
the same aggregation is cheap or expensive depending only on the form in
which the layer arrives.  The verb of the first clause, \emph{to hold}, is
the operation charged here as compiler state.  The chapter further
separates planning early, which it recommends, from grasping early, warning
that ``one who grasps it loses it''; that separation is,
here, the difference between a compiler that defers and one that commits,
and the price of the difference is the commitment toll.  And the rule on
which the chapter closes is the design rule of our abstract, read from the
other side.

Neither passage contains a theorem, and neither anticipates one; no result
reported here depends on either.  A reader who finds these readings
strained may discard them entirely and judge the paper by whatever in it
proves useful.  They are recorded because the question did come from there
and because, for a framing as for a datum, a reader may rightly ask where it
came from.  Translations are J.Y.'s own.
\fi

\begingroup
\looseness=-1
Generative AI tools were used to assist with manuscript drafting and
revision, code and repository maintenance, literature discovery, and
consistency checks.  The authors directed its use through task-specific
prompts and independently reviewed the generated text, citations, code, and
numerical claims against the cited literature, source files, executable tests,
and the symbolic or dense certificates reported in this work.  AI output was
not treated as a scientific source or as independent proof verification.  The
authors take full responsibility for the manuscript, proofs, code, data,
figures, citations, and conclusions.\par
\endgroup
\end{acknowledgements}

\section*{Author contributions}

J.Y. developed the theory and proofs, implemented the compiler instrumentation
and experimental pipeline, and performed the majority of the measurements and
analyses.  Y.L. supervised the research, reviewed and revised the manuscript,
and provided project administration and funding support.  X.-H.D.
independently validated the theoretical derivations, numerical results, and
computational artifacts.  All authors wrote the manuscript together,
discussed the results, and approved the final version.

\bibliographystyle{quantum}
\bibliography{references}

\appendix
\input{appendix/appendix_A}
\input{appendix/appendix_B}
\input{appendix/appendix_C_prxq}

\end{document}

%% file: generated/paper_numbers.tex
\newcommand{\TradeoffRuns}{\DataNumber{w3-cells}{592}}
\newcommand{\MatchedCells}{\DataNumber{w4-cells}{1,080}}
\newcommand{\MatchedCompleted}{\DataNumber{w4-completed}{1,065}}
\newcommand{\MatchedPredicateErrors}{\DataNumber{w4-predicate-errors}{15}}
\newcommand{\ExternalCells}{\DataNumber{w5-cells}{1,728}}
\newcommand{\ExternalCompleted}{\DataNumber{w5-completed}{1,695}}
\newcommand{\ExternalPredicateErrors}{\DataNumber{w5-predicate-errors}{33}}
\newcommand{\CertificateMutationCases}{\DataNumber{w6-mutations}{56}}
\newcommand{\CertificateFalseAccepts}{\DataNumber{w6-false-accepts}{0}}
\newcommand{\CertificateCrosscheckCases}{\DataNumber{w6-crosschecks}{12}}
\newcommand{\QRECells}{\DataNumber{w7-cells}{288}}
\newcommand{\QREMaterializeSpacetimeRatio}{\DataNumber{w7-spacetime-ratio}{2814}}
\newcommand{\NaturalCells}{\DataNumber{w8-cells}{21,060}}
\newcommand{\NaturalStableFamilies}{\DataNumber{w8-stable-families}{5}}
\newcommand{\NaturalNegativeRegressions}{\DataNumber{w8-negative-regressions}{0}}
\newcommand{\PresetGateEffect}{\DataNumber{w8-preset-effect}{364.9}}
\newcommand{\PresetGateEffectLow}{\DataNumber{w8-preset-effect-low}{349.8}}
\newcommand{\PresetGateEffectHigh}{\DataNumber{w8-preset-effect-high}{380.5}}
\newcommand{\UnifiedRows}{\DataNumber{w9-unified-rows}{24,748}}
\newcommand{\UnifiedPredicateErrors}{\DataNumber{w9-predicate-errors}{48}}
\newcommand{\CleanTheoryTests}{\DataNumber{w9-clean-theorem-tests}{9}}
\newcommand{\CleanGridPairs}{\DataNumber{w9-clean-grid-pairs}{17}}

%% file: theory/model_implementation_mapping.tex
\subsection{Executable Model--Implementation Correspondence}
\label{sec:model-implementation-mapping}

The formal fields above are implemented by one canonical binary codec rather
than by dimension-changing estimates.  The executable specification is
\texttt{encoding/FORMAT.md}; its reference implementation is
\texttt{encoding/codec.py}, and the only bit-accounting entry point is
\texttt{instrumentation/resource\_accounting.py}.  A complete codec frame has
the form
\begin{equation}
 \mathtt{UCCACCT0}\,\Vert\,
 \mathsf u(\text{version})\,\Vert\,
 \mathsf u(\text{type})\,\Vert\,
 \mathsf u(|z|)\,\Vert\,z,
 \label{eq:executable-codec-frame}
\end{equation}
where the displayed terminal zero denotes a null byte, $\mathsf u$ is the
canonical base-$128$ unsigned varint, and $z$ is a typed canonical payload.
Signed integers use zigzag coding; rational coefficients are reduced signed
numerator/positive-denominator pairs.  Maps are ordered by their UTF-8 key
bytes.  Qubit, generator, symbol, location, node, edge-endpoint, port, pass,
and head addresses are literal varints in $z$.

\begin{table*}[!t]
\centering
\scriptsize
\setlength{\tabcolsep}{3pt}
\begin{tabularx}{\textwidth}{@{}l l X X@{}}
\toprule
Formal object & Codec object/file & Included information & Measured field \\
\midrule
semantic input & \texttt{SemanticStream} & width and address-complete aggregate records & input artifact bits \\
dispersed flat input & \texttt{FlatUpdateStream} & width, round count, and every additive update & input artifact bits \\
$\mathsf{ser}_{\rm control}$ & \texttt{ControlState}/\texttt{control.uccbin} & pass count/index, heads, scheduling, status, framing and seed state when present & $B_{\rm control}$ \\
$\mathsf{ser}_{\rm store}$ & \texttt{StoreState}/\texttt{store.uccbin} & aggregate tables, maps, caches, stacks and their layout, excluding tagged parameter/IR payloads & $B_{\rm store}$ \\
$\mathsf{ser}_{\rm window}$ & \texttt{LiveWindow}/\texttt{window.uccbin} & ordered live complete-token structure and explicit parameter-slot count & $B_{\rm window}$ \\
$\mathsf{ser}_{\rm parameter}$ & \texttt{ParameterLedger}/\texttt{parameter.uccbin} & sparse symbolic parameters under complete location tuples & $B_{\rm parameter}$ \\
$\mathsf{ser}_{\rm IR}$ & \texttt{IRField}/\texttt{ir.uccbin} & serialized DAG width, opcodes, operands, nodes, edges, ports, boundaries and metadata, with parameters removed & $B_{\rm IR}$ \\
committed prefix & \texttt{committed-prefix.bin} & literal immutable prefix of a codec output frame & $B_{\rm com}$ \\
external output & \texttt{ExternalToolOutput} & tool, format and literal API payload & output artifact bits \\
certificate input & \texttt{CertificateInput} & certificate kind, certified-artifact SHA-256 and literal sidecar & certificate artifact bits \\
pass count & manifest \texttt{passes.count} and row \texttt{p} & actual declared scans of the instrumented compiler & $p$ \\
\bottomrule
\end{tabularx}
\caption{One-to-one map from formal serialized objects to executable artifacts.
No row permits a gate, token, support, or IR-node count to stand in for bits.}
\label{tab:model-implementation-map}
\end{table*}

\begin{lemma}[Canonical prefix-free implementation]
\label{lem:codec-prefix-free}
Version-one complete frames are injective and prefix-free.  Concatenations of
frames are uniquely parseable without an out-of-band arity or record count.
The semantic and flat decoders preserve their exact aggregate coefficient
table.
\end{lemma}

\begin{proof}
Canonical varints have a unique terminating byte, and the version, type, and
payload length are therefore uniquely decoded.  The declared payload length
locates the end of the frame.  A complete valid frame cannot be a proper prefix
of another complete valid frame: the shorter declared payload ends at its frame
boundary, whereas any following byte belongs to a subsequent frame.  Typed
payload values have distinct one-byte tags; lists and maps carry counts, maps
have strictly increasing keys, and reduced rational pairs are unique.  This
proves injectivity and unique concatenated parsing.  Finally,
\texttt{semantic\_signature} sums exact rational sparse parameters by the
literal pair (generator address, Pauli support), so encode--decode leaves the
aggregate table unchanged for both input representations.
\end{proof}

\begin{proposition}[File-backed accounting soundness]
\label{prop:file-backed-accounting}
Suppose the version-one verifier accepts a cut manifest $M$.  For every
$f\in\{\mathrm{control},\mathrm{store},\mathrm{window},
\mathrm{parameter},\mathrm{IR},\mathrm{com}\}$, let $A_f$ be the file named by
$M$.  Then the implementation reports
\begin{equation}
\begin{aligned}
 B_f&=8|A_f|_{\rm byte},\\
 B_{\rm cut}&=B_{\rm control}+B_{\rm store}+B_{\rm window}\\
             &\quad+B_{\rm parameter}+B_{\rm IR},
\end{aligned}
 \label{eq:file-backed-accounting}
\end{equation}
and reports the manifest pass count as $p$.  The reported values change when a
serialized address crosses a varint-width boundary.  They cannot be obtained
from a node or gate count alone.
\end{proposition}

\begin{proof}
The accountant first writes the five complete frames and the literal committed
prefix, then obtains every byte count by \texttt{stat}, stores the SHA-256 of
each file, writes $M$, and invokes the verifier before returning a row.  The
verifier reopens all six files, recomputes sizes and digests, decodes the five
restart fields, and rejects any disagreement.  It then forms the sum in
\eqref{eq:file-backed-accounting}; there is no count-to-bit code path.  Since an
address is a literal canonical varint, for example $127$ and $128$ occupy
different address lengths.  Conversely, two one-node DAGs may contain those
different addresses and hence have unequal file lengths despite equal node
count.
\end{proof}

\begin{remark}[Observable boundary versus restart cut]
\label{rem:observable-boundary-accounting}
The masked-share tracer serializes the five declared cut fields after every
update.  The verifier reopens and decodes every recorded field slice, checks
its size and digest, and recomputes the cut sum; its maximum is therefore a
file-backed measurement of the declared $B_{\rm cut}$ coordinate.  The W3
campaign does not claim that the Python process is destroyed and restarted
between events.  External compilers do not expose their private in-process DAGs,
caches, or tapes.  Their matched-matrix and natural-factorial rows consequently
report a file-backed \emph{API-boundary} snapshot and are diagnostic only; they
are not promoted to measurements of the maximum private $B_{\rm cut}$.  Peak
RSS remains a separate byte dimension and is not substituted for any field in
\eqref{eq:info-budget}.
\end{remark}

\begin{remark}[Self-contained and conditional descriptions]
\label{rem:dictionary-accounting-code}
Both input-stream objects declare either \texttt{self\_contained} or
\texttt{conditional}.  The former forbids an external dictionary.  The latter
requires the SHA-256 of the exact public family dictionary, and its manifest
records the dictionary's separate file size and whether it is charged at the
cut.  Thus a conditional description cannot silently be reported as a
self-contained circuit code.
\end{remark}

%% file: theory/qary_packing.tex
\subsection{Exact cyclic $K$-ary packing}
\label{sec:w1-qary-packing}

The no-wrap construction of \cref{def:packing-construction} is convenient
for output counting, but modular additive shares require a grid whose period
is also the sharing modulus.  The following formulation records exactly when
wraparound creates a collision.

\begin{definition}[Wrap margin]
\label{def:w1-wrap-margin}
For an alphabet $[K]_0$, step $\Delta>0$, and $K\geq2$, define
\begin{align}
 \operatorname{dist}_{2\pi}(a,0)
 &:=\min_{z\in\mathbb Z}|a-2\pi z|,\\
 \gamma(K,\Delta)
 &:=\min_{1\leq d\leq K-1}
       \operatorname{dist}_{2\pi}(d\Delta,0).
\end{align}
Let $P_1,\ldots,P_m$ be independent Pauli-$Z$ characters as in
\cref{def:packing-construction}, and set
\begin{equation}
 U_x^{(K,\Delta)}:=
 \exp\!\left[-\frac{i}{2}\sum_{j=1}^m x_j\Delta P_j\right],
 \qquad x\in[K]_0^m.
\end{equation}
\end{definition}

\begin{theorem}[Exact worst-pair separation]
\label{thm:w1-exact-separation}
For every $m\geq1$,
\begin{equation}
 \min_{x\neq y}d_{\mathrm{proj}}
 \bigl(U_x^{(K,\Delta)},U_y^{(K,\Delta)}\bigr)
 =2\sin\!\left(\frac{\gamma(K,\Delta)}4\right).
 \label{eq:w1-exact-separation}
\end{equation}
Consequently the family is projectively injective exactly when
$\gamma(K,\Delta)>0$.  If $(K-1)\Delta\leq\pi$, then
$\gamma(K,\Delta)=\Delta$.  If $Q\geq K+1$ and
$\Delta=2\pi/Q$, then again $\gamma(K,\Delta)=\Delta$, even though the
coefficient interval $[0,(K-1)\Delta]$ may cross the principal-argument
cut at $\pi$.
\end{theorem}

\begin{proof}
Fix $x\neq y$, put $s=x-y$, and choose $j$ with $s_j\neq0$.
Independence of the characters makes their joint sign map surjective.
Hence the relative unitary $U_x^{(K,\Delta)}
(U_y^{(K,\Delta)})^\dagger$ has two eigenvalues obtained from sign patterns
that agree off $j$ and whose circular angular separation is
$\operatorname{dist}_{2\pi}(|s_j|\Delta,0)\geq\gamma(K,\Delta)$.
After any common phase rotation, one of two circle points separated by
$\eta\in[0,\pi]$ is at chordal distance at least
$2\sin(\eta/4)$ from $1$.  This proves the lower bound in
\eqref{eq:w1-exact-separation} after taking the global-phase infimum.

Choose $d\in\{1,\ldots,K-1\}$ attaining $\gamma(K,\Delta)$ and take a pair that
differs by $d$ in one coordinate only.  Its relative spectrum consists of
two points with circular separation $\gamma(K,\Delta)$, so centering the
shorter arc attains $2\sin(\gamma/4)$.  This proves equality.  Under the
no-wrap hypothesis, every $d\Delta$ lies in $[\Delta,\pi]$, whose minimum
is $\Delta$.  For $\Delta=2\pi/Q$ and
$1\leq d\leq K-1\leq Q-2$, every nonzero grid distance is at least
$2\pi/Q$, with equality at $d=1$.
\end{proof}

\begin{corollary}[$\epsilon$-tuned cyclic family]
\label{cor:w1-epsilon-cyclic}
For $0<\epsilon\leq1/8$, let
\begin{equation}
 Q_\epsilon:=\left\lfloor
 \frac{\pi}{2\arcsin(2\epsilon)}\right\rfloor,
 \qquad K_\epsilon:=Q_\epsilon-1,
 \qquad \Delta_\epsilon:=\frac{2\pi}{Q_\epsilon}.
 \label{eq:w1-epsilon-grid}
\end{equation}
Then
\begin{equation}
 \frac1{4\epsilon}\leq K_\epsilon\leq\frac{\pi}{4\epsilon},
 \quad
 \min_{x\neq y}d_{\mathrm{proj}}(U_x,U_y)
 =2\sin\!\left(\frac{\pi}{2Q_\epsilon}\right)
 \geq4\epsilon.
 \label{eq:w1-epsilon-pack}
\end{equation}
Thus, with $\mathcal G_{m,\epsilon}:=
\{U_x:x\in[K_\epsilon]_0^m\}$,
\begin{equation}
 \log_2|\mathcal G_{m,\epsilon}|
 =m\log_2K_\epsilon
 =\Theta\!\left(m\log\frac1\epsilon\right).
 \label{eq:w1-packing-entropy}
\end{equation}
The allowed coefficient interval is
$[0,(Q_\epsilon-2)2\pi/Q_\epsilon]$.  It can extend beyond $\pi$ but
omits the residue $Q_\epsilon-1$; \cref{eq:w1-exact-separation}, rather
than a principal-interval argument, controls its wraparound.
\end{corollary}

\begin{proof}
The bound $\arcsin(2\epsilon)\leq\pi\epsilon$ and
$\lfloor a\rfloor\geq a-1$ give
$K_\epsilon\geq1/(2\epsilon)-2\geq1/(4\epsilon)$; the reverse bound uses
$\arcsin(2\epsilon)\geq2\epsilon$.  Since
$Q_\epsilon\leq\pi/[2\arcsin(2\epsilon)]$,
$\pi/(2Q_\epsilon)\geq\arcsin(2\epsilon)$.  Apply
\cref{thm:w1-exact-separation}.  The cardinality statement is immediate.
\end{proof}

%% file: theory/dispersed_update_tradeoff.tex
\subsection{Dispersed additive-update representation}
\label{sec:w1-dispersed}

\begin{definition}[$r$-round dispersed update stream]
\label{def:w1-dispersed-stream}
The cyclic specialization of \cref{def:phase-family} with modulus
$Q\geq3$, $K=Q-1$, $\Delta=2\pi/Q$, $r\geq2$ rounds, and
$\Sigma=\{+1\}$: for $x\in[K]_0^m$ the shares
$a^{(1)},\ldots,a^{(r)}\in\mathbb Z_Q^m$ satisfy
\begin{equation}
 \sum_{t=1}^r a_j^{(t)}=x_j\pmod Q
 \qquad(j\in[m]),
 \label{eq:w1-share-sum}
\end{equation}
the flat input is the round-major record sequence
$\mathsf{Update}(t,j,a_j^{(t)})$, and the semantic input is the addressed
aggregate vector $x$.
\end{definition}

\begin{proposition}[Semantic equivalence and token privacy]
\label{prop:w1-stream-equivalence}
Every valid flat stream of \cref{def:w1-dispersed-stream} implements the
same $U_x$ as the semantic input, up to global phase.  Moreover, a single
flat token never determines $x_j$: for every fixed $(t,j,a)$ and every
$b\in[K]_0$, there is a valid stream with $a_j^{(t)}=a$ and $x_j=b$.
Under the canonical sharing distribution in which the first $r-1$ shares
are independent uniform residues and the last is fixed by
\eqref{eq:w1-share-sum}, every individual residue is uniform on
$\mathbb Z_Q$ and independent of $x_j$.
\end{proposition}

\begin{proof}
For integer representatives there are $k_j\in\mathbb Z$ with
$\sum_ta_j^{(t)}=x_j+Qk_j$.  Commutativity and $Q\Delta=2\pi$ give
\begin{equation*}
 \prod_tR_{P_j}(\Delta a_j^{(t)})
 =R_{P_j}(\Delta x_j+2\pi k_j)
 =(-1)^{k_j}R_{P_j}(\Delta x_j).
\end{equation*}
The product of the signs is global.  For privacy, fix the named token and
choose any one of the other $r-1$ residues to enforce the desired sum;
the remaining residues are arbitrary.  The distributional claim follows
because a sum containing at least one independent uniform residue is
uniform, including the last share after conditioning on $x_j$.
\end{proof}

\begin{definition}[Forward-pass cut accounting]
\label{def:w1-pass-accounting}
Partition the round-major stream into Alice's prefix $A$ (rounds
$1,\ldots,r-1$) and Bob's suffix $B$ (round $r$).  A forward $p$-pass
compiler scans $AB$ from left to right at most $p=p(m)$ times.  There are
$p$ crossings $A\to B$ and at most $p-1$ rewind crossings $B\to A$.
At every crossing, the complete restart state is the five-field
serialization of \eqref{eq:info-budget}; rereadable input and every live
IR object are therefore included.  For a run on $I$, let $A_p(I)$ be the
number of bits
irreversibly committed to the final output while the machine scans $A$,
summed over all passes, and define
\begin{equation}
 \Sigma_p(I):=\sum_{t=1}^{p}B_{\mathrm{cut}}(c_{A\to B,t})
       +\sum_{t=1}^{p-1}B_{\mathrm{cut}}(c_{B\to A,t}).
 \label{eq:w1-transcript-budget}
\end{equation}
Both are literal, self-delimiting transcript lengths.  Unbarred resource
symbols take maxima over valid inputs.  For randomized compilers, write
\begin{equation*}
 \overline A_p:=\max_I\mathbb E_\rho A_p(I,\rho),\qquad
 \overline\Sigma_p:=\max_I\mathbb E_\rho\Sigma_p(I,\rho),
\end{equation*}
and let
$S:=\max_I\mathbb E_\rho[\max_cB_{\mathrm{cut}}(c;I,\rho)]$ over the
$2p-1$ named crossings.  Then
$\overline\Sigma_p\leq(2p-1)S$; the deterministic version is identical
without expectations.
\end{definition}

\begin{theorem}[One-pass zero-error mask bound]
\label{thm:w1-one-pass-mask}
Let a deterministic one-pass compiler be correct to
$\epsilon<\sin(\pi/(2Q))$ on every valid stream of
\cref{def:w1-dispersed-stream}.  Then
\begin{equation}
 A_1+B_{\mathrm{cut}}(c_{A\to B,1})
 \geq\left\lceil m\log_2Q\right\rceil.
 \label{eq:w1-one-pass-mask}
\end{equation}
\end{theorem}

\begin{proof}
For each $u\in\mathbb Z_Q^m$, choose a canonical prefix whose aggregate
is $u$ (put $u$ in its first round and zero in the other prefix rounds).
The committed prefix and restart state after this prefix must distinguish
all $Q^m$ choices.  Indeed, for $u\neq u'$, choose one common final share
$v$ coordinatewise.  At each coordinate at most two residues of $v_j$
would make either $u_j+v_j$ or $u'_j+v_j$ equal to the single excluded
aggregate $Q-1$; since $Q\geq3$, an allowed $v_j$ exists.  Both completions
are valid, and they differ in every coordinate in which $u$ and $u'$ differ.
By \cref{thm:w1-exact-separation}, one circuit cannot be within
$\epsilon$ of both targets.  The messages are prefix-free under the
declared serializers, so Kraft's inequality gives
\eqref{eq:w1-one-pass-mask}.
\end{proof}

%% file: theory/randomized_multipass_extension.tex
\subsection{Randomized multipass frontier}
\label{sec:w1-randomized-multipass}

Put
\begin{equation}
 L_\delta(m,K):=(1-\delta)m\log_2K-h_2(\delta),
 \qquad 0\leq\delta<\tfrac12.
 \label{eq:w1-Ldelta}
\end{equation}

\begin{theorem}[$K$-ary randomized multipass recoverability frontier]
\label{thm:w1-unified-frontier}
Let $Q\geq3$, $K=Q-1$, and
$0\leq\epsilon<\sin(\pi/(2Q))$.  Let a possibly randomized compiler be
correct within projective error $\epsilon$, with probability at least
$1-\delta$ for every $x\in[K]_0^m$ and every valid $r$-round dispersed
representation of $x$.  Assume that it makes at most $p=p(m)$ forward
passes in the charged model of \cref{def:w1-pass-accounting}.  Then:
\begin{align}
 \overline B_{\mathrm{rel}}^{\mathrm{out}}
 &\geq L_\delta(m,K),
 &
 \overline B_{\mathrm{self}}^{\mathrm{out}}
 &\geq L_\delta(m,K),
 \label{eq:w1-output-lower}\\
 \overline A_p+\overline\Sigma_p
 &\geq L_\delta(m,K).
 \label{eq:w1-pass-lower}
\end{align}
In particular, for the worst-input expected crossing cap $S$ of
\cref{def:w1-pass-accounting},
\begin{equation}
 \boxed{\quad
 \overline A_p+(2p-1)S
 \geq(1-\delta)m\log_2K-h_2(\delta).
 \quad}
 \label{eq:w1-frontier-box}
\end{equation}
Here $S$ includes RAM, parameters, cache/hash metadata, materialized input,
and all IR fields from \eqref{eq:info-budget}; it is not a RAM-only bound.
For the tuned choice of \cref{cor:w1-epsilon-cyclic}, both right-hand sides
are $\Omega((1-\delta)m\log(1/\epsilon))$ for fixed
$\delta<1/2$.
\end{theorem}

\begin{proof}
For \eqref{eq:w1-output-lower}, take $X$ uniform on $[K]_0^m$ and fix one
valid dispersed representation for each value of $X$.
The packing separation is greater than $2\epsilon$, so nearest-packing
decoding recovers $X$ from a successful output.  Fano's inequality gives
\begin{equation*}
 I(X;C)\geq m\log_2K-h_2(\delta)
 -\delta\log_2(K^m-1)\geq L_\delta(m,K).
\end{equation*}
Expected length in either separately fixed prefix-free output code is at
least this mutual information; a worst-input expectation is at least the
uniform average.

For \eqref{eq:w1-pass-lower}, independently choose a uniform prefix
aggregate $U\in\mathbb Z_Q^m$, put it in the first prefix round, put zero
in any other prefix rounds, and set the last share $V=X-U\pmod Q$.
Then $V$ is uniform and independent of $X$, so
$H(X\mid V)=m\log_2K$.  Convert the $p$ scans into the standard alternating
simulation across the $A\mid B$ partition.  The transcript $T$ contains
each complete crossing snapshot and all output chunks written on Alice's
side.  Given $T$ and $V$, Bob can reproduce every Bob-side segment and the
complete final output; random-tape state and dynamic control are present in
the crossing serialization.  The output decoder therefore recovers $X$
with error at most $\delta$, and conditional Fano yields
\begin{equation*}
 I(X;T\mid V)\geq L_\delta(m,K).
\end{equation*}
The framed transcript is prefix-free conditional on $V$, so its expected
literal length is at least this information.  Its length is precisely
$A_p+\Sigma_p$ up to fields already included in the self-delimiting
serializers.  Maximizing over inputs proves \eqref{eq:w1-pass-lower}; the
cap and tuned-$\epsilon$ consequences follow.
\end{proof}

\begin{remark}[What is and is not optimal]
\label{rem:w1-optimality-scope}
Equation~\eqref{eq:w1-output-lower} is an output-information bound and
\eqref{eq:w1-frontier-box} is a transcript bound; neither converts
self-contained circuit bits into gates.  The factor $2p-1$ counts the
actual alternating crossings of a forward scan.  A random-access input
oracle, an uncharged rereadable file, or reverse scans would be a different
model.  For $p=1$ and zero error, \cref{thm:w1-one-pass-mask} strengthens
$m\log_2K$ to $m\log_2(K+1)$.  No such strengthening is claimed for
randomized or general multipass protocols.
\end{remark}

%% file: theory/matching_upper_bounds.tex
\subsection{Memory-capped constructions}
\label{sec:w1-upper}

\begin{proposition}[Block hybrid]
\label{prop:w1-block-hybrid}
Fix $p,h\geq1$ and let $q=\min\{m,ph\}$.  There is a deterministic
zero-error $p$-pass compiler for \cref{def:w1-dispersed-stream} that
aggregates $q$ coordinates and passes the other update records through in
the public $p$-block hybrid mode of $E_{\mathrm{rel}}$.  It obeys
\begin{align}
 S&\leq \left\lceil\frac qp\right\rceil\log_2Q
       +O(\log(mQrp)),
 \label{eq:w1-upper-state}\\
 A_p&\leq(r-1)(m-q)\log_2Q+O(r\log(mQrp)),
 \label{eq:w1-upper-A}\\
 B_{\mathrm{rel}}^{\mathrm{out}}
 &\leq q\log_2K+r(m-q)\log_2Q
       +O(r\log(mQrp)).
 \label{eq:w1-upper-output}
\end{align}
In the compact-output corner $h=\lceil m/p\rceil$, this becomes
\begin{align}
 B_{\mathrm{rel}}^{\mathrm{out}}
 &\leq m\log_2K+O(\log(mQrp)),\\
 S&\leq\left\lceil\frac mp\right\rceil\log_2Q
       +O(\log(mQrp)).
 \label{eq:w1-compact-upper}
\end{align}
For two rounds, the construction satisfies
$A_p+pS\leq m\log_2Q+O(p\log(mQrp))$ after using a smaller final block.
Compared with
\eqref{eq:w1-frontier-box} at $\delta=0$, it is within the constant
$(2p-1)\log_2(K+1)/(p\log_2K)<2\log_2 3$ for every $K\geq2$.
\end{proposition}

\begin{proof}
Partition the first $q$ coordinates into at most $p$ consecutive blocks of
size at most $h$.  On pass $t$, keep only the residues of block $t$ and add
their first $r-1$ rounds modulo $Q$; when the final-round mates arrive,
emit the corresponding aggregate base-$K$ block.  The live packed block,
indices, and framing give \eqref{eq:w1-upper-state}.  On the first pass,
emit every share of each unselected coordinate in the declared round-major
order; these rotations already have the correct product and require no
retained coefficient table.  The selected coordinates can be balanced
across the $p$ blocks, so the maximum block has size $\lceil q/p\rceil$.
Their consecutive boundaries are a public function of $(m,p,q)$, so the
header stores those integers rather than $p$ separate delimiters.
The prefix contains $r-1$ residues and the
whole output contains $r$ residues for each such coordinate, giving
\eqref{eq:w1-upper-A} and \eqref{eq:w1-upper-output}.  If $ph\geq m$ there
are no passthrough coordinates.  For $r=2$, $(m-q)\log_2Q+ph\log_2Q$ is
at most $m\log_2Q$ up to the last partial block and framing.  The constant
comparison uses $(2p-1)/p<2$ and
$\log_2(K+1)/\log_2K\leq\log_2 3$.
\end{proof}

\begin{proposition}[Semantic and randomized upper corners]
\label{prop:w1-semantic-upper}
Under the same public dictionary, a one-pass semantic compiler receiving
$x$ emits the packed base-$K$ aggregate code with
\begin{equation}
 B_{\mathrm{rel}}^{\mathrm{out}}
 \leq\lceil m\log_2K\rceil+O(\log(mK)),
 \qquad S=O(\log(mK)),
 \label{eq:w1-semantic-upper}
\end{equation}
and implements the same $U_x$ as the flat stream.  For any
$0\leq\delta<1/2$, an input-independent abort coin that runs this exact
compiler with probability $1-\delta$ and otherwise emits a fixed short
circuit has worst-input success at least $1-\delta$ and expected relative
output
\begin{equation}
 (1-\delta)\lceil m\log_2K\rceil+O(\log(mK)+1).
 \label{eq:w1-random-upper}
\end{equation}
Thus the randomized family-relative output lower bound is matched up to
framing and the binary-entropy additive term.  A self-contained
$R_Z$-circuit upper bound remains
$O(m(\log m+\log Q))$ bits and at most $m$ rotations; no matching
self-contained-bit or gate-count claim is made.
\end{proposition}

\begin{proof}
Pack the ordered aggregate vector as one base-$K$ integer and retain only
its streaming index and framing state.  Proposition~\ref{prop:w1-stream-equivalence}
gives semantic equality with the flat representation.  Prefix the output
with the abort-coin branch tag for the randomized construction and take
expectations.  The self-contained construction writes each nonzero
$R_{P_j}(2\pi x_j/Q)$ with its address and rational parameter.
\end{proof}

%% file: theory/certificate_soundness.tex
Let

\begin{equation*}
 \mathcal C_n=\langle H,S,S^\dagger,X,Y,Z,\mathrm{CX},
 \mathrm{CZ},\mathrm{SWAP}\rangle
\end{equation*}
be the Clifford gate language on addressed qubits.  Fix a declared finite
formal-symbol set $\mathcal S$ and the exact coefficient module
\begin{equation*}
 \Lambda=\mathbb Q\pi\oplus
 \bigoplus_{s\in\mathcal S}\mathbb Q\theta_s.
\end{equation*}
For a signed Hermitian $n$-qubit Pauli $P$ and $\lambda\in\Lambda$, write
$R_P(\lambda)=\exp(-i\lambda P/2)$.

\begin{definition}[Exact commuting-Pauli/Clifford certificate domain]
\label{def:certificate}
The domain $\mathcal D_{\rm CPC}$ consists of finite addressed circuits over
$\mathcal C_n$, exact records $R_Z(\lambda)$ and $R_P(\lambda)$, identity and
barrier records, and an arbitrary exact global phase.  Every coefficient is
a reduced rational $\pi$ component plus a sorted sparse list of reduced
rational formal-symbol components carried by the input; a floating-only
parameter is not in the exact domain.

Scan a circuit in execution order while maintaining a Clifford frame $F$ and
the invariant
\begin{equation}
 U_{\rm prefix}=F R_{P_t}(\lambda_t)\cdots R_{P_1}(\lambda_1)
 \quad\text{up to global phase}.
 \label{eq:w6-frame-invariant}
\end{equation}
The frame is represented by the complete signed inverse-conjugation tableau
$\{F^\dagger X_jF,F^\dagger Z_jF\}_{j=1}^n$.  Appending a physical rotation
$R_Q(\lambda)$ appends $F^\dagger QF$ to the Pauli list.  The certificate
predicate additionally requires all extracted $P_j$ to commute pairwise.

Canonicalization moves a Pauli sign into its coefficient, discards an
identity-axis rotation, adds coefficients on identical unsigned supports,
reduces the rational $\pi$ coefficient modulo $2$ to the interval $[0,2)$,
keeps every formal-symbol coefficient exact, deletes zero records, and sorts the remaining records lexicographically by
their binary $(x,z)$ masks.  Circuit and Clifford global phases are omitted.
The resulting record is
\begin{align}
 \beta_j&:=(b_{j,s})_{s\in\mathcal S}\in\mathbb Q^{\mathcal S},
 &\bar a_j&\in\mathbb Q\cap[0,2),\nonumber\\
 \operatorname{can}(C)&=
 \bigl(n,T_F,((x_j,z_j,\bar a_j,\beta_j))_{j=1}^{\ell}\bigr).
 \label{eq:w6-canonical-record}
\end{align}
The symbolic checker accepts exactly when the two canonical byte records are
equal.
\end{definition}

\begin{proposition}[Soundness of canonical symbolic acceptance]
\label{prop:symbolic-soundness}
For $C,C'\in\mathcal D_{\rm CPC}$ satisfying the pairwise-commutation
predicate,
\begin{equation*}
 \operatorname{can}(C)=\operatorname{can}(C')
 \quad\Longrightarrow\quad
 U_C=e^{i\phi}U_{C'}
\end{equation*}
for some $\phi\in\mathbb R$.
\end{proposition}

\begin{proof}
Induction over the input records proves \eqref{eq:w6-frame-invariant}.
For a Clifford record $G$, replace $F$ by $GF$.  For a rotation,
\begin{equation*}
 R_Q(\lambda)F=F R_{F^\dagger QF}(\lambda),
\end{equation*}
which is exactly the tableau update used by the checker.  Thus the extracted
record implements the original circuit up to its explicitly ignored global
phase.

Pairwise commutation permits deterministic reordering without changing the
unitary.  Sign normalization is exact because
$R_{-P}(\lambda)=R_P(-\lambda)$, while coefficient reduction changes a factor only
by
\begin{equation*}
 R_P(\lambda+2k\pi)=(-1)^kR_P(\lambda).
\end{equation*}
All such signs multiply to a global phase; an
identity-axis rotation is also a global phase.  Equality of the sorted
rotation records, including every formal-symbol component, therefore gives
equal rotation products projectively.

Finally, equality of the signed tableaus means that
$F'F^\dagger$ commutes with every Pauli generator.  The Pauli representation
generates the full matrix algebra, so $F'F^\dagger$ is a scalar.  Hence the
two Clifford frames also agree up to global phase, proving the claim.
\end{proof}

\begin{remark}[Status and completeness boundary]
\label{rem:w6-certificate-status}
Only status \texttt{completed\_valid}, accompanied by the content-addressed
canonical records and an independent certificate payload, is an acceptance.
A canonical mismatch is \texttt{completed\_invalid}; an unlisted gate or
malformed address is \texttt{unsupported}; failure of pairwise commutation is
\texttt{predicate\_error}; and a floating coefficient without an exact
sidecar, an unassigned formal symbol in the dense checker, or a dense distance
inside the numerical guard band is
\texttt{numerical\_inconclusive}.  Only \texttt{completed\_valid} rows enter
quality means.  The proposition is a soundness statement, not a claim that
every equivalent Clifford-plus-rotation circuit has the same record.
\end{remark}

For $n\leq6$, the audit independently forms dense matrices and checks the
normalized Frobenius projective distance
\begin{equation*}
 d_{\rm F,proj}(U,V)=2^{-n/2}\min_{|\omega|=1}\|U-\omega V\|_{\rm F}.
\end{equation*}
The frozen suite covers both accepted and rejected pairs at every width
$1\leq n\leq6$ and changes, separately, an angle, a sign, one support bit,
a qubit permutation, the Clifford frame, and a deleted or duplicated gate.
The executable definition and status contract are
\path{certificates/CERTIFICATE_SPEC.md}; the canonicalizer and checker are
\path{certificates/canonicalize.py} and
\path{certificates/symbolic_checker.py}.

%% file: generated/qre_main_slice.tex
External semantic-capable (TKET) & 29 & 798 & 258,152 & 16 & 13,317 & 3.44e+03 \\
Materialize-first & 23,992 & 2,679,110 & 468,968 & 16 & 57,966,528 & 2.72e+07 \\
Semantic-first reference & 22 & 1,800 & 305,000 & 16 & 31,675 & 9.66e+03 \\
\bottomrule

%% file: generated/natural_families.tex
controlled\_power\_ladder & 9.0 [5.7, 12.7] & yes \\
diagonal\_hamiltonian & 38.0 [25.9, 51.8] & yes \\
phase\_polynomial\_blocks & -24.3 [-33.8, -16.3] & no \\
qaoa\_ising & 44.7 [33.7, 56.6] & yes \\
qft\_arithmetic & 8.1 [1.9, 13.8] & no \\
qpe\_controlled\_powers & -12.0 [-22.0, -3.0] & yes \\
trotter\_commuting\_blocks & 33.0 [17.6, 48.2] & yes \\
\bottomrule

%% file: appendix/appendix_A.tex
\section{Supplementary Theory and Full Proofs}

\subsection{Operational Gap Notation}

For reference, the experimental gap can be written without adding another
formal theorem hypothesis. Let:
\begin{itemize}[leftmargin=1.5em]
    \item $L_B(C)$ be exact lowering into target basis $B$,
    \item $\kappa$ be a conservative structural cost,
    \item $\mathcal{T}_K(C)$ be the bounded set of semantic candidates produced
    by the semantic lift under budget $K$,
    \item $\mathfrak{A}$ be a configured post-lowering pipeline family used for
    an experiment.
\end{itemize}

Define the $K$-representation gap of $C$ by
\[
\mathrm{Gap}^{K}_{\mathrm{rep}}(C)
:=
\inf_{A \in \mathfrak{A}} \kappa(A(L_B(C)))
-
\inf_{T \in \mathcal{T}_K(C)} \kappa(L_B(T)).
\]
Define the operational recoverability gap by
\[
\mathrm{Gap}^{K}_{\mathrm{rec}}(C)
:=
\kappa(C_{\mathrm{flat}}^\star) - \kappa(C_{\mathrm{term}}^\star),
\]
where $C_{\mathrm{flat}}^\star$ is the best candidate recovered from the flat
basis-lowered representation and $C_{\mathrm{term}}^\star$ is the best
candidate recovered from the bounded semantic term representation.

The meaning is simple: semantics may survive lowering, while cheap
recoverability of the opportunity may not. This notation is operational and
experimental; the formal lower bounds in the main text are the
budget-explicit output and cut-transcript theorems for the declared
finite-window stream model.

\begin{definition}[Aggregate recovery]
\label{def:aggregation-equivalent-recovery}
Consider the Fourier-layer family
$C_{m,r}=H_{\Lambda}D_m(\boldsymbol{\theta})^rH_{\Lambda}$ with
$D_m(\boldsymbol{\theta})=\prod_{j=1}^m e^{-i\theta_jP_j/2}$. A compiler
performs \emph{aggregation-equivalent recovery} on this family if, for each
non-removable term $P_j$, its successful output can be associated with an
$O(1)$-size record, subcircuit, or algebraic object whose coefficient is
equivalent modulo $2\pi$ to the aggregate $r\theta_j$. The object may be
explicit, as in a phase-polynomial or commuting-diagonal IR, or implicit, as
in a specialized rewrite, lookup, ZX transformation, or synthesis routine that
computes the same aggregate diagonal action without storing all $r$ lowered
occurrences separately.
\end{definition}

\begin{proposition}[Operational recovery channels]
\label{prop:classification-boundary}
For the Fourier-layer family, a pipeline that recovers the constant
fixed-width output can be classified by the channel through which it carries
the aggregate coefficient information. Examples include:
\begin{enumerate}[leftmargin=1.5em]
    \item store enough persistent summary information to carry the aggregate;
    \item lift enough lowered volume into a richer IR to reconstruct the
    aggregate;
    \item use nonlocal provenance information from before basis lowering that
    already contains the aggregate;
    \item reconstruct a global semantic representation equivalent to a
    phase-polynomial, ZX, stabilizer-phase, or commuting-diagonal Hamiltonian
    object.
\end{enumerate}
Consequently, a successful future flat-looking implementation would not by
itself refute the main theorem. It would instead classify the implementation by
identifying where the aggregate information crosses the finite-state boundary,
or by showing that the formal transducer model is missing a cheaper channel.
\end{proposition}

\begin{proof}[Justification]
The fixed-width theorem lower-bounds final output length for constant-budget
finite-window transducers on a certified noncolliding family. A pipeline that
achieves the semantic constant output must still represent the aggregate
coefficients that distinguish the noncolliding cases. The listed channels are
implementation-level ways to carry that information. This is a classification
statement rather than an optimality claim.
\end{proof}

\subsection*{Additional Structural Families}

The artifact also targets phase-ladder, mirrored, and conjugation patterns that
appear in QPE/QFT-style workloads and Grover-style workloads. These families
are not used here as additional theorem instances. They are best read as
structural stress families: the semantic controller preserves repeated
phase-heavy ladders, mirrored shells, and conjugation shells long enough to
compare projected costs before materializing the selected candidate.

For phase-ladder workloads, the relevant empirical question is whether bounded
semantic terms preserve a repeated commuting-diagonal opportunity that is
harder to recover after early lowering. For mirrored and conjugation workloads,
the relevant empirical question is whether the artifact recognizes repeated
shell structure and avoids materializing every candidate before selection. The
Grover and mirrored-shell experiments in Appendix~C are therefore empirical
positive cases and robustness checks, not support for a separate theorem-level
family separation.

\subsection{Self-Contained Rational CP Period Certificate}
\label{app:sec:rational-period}

The canonical four-qubit rational witness has diagonal block
\begin{equation*}
\begin{aligned}
D={}&\prod_{j=0}^{3}RZ_j(\theta_j)
 \prod_{0\leq a<b\leq3}CP_{ab}(\theta_{ab}),\\
\boldsymbol\theta_{RZ}/\pi
={}&(1/7,1/8,1/9,1/10),\\
(\theta_{01},\theta_{02},\theta_{03})/\pi
={}&(1/12,1/13,1/14),\\
(\theta_{12},\theta_{13},\theta_{23})/\pi
={}&(1/14,1/15,1/16).
\end{aligned}
\end{equation*}
For a reduced rational angle $\theta/\pi=p/q$, both $RZ(\theta)$ and
$CP(\theta)$ have projective single-term order
\begin{equation*}
 \operatorname{ord}(p/q)
 =\frac{2q}{\gcd(|p|,2q)},
\end{equation*}
because the relevant relative phase is $p/q$ modulo $2$.  The complete
single-term certificate is
\begin{center}
\begin{tabular}{lccc}
\toprule
term & support & $\theta/\pi$ & $T_j$\\
\midrule
$RZ_0$ & $0$ & $1/7$ & $14$\\
$RZ_1$ & $1$ & $1/8$ & $16$\\
$RZ_2$ & $2$ & $1/9$ & $18$\\
$RZ_3$ & $3$ & $1/10$ & $20$\\
$CP_{01}$ & $01$ & $1/12$ & $24$\\
$CP_{02}$ & $02$ & $1/13$ & $26$\\
$CP_{03}$ & $03$ & $1/14$ & $28$\\
$CP_{12}$ & $12$ & $1/14$ & $28$\\
$CP_{13}$ & $13$ & $1/15$ & $30$\\
$CP_{23}$ & $23$ & $1/16$ & $32$\\
\bottomrule
\end{tabular}
\end{center}
Hence
\begin{equation*}
\begin{aligned}
T_{\mathrm{terms}}
&=\operatorname{lcm}(14,16,18,20,24,26,28,28,30,32)\\
&=2^5\cdot3^2\cdot5\cdot7\cdot13=131040,
\end{aligned}
\end{equation*}
which proves that $D^{131040}$ is a global phase.  To certify minimality
for the product rather than merely obtain this common upper bound, subtract
the phase of $0000$.  For $z=z_0z_1z_2z_3$, the exact relative phase is
\begin{equation*}
\begin{aligned}
w(z)={}&z_0/7+z_1/8+z_2/9+z_3/10\\
 &+z_0z_1/12+z_0z_2/13+z_0z_3/14\\
 &+z_1z_2/14+z_1z_3/15+z_2z_3/16
 \pmod 2.
\end{aligned}
\end{equation*}
In particular,
\begin{equation*}
\begin{aligned}
 w(1011)&=\frac{37007}{65520},\\
 \operatorname{ord}(w(1011))
 &=\frac{2\cdot65520}{\gcd(37007,2\cdot65520)}\\
 &=131040.
\end{aligned}
\end{equation*}
Thus no smaller positive power of $D$ is a global phase, and
$T=131040$ exactly.  It follows that
$D^1,\ldots,D^R$ are pairwise distinct up to global phase for every
$R<T$.

The machine certificate uses only Python \texttt{fractions.Fraction} and
prints all 16 relative numerators, denominators, and orders.  The frozen
integer output contains
\begin{equation*}
\begin{gathered}
\texttt{term\_order\_lcm=131040},\\
\texttt{relative\_order\_lcm=131040},\\
\texttt{minimality\_state\_1011\_order=131040},\\
\texttt{exact\_global\_phase\_period\_T=131040}.
\end{gathered}
\end{equation*}
The reproducibility files and hashes are listed below; each digest is split
at its 32-hex midpoint only for typesetting.
\begin{itemize}[leftmargin=1.5em]
 \item verifier: \path{research/rational_witness/verify_period.py},
 SHA256 \texttt{c7c90f26e00e515383e94e4bbd707574}\allowbreak
 \texttt{eb3aabf94ad71c3cd5e2cdb8d5db7e76};
 \item integer output:
 \path{research/rational_witness/period_certificate.txt}, SHA256
 \texttt{375295d776f21ea2a82ee4f2bab5205d}\allowbreak
 \texttt{25bfef85ee864900baacfbd80be9bd6e};
 \item manifest: \path{research/rational_witness/SHA256SUMS}.
\end{itemize}
From the repaired-manuscript root, the byte-for-byte check is
\begin{quote}
\footnotesize
\texttt{cd research/rational\_witness}\\
\texttt{python3 verify\_period.py \textbackslash}\\
\texttt{\quad --check period\_certificate.txt}
\end{quote}

\subsection{Full Proofs for the Transducer Lower Bound}
\label{app:sec:full-proofs}

This section supplies the full finite-state proof layer underlying
\Cref{thm:final-output,cor:finite-tail,cor:counting-toll}. The main text states
the shorter version; the appendix keeps the pumpability data explicit.
The following statements restate the main-text theorem and corollaries with
full hypotheses and complete proofs; they are not additional independent
results.

\subsubsection{Pumpable Pair Families}

Persistent store can depend on $r$, so the induction tracks both a store value
and a tape word. In this subsection, $D$ in pumpability data denotes a period
modulus, not the diagonal witness layer in $C_r=H_\Lambda D^rH_\Lambda^{-1}$.
The next definition restates the main-text pumpable-pair definition with
explicit residue data for the full proof.

\begin{definition}[Pumpable pair family, restated with explicit residue data]
\label{app:def:pumpable}
Let $\Gamma$ be a finite alphabet and $P$ a finite set of store values. A family
\[
  X(r)=(\beta(r),W(r))\in P\times\Gamma^*,\qquad r\in\mathbb{N},
\]
is \emph{pumpable} if there exist integers $h\ge0$ and $D\ge1$ such that for
every residue $\rho\in\{0,\ldots,D-1\}$ there are
\[
  \beta_\rho\in P,
  \qquad
  A_\rho,B_\rho,C_\rho\in\Gamma^*,
\]
with
\[
  X(h+\rho+kD)=(\beta_\rho,A_\rho B_\rho^k C_\rho)
\]
for every $k\ge0$. The word $B_\rho$ is the pump word on residue $\rho$. The
residue is \emph{live} if $B_\rho\ne\epsilon$ and \emph{collapsed} if
$B_\rho=\epsilon$.
\end{definition}

\begin{lemma}[Input pumpability]
\label{app:lem:input-pumpable}
For the fixed-width witness $T_0(r)=\Pi M^r\Sigma$, the pair family consisting
of the initial persistent store and $T_0(r)$ is pumpable with $h=0$, $D=1$,
pump word $M$, and fixed store value.
\end{lemma}

\begin{proof}
Take $A_0=\Pi$, $B_0=M$, and $C_0=\Sigma$.
\end{proof}

\subsubsection{One-Pass Propagation}

Fix a precision bound $q$ and a charged alphabet $\Gamma=\Gamma_B(q)$. Under
this bound, a pass of $\mathcal{A}$ can be viewed as an ordinary deterministic
subsequential transducer over a finite alphabet. Its finite state consists of
the ephemeral state, the persistent store, and the $w$-token window. The
end-of-pass persistent store is one component of the final finite state.

\begin{lemma}[One-pass pumpability propagation]
\label{app:lem:one-pass-pump}
Let $X(r)=(\beta(r),W(r))$ be a pumpable pair family over a finite input
alphabet. Let $\tau$ be one deterministic finite-window, finite-burst pass whose
output alphabet and store set are finite. Let
\[
  \tau(\beta(r),W(r))=(\beta'(r),W'(r))
\]
be the resulting final store and output tape. Then
$X'(r)=(\beta'(r),W'(r))$ is pumpable.
\end{lemma}

\begin{proof}
Fix pumpability data
$h,D,\{\beta_\rho,A_\rho,B_\rho,C_\rho\}_\rho$ for $X$. We first treat one
residue $\rho$, then recombine all residues into $r$-indexed pumpability data
for $X'$.

For a fixed residue $\rho$, write
\[
  W_k=A B^k C,
  \qquad
  \beta_k=\beta,
\]
where $k\ge0$ and $A,B,C,\beta$ abbreviate the residue data. Absorb the
ephemeral state, persistent store, and window into the pass's finite state set
$S$. Starting from the state determined by $\beta$, process $A$. This emits a
fixed word $O_A$ and reaches a state $s_A\in S$.

Let $f:S\to S$ be the state map obtained by reading the block $B$, and let
$o(s)\in(\Gamma')^*$ be the output emitted while reading $B$ starting in state
$s$. The sequence
\[
  s_A,\ f(s_A),\ f^2(s_A),\ldots
\]
in the finite set $S$ is ultimately periodic. Thus there exist $u\ge0$ and
$\ell\ge1$ with $u+\ell\le |S|+1$ such that
\[
  f^{t+\ell}(s_A)=f^t(s_A)
  \qquad\text{for all }t\ge u.
\]
For $k=u+v\ell+\eta$ with $0\le\eta<\ell$, the output while reading $B^k$ is
\[
\begin{aligned}
  &\bigl(o(s_A)o(f(s_A))\cdots o(f^{u-1}(s_A))\bigr)\\
  &\quad\cdot\Bigl(o(f^u(s_A))\cdots
  o(f^{u+\ell-1}(s_A))\Bigr)^v\\
  &\quad\cdot\bigl(o(f^u(s_A))\cdots
  o(f^{u+\eta-1}(s_A))\bigr).
\end{aligned}
\]
After this, the pass reads the fixed suffix $C$, emits a word depending only
on the state $f^{u+\eta}(s_A)$, and ends in a persistent store value depending
only on the same state.

Restore the residue subscripts. For each original residue
$\rho\in\{0,\ldots,D-1\}$, let $u_\rho\ge0$ and $\ell_\rho\ge1$ be the
preperiod and cycle length obtained above, let $f_\rho$, $s_{A_\rho}$, and
$O_{A_\rho}$ be the corresponding state map, post-prefix state, and prefix
output. For $0\le\eta<\ell_\rho$, let $F_{\rho,\eta}$ be the complete word
emitted while reading the fixed suffix $C_\rho$ from state
$f_\rho^{u_\rho+\eta}(s_{A_\rho})$, including the end-of-tape finalizer output,
and define
\begin{align*}
  P_\rho &:=
  o\bigl(s_{A_\rho}\bigr)o\bigl(f_\rho(s_{A_\rho})\bigr)\cdots
  o\bigl(f_\rho^{u_\rho-1}(s_{A_\rho})\bigr),\\
  Q_\rho &:=
  o\bigl(f_\rho^{u_\rho}(s_{A_\rho})\bigr)\cdots
  o\bigl(f_\rho^{u_\rho+\ell_\rho-1}(s_{A_\rho})\bigr),\\
  R_{\rho,\eta} &:=
  o\bigl(f_\rho^{u_\rho}(s_{A_\rho})\bigr)\cdots
  o\bigl(f_\rho^{u_\rho+\eta-1}(s_{A_\rho})\bigr)\cdot F_{\rho,\eta}.
\end{align*}
Let $\beta'_{\rho,\eta}$ be the end-of-pass persistent store value reached from
that state. Define
\[
\begin{aligned}
  u^*&:=\max_{\rho} u_\rho,
  &\ell^*&:=\lcmop_{\rho}\ell_\rho,\\
  h'&:=h+D u^*,
  &D'&:=D\ell^*.
\end{aligned}
\]
Each new residue $\rho'\in\{0,\ldots,D'-1\}$ decomposes uniquely as
\[
  \rho'=\rho+D\tilde\eta,
  \qquad
  \rho\in\{0,\ldots,D-1\},\quad
  \tilde\eta\in\{0,\ldots,\ell^*-1\}.
\]
For $r=h'+\rho'+k'D'$ with $k'\ge0$, the old representation
$r=h+\rho+kD$ has
\[
  k=u^*+\tilde\eta+k'\ell^*.
\]
Since $u^*\ge u_\rho$, we have $k\ge u_\rho$ for every $k'\ge0$; and since
$\ell_\rho\mid\ell^*$, the residue of $k-u_\rho$ modulo $\ell_\rho$ does not
depend on $k'$. Set
\[
  \eta:=(u^*+\tilde\eta-u_\rho)\bmod \ell_\rho,
  \qquad
  v_0:=\Bigl\lfloor\frac{u^*+\tilde\eta-u_\rho}{\ell_\rho}\Bigr\rfloor\ge0,
\]
so that
\[
  k=u_\rho+(v_0+k'\ell^*/\ell_\rho)\ell_\rho+\eta.
\]
By the output decomposition above, for every $k'\ge0$,
\[
  X'(h'+\rho'+k'D')
  =\bigl(\beta'_{\rho,\eta},\ A'_{\rho'}(B'_{\rho'})^{k'}C'_{\rho'}\bigr),
\]
where
\[
  A'_{\rho'}:=O_{A_\rho}P_\rho Q_\rho^{v_0},
  \qquad
  B'_{\rho'}:=Q_\rho^{\ell^*/\ell_\rho},
  \qquad
  C'_{\rho'}:=R_{\rho,\eta}.
\]
This is exactly pumpability of $X'$ with data $(h',D')$ and residue data
$\{\beta'_{\rho,\eta},A'_{\rho'},B'_{\rho'},C'_{\rho'}\}_{\rho'}$.
\end{proof}

\begin{corollary}[Per-tape induction, for proof]
\label{app:cor:per-tape}
Assume $\mathcal{A}$ has constant budget, so all executions use one finite
alphabet $\Gamma_B(q_*)$. For every pass index $i=0,1,\ldots,p$, the pair
family consisting of the persistent store after pass $i$ and the tape $T_i(r)$
is pumpable.
\end{corollary}

\begin{proof}
The base case is \Cref{app:lem:input-pumpable}. Apply
\Cref{app:lem:one-pass-pump} successively to passes $1,\ldots,p$.
\end{proof}

\begin{remark}[Unary scratch tapes]
\label{app:rem:unary}
If pass 1 writes one marker per input period, then
$T_1(r)=\mathtt{marker}^r$ is pumpable with live pump word
$\mathtt{marker}$. \Cref{app:lem:one-pass-pump} applies to every later
pass. A later constant-budget pass may output one marker every $c$ input
markers, permute markers, or add a fixed periodic decoration, but it remains
pumpable. It cannot turn $\mathtt{marker}^r$ into a family of $O(1)$ many
correct aggregate circuits for all noncolliding $r$; if the final pump
collapses, the final output repeats on a residue class.
\end{remark}

\subsubsection{Collapse Versus Final-Output Length}

\begin{lemma}[Pumpable outputs are bounded-or-linear on each residue]
\label{app:lem:bounded-or-linear}
Let $W(r)$ be a pumpable word family with data
$(h,D,A_\rho,B_\rho,C_\rho)$. For $r=h+\rho+kD$:
\begin{enumerate}[label=(\alph*)]
\item if $B_\rho=\epsilon$, then $W(h+\rho+kD)$ is independent of $k$;
\item if $B_\rho\ne\epsilon$, then
\[
  |W(h+\rho+kD)|\ge k.
\]
\end{enumerate}
Consequently, on a live residue,
\[
  |W(r)|\ge \frac{r-h-D+1}{D}
\]
for all $r\ge h+D$ in that residue.
\end{lemma}

\begin{proof}
Both statements are immediate from
$W(h+\rho+kD)=A_\rho B_\rho^k C_\rho$. If $B_\rho\ne\epsilon$, then
$|B_\rho|\ge1$, so $|A_\rho B_\rho^k C_\rho|\ge k$. Since
$r=h+\rho+kD$ and $0\le\rho<D$, $k\ge(r-h-D+1)/D$.
\end{proof}

\begin{lemma}[Correct noncolliding final outputs cannot collapse]
\label{app:lem:no-collapse}
Assume finite-range noncollision (Assumption~\ref{ass:noncollision}) with
$N(R)=R$ on the interval $1\le r\le R$. Suppose the final output family
$T_p(r)$ is pumpable with data $(h,D,A_\rho,B_\rho,C_\rho)$ and
$\mathcal{A}$ is correct on every $r\in\{1,\ldots,R\}$. Without loss of
generality take $h\ge1$: if $h=0$, replace $h$ by $h+D$; this preserves
pumpability after discarding the finite prefix of indices $k=0$ on each residue
and absorbing one copy of $B_\rho$ into $A_\rho$. Then every residue class
appearing twice in the tail $\{h,h+1,\ldots,R\}$ is live. In particular, for
every $r\in\{h+D,\ldots,R\}$,
\[
  |T_p(r)|\ge \frac{r-h-D+1}{D}.
\]
\end{lemma}

\begin{proof}
Let $r=h+\rho+kD$ with $k\ge1$ and $r\le R$. Then
$r-D=h+\rho+(k-1)D$ also lies in the tail, and by the convention $h\ge1$ it
satisfies $r-D\ge h\ge1$, so both $r$ and $r-D$ lie in the correctness domain
$\{1,\ldots,R\}$. If $B_\rho=\epsilon$,
\Cref{app:lem:bounded-or-linear} gives
\[
  T_p(r)=T_p(r-D)
\]
as literal circuit strings. Hence the two outputs have the same unitary.
Correctness at both $r$ and $r-D$ would imply
\[
U(\mathcal{L}_B(C_r))=U(\mathcal{L}_B(C_{r-D}))
\]
up to global phase, which by exactness of the lowering means $D^r\equiv
D^{r-D}$ up to global phase, contradicting Assumption~\ref{ass:noncollision}
with $N(R)=R$. Therefore $B_\rho\ne\epsilon$, and
\Cref{app:lem:bounded-or-linear} gives the stated length bound.
\end{proof}

\subsubsection{Final-Output Theorem and Corollaries}

\begin{theorem}[$p$-pass final-output lower bound, restated with full hypotheses]
\label{app:thm:final-output}
Fix a fixed-width diagonal witness family
$C_r=H_\Lambda D^rH_\Lambda^{-1}$ with exact lowering
$\mathcal{L}_B(C_r)=\Pi M^r\Sigma$, and assume unbounded noncollision
(Assumption~\ref{ass:unbounded-noncollision}): the powers $D^r$ are pairwise
distinct up to global phase for all $r\ge1$, so $N(R)=R$ on every range. Let
$\mathcal{A}$ be a uniform deterministic $p$-pass finite-window, finite-burst
transducer with constant budget that is correct on the input
$\mathcal{L}_B(C_r)=\Pi M^r\Sigma$ for every $r\ge1$.
Here uniform means one fixed deterministic transducer, with the same passes,
transition rules, finite state, window, burst bound, and store bound for all
repetition counts $r$.

Then there exist constants $h\ge1$ and $D\ge1$, depending on $\mathcal{A}$ and
the fixed witness but not on $r$, such that for every $r\ge h+D$,
\[
  |T_p(r)|\ge \frac{r-h-D+1}{D}=\Omega_{\mathcal{A}}(r).
\]
Equivalently, no constant-budget transducer in this model that is correct on
the whole noncolliding fixed-width family can emit $o(r)$-length final outputs
on that family.
\end{theorem}

\begin{proof}
By constant budget, all executions use a finite charged token alphabet.
\Cref{app:cor:per-tape} gives pumpability of the final pair family and
hence of $T_p(r)$, with data $(h,D)$ depending only on $\mathcal{A}$ and the
witness. By the convention of \Cref{app:lem:no-collapse}, take $h\ge1$,
so every comparison point $r-D\ge h\ge1$ used below lies in the correctness
domain.

Fix any $r\ge h+D$ and any $R\ge r$. By Assumption~\ref{ass:unbounded-noncollision},
$N(R)=R$, and $\mathcal{A}$ is correct on every $r'\in\{1,\ldots,R\}$. Thus
\Cref{app:lem:no-collapse} applies and gives the displayed bound at $r$.
In particular every final residue is live: if any final residue collapsed, two
sufficiently large inputs in that residue would receive the same final output
string, contradicting correctness and noncollision.
\end{proof}

\paragraph{Correctness quantifier.}
Correctness in \Cref{app:thm:final-output} is required at every $r\ge1$; the
conclusion then holds at every $r\ge h+D$, without exception. This global
quantifier is not decorative: the collapse argument in
\Cref{app:lem:no-collapse} compares the outputs at $r$ and $r-D$ and
needs correctness at both points. Correctness at isolated points is not
sufficient; if $\mathcal{A}$ were correct at $r$ but not at $r-D$, identical
final strings at those two inputs would yield no contradiction.

\begin{corollary}[Finite-range tail form, restated with full hypotheses]
\label{app:cor:finite-tail}
Let $\mathcal{A}$ be a constant-budget transducer, and let $(h,D)$ be
pumpability data for its final pair family as provided by
\Cref{app:cor:per-tape}. Suppose finite-range noncollision
(Assumption~\ref{ass:noncollision}) holds with $N(R)=R$ on the range
$1\le r\le R$ and $\mathcal{A}$ is correct on every
$r\in\{1,\ldots,R\}$. Then
\[
  |T_p(r)|\ge \frac{r-h-D+1}{D}
  \qquad\text{for every }h+D\le r\le R.
\]
The constants $(h,D)$ are determined by $\mathcal{A}$ and the witness family
alone; in particular they do not depend on $R$. No statement is made for
transducers outside the constant-budget regime: there, pumpability data in the
sense of \Cref{app:def:pumpable} need not exist, and any finite-range
inequality would require explicit control of pumping constants not supplied by
this theorem.
\end{corollary}

\begin{proof}
Constant budget puts all executions over one finite charged alphabet, so
\Cref{app:cor:per-tape} supplies pumpability data $(h,D)$ for the
final tape. The construction of this data uses only the fixed transducer, its
finite alphabets and state sets, and the fixed witness words $\Pi,M,\Sigma$;
it is obtained before any finite endpoint $R$ is chosen. Hence $(h,D)$ are
independent of $R$. \Cref{app:lem:no-collapse} then applies verbatim on
the range $1\le r\le R$.
\end{proof}

\begin{corollary}[Bit-level output counting toll, restated for proof]
\label{app:cor:counting-toll}
Assume finite-range noncollision supplies $N(R)$ pairwise distinct target
unitaries, up to global phase, on a range. Suppose a deterministic exact
compiler emits final outputs $T_p(r)$. Then
\[
  \max_{1\le r\le R}B_{\mathrm{self}}(T_p(r))
  \ge \left\lceil\log_2 N(R)\right\rceil.
\]
The same inequality holds for a fixed family-relative code when one public
dictionary is shared by the whole range.  For the rational-angle CP witness,
only the displayed per-range inequality is claimed on fixed ranges
$R<T=131040$.
\end{corollary}

\begin{proof}
A circuit string determines its unitary, so a deterministic exact compiler
that is correct on two inputs whose target unitaries are distinct up to global
phase must emit two distinct output strings. The noncollision hypothesis
therefore forces at least $N(R)$ pairwise distinct codewords.  If their
prefix-free lengths are $\ell_1,\ldots,\ell_{N(R)}$, Kraft's inequality
gives $\sum_i2^{-\ell_i}\leq1$, so
$N(R)2^{-\max_i\ell_i}\leq1$.  Hence the largest length is at least
$\lceil\log_2N(R)\rceil$.  The argument applies independently to either
fixed code and does not convert token count or gate count into bits.
\end{proof}

\subsection{Bookkeeping for the Masked-Share Cut}
\label{app:sec:streaming-proof}

This subsection records the restart property used by the masked-share
tradeoff (\cref{thm:main-tradeoff}).  It also makes explicit why a
materialized update log, phase table, graph, RAM/cache layout, old-input
spool, dynamic-pass state, or random seed is part of the cut message rather
than an uncharged reconstruction channel.

\begin{lemma}[Restart from the serialized cut]
\label{app:lem:reconstruction}
Fix $\mathcal A$ as in \cref{def:compiler}.  At the inter-round cut
$c_{\mathrm{MS}}$, the five fields in \eqref{eq:info-budget}, together
with any prefix already committed to the final write-once output, determine
the continuation of that run for every supplied second-round continuation
$v$.  For a randomized run this statement is conditional on the serialized
seed/tape and cursor.
\end{lemma}

\begin{proof}
The control field restores the exact transition, all head positions,
dynamic pass/termination schedule, and random seed/tape cursor.  The store
and window fields restore RAM, stack, hash/cache metadata, nonparameter
machine data, and ordered live token structures.  The parameter field
restores every omitted coefficient by its location tag.  The IR field
restores every other live intermediate tape, DAG, table, addressed record,
and any consumed input block retained for rescanning.  These disjoint fields
therefore reconstruct the complete machine configuration.  Conditional on
the serialized random tape, the receiver can append the supplied second
round and execute the same dynamically chosen continuation.  The separately
transmitted final-output prefix is prepended to the suffix produced after
restart.  No additional state, old-input file, or public-family side channel
is used.
\end{proof}

%% file: appendix/appendix_B.tex
\section{Artifact Engineering}
The implementation described in this appendix is a bounded recognizer and
candidate controller, not a general maximal-commuting-region extractor for
arbitrary DAGs.  Its semantic nodes are limited to explicit supported
Pauli/phase records, local phase-ladder or Fourier spans,
mirrored/self-inverse composite blocks, and repeated-run or preset semantic
references.  The word ``region'' below means a contiguous span accepted by
that configured grammar under its scan order.  Soundness of a selected
rewrite is supplied by the independent certificate and conservative
fallback, not by a claim of global maximality.

\subsection{Candidate Control and Bounded Pipeline}

The semantic path is advisory, not authoritative: semantic terms are
lowered late, semantic subterms are compiled once and reused, and a
semantic candidate loses to a flat reference whenever it does not
clearly improve the conservative structural cost.  The term caches,
compiled-subterm caches, and repeated-run block caches in the
implementation follow from treating the semantic term rather than the
materialized circuit as the unit of comparison.  The whole method is the
following bounded pipeline.

\begin{enumerate}[leftmargin=1.8em]
    \item \textbf{Input.} A high-level circuit $C$, target basis $B$, and iteration budget $K$.
    \item Compute stable signatures for instructions and candidate blocks in $C$.
    \item Run bounded preprocessing $\mathcal{S}_K(C)$: repeated-prefix
    and repeated-run detection and reuse; adjacent
    block--inverse(block) cancellation; commutative inverse
    cancellation; semantic lifting into instruction-span, phase-ladder,
    and mirrored/self-inverse nodes where applicable; early stop if one
    full round does not improve the selected cost.
    \item Construct the candidate set
    \[
    \begin{aligned}
    \mathcal{C}(C) = \{&\mathcal{L}_B(C),
    \mathcal{P}_{\mathrm{ucc}}(\mathcal{L}_B(C)),\\
    &\mathcal{P}_{\mathrm{ucc}}(\mathcal{L}_B(\mathcal{S}_K(C))),
    \mathcal{R}(C)\},
    \end{aligned}
    \]
    where $\mathcal{R}(C)$ is an optional semantic or preset-reference
    candidate for large repeated structures, itself produced from the
    semantic term algebra (Fourier-layer lowering, mirrored/self-inverse
    lowering, or a bounded repeated-run rebuild).
    \item If the circuit factors as $C=P\cdot B^r\cdot S$, form the
    bounded block-candidate set $\mathcal{B}(B)=\{B_1,\dots,B_q\}$ and
    select $i^\star=\arg\min_i\kappa_{\mathrm{proj}}(P,B_i,r,S)$ on the
    projected full-circuit cost, rebuilding only $B_{i^\star}$.  This is
    what lets the repeated-run path avoid materializing several giant
    candidates merely to compare them.
    \item Evaluate all candidates under the conservative structural cost
    $\kappa$, with fixed lexicographic tie-breaking, and return
    $\widehat{C}=\arg\min_{X\in\mathcal{C}(C)}\kappa(X)$.
\end{enumerate}

The role of candidate control is not to dominate all downstream flows.
It is to prevent avoidable structured-circuit regressions by preserving
the best available output among a small family of semantically
equivalent candidates.

\subsection{Complexity Discussion}

The method is intentionally bounded and does not attempt global optimization. Let $n := |C|$ be the number of high-level instructions. Suppose one outer iteration consists of $r$ bounded scans, and scan $j$ has cost $T_j(n)$. Then the total runtime of preprocessing satisfies
\[
T_{\mathrm{pre}}(n)
\in
O\!\left(
K \cdot \sum_{j=1}^{r} T_j(n)
\right).
\]
The dominant scans are near-linear in $n$ up to implementation-dependent
hashing, block-comparison, and bookkeeping costs, so practical behavior
stays near-linear up to a bounded constant on repeat-heavy workloads.  A
large-circuit fast path handles repeat-dominant cases, source-level
short-circuiting avoids downstream work on Qiskit-native all-to-all
composites, and backend-aware portfolio expansion is restricted to cases
where instability matters.  The method is therefore a bounded systems
heuristic, not a globally optimizing pass.

\begin{proposition}[Expected regime of benefit]
The preprocessing stage is most likely to help on circuits whose useful structure is already explicit at the high-level instruction stream or in the bounded semantic lift, especially exact inverse blocks, repeated blocks, phase-ladder layers, mirrored/self-inverse shells, and dominant repeated prefixes. It is not expected to uniformly improve families whose performance is dominated by later synthesis or routing decisions that are not captured by the preprocessing signatures or semantic nodes.
\end{proposition}

\begin{proof}[Justification]
This follows from the construction of the pass. Every preprocessing rule
operates on high-level structural regularities detected through stable
signatures, exact inverse relations, or the bounded semantic lift into semantic
terms such as Fourier layers and mirrored/self-inverse shells. Therefore the
pass can only exploit structure that survives in this representation.
Conversely, if output quality is dominated by later basis-specific synthesis or
hardware-routing effects not visible at this stage, the preprocessing pass may
improve little or only preserve parity with stronger downstream baselines. The
representation-aware compilation principle is therefore selective rather than
universal: it claims that some opportunities are materially easier to preserve
and compare before lowering, not that all useful optimization should be moved
to the semantic layer.
\end{proof}

\subsection{Integration Into UCC}
\label{sec:integration}

The semantic artifact changes the UCC flow in six places:
\begin{enumerate}[leftmargin=1.5em]
    \item it performs structural simplification before basis lowering,
    \item it lifts suitable subcircuits into a bounded semantic term algebra before candidate generation,
    \item it compares candidate outputs conservatively to avoid structured-circuit regressions,
    \item it uses source-level short-circuiting on suitable all-to-all Qiskit-native cases,
    \item it restricts backend-aware portfolio expansion to cases where extra probing is justified,
    \item it memoizes bounded semantic references and repeated-run block selections for reuse.
\end{enumerate}

The organizing distinction is that some circuit growth is unavoidable
basis-translation overhead while some is structural regression caused by
losing recoverable semantic evidence before candidate selection.

\subsection{Method and Provenance Map}
\label{app:method-provenance}

Every frozen numerical row is addressed by the complete primary key
\begin{equation*}
\begin{split}
(&\texttt{experiment},\texttt{instance},\texttt{representation},\\
 &\texttt{method},\texttt{seed},\texttt{backend\_hash},\\
 &\texttt{tool\_version},\texttt{artifact\_commit},
  \texttt{config\_hash}).
\end{split}
\end{equation*}
The backend and configuration hashes are SHA256 digests of canonical JSON
(sorted keys and no insignificant whitespace).  A row collision on all nine
fields is rejected.  In particular, the reorder-probe hardware rows and the
canonical hardware-aware rows have different experiment and configuration
keys, and therefore different reader labels even when the underlying MQT
instance name is the same.  The frozen source is
\path{research/experiment_registry/frozen_hardware_rows.json};
\path{scripts/generate_keyed_tables.py} validates key uniqueness and generates
the three affected table fragments plus
\path{expanded_primary_keys.json}.  The generated labels expose the campaign,
instance, seed, and the first eight configuration-hash digits.  Hand-edited
numeric rows are not accepted by this path.

The canonical validation data are grouped by experiment family in the artifact.
The external-baseline method map immediately below uses the disclosed
\(\texttt{timeout\_s}=600\) second wall-clock budget per method per instance;
it does not override the distinct \texttt{RP23}, \texttt{HW24}, or
\texttt{REAL24} campaign manifests. The table maps paper-facing names to
JSON keys, tool/pass sequences, artifact flags, and source classes. The older
real-panel UCC default row is included only to define that frozen-control
label in the main text and Appendix~C; it is not one of the canonical methods
used in the 600\,s external-baseline comparison.

\begin{table*}[t]
\centering
\tiny
\setlength{\tabcolsep}{2pt}
\begin{tabularx}{\textwidth}{@{}>{\raggedright\arraybackslash}p{0.18\textwidth}
                  >{\raggedright\arraybackslash}p{0.12\textwidth}
                  >{\raggedright\arraybackslash}X
                  >{\raggedright\arraybackslash}p{0.16\textwidth}
                  >{\raggedright\arraybackslash}p{0.07\textwidth}
                  >{\raggedright\arraybackslash}p{0.10\textwidth}@{}}
\toprule
Paper name & JSON key & Tool/pass sequence & Artifact flags & Budget & Source classes \\
\midrule
semantic UCC (Fourier-layer IR enabled)
& \path{semantic_ucc}
& \path{ucc.compile} to \{\texttt{cx}, \texttt{rx}, \texttt{ry}, \texttt{rz}, \texttt{h}\} with Fourier-layer IR enabled
& \path{UCC_DISABLE_FOURIER_LAYER_IR} unset
& 600\,s
& main, abl, res, corr \\
upstream UCC v0.4.12 pre-semantic default
& \path{baseline_ucc}
& \path{ucc.compile} loaded from the legacy baseline repository in \path{research/compare_real_instances.py}; target basis \{\texttt{cx}, \texttt{rx}, \texttt{ry}, \texttt{rz}, \texttt{h}\}
& \href{https://github.com/unitaryfoundation/ucc/tree/1b4212a10082ce667c53b6a75f2cd3c19449dbed}{legacy baseline repo}; semantic candidate controls absent
& 240\,s in real-panel runner
& real, appx \\
artifact UCC (Fourier-layer IR disabled)
& \path{no_fourier_ucc}
& same artifact compile path with Fourier-layer IR disabled
& \path{UCC_DISABLE_FOURIER_LAYER_IR=1}
& 600\,s
& abl, res \\
\texttt{qiskit opt3}
& \path{qiskit_opt3}
& \path{qiskit.transpile(..., optimization_level=3, layout_method="trivial", routing_method="none")}
& none
& 600\,s
& main, abl, res \\
PyZX configured bridge
& \path{pyzx_configured}
& Qiskit QASM2 to PyZX, \path{to_basic_gates}, \path{basic_optimization}, Qiskit target-basis cleanup
& none
& 600\,s
& main \\
PyZX \texttt{full\_reduce}
& \path{pyzx_full_reduce}
& Qiskit QASM2 to PyZX graph, \path{full_reduce}, circuit extraction, Qiskit target-basis cleanup
& none
& 600\,s
& main \\
TKET FullPeephole
& \path{tket_fullpeephole}
& Qiskit-to-TKET bridge when available, otherwise QASM fallback, then \path{FullPeepholeOptimise}
& none
& 600\,s
& main \\
TKET PauliSimp (unrebased configuration)
& \path{tket_paulisimp_unrebased}
& direct Qiskit--TKET bridge, then \path{PauliSimp}; retained only to record the predicate error
& none
& 600\,s
& status \\
TKET PauliSimp (rebased)
& \path{tket_paulisimp_rebased}
& direct bridge, \path{DecomposeBoxes}, \texttt{AutoRebase(\{CX, RZ, Rx\})}, \path{RemoveRedundancies}, \path{PauliSimp}, \path{RemoveRedundancies}, target-basis cleanup
& generic identical preamble
& 600\,s
& main, irr \\
TKET GuidedPauliSimp
& \path{tket_guided_paulisimp}
& same TKET bridge/fallback sequence, then \path{GuidedPauliSimp}
& none
& 600\,s
& main \\
staq rotation folding
& \path{staq_rotation_folding}
& Qiskit opt0 target-basis QASM2, public \path{staq} \path{fold_rotations}, QASM2 parse, target-basis cleanup
& no family labels; generic double-precision parser/serializer patch disclosed
& 600\,s
& main, irr \\
phase-polynomial reference
& \path{phase_poly_reference}
& recognize the $H D^r H$ witness, aggregate commuting \texttt{rz}/\texttt{cp} coefficients, then lower through Qiskit opt0
& none
& 600\,s
& main, abl, res, corr \\
\bottomrule
\end{tabularx}
\caption{Method/provenance map for the canonical validation protocol. The direct Qiskit-to-TKET bridge was available in the validation environment; the QASM fallback is documented for reproducibility when that bridge is unavailable.}
\label{tab:method-provenance}
\end{table*}

The source classes in \cref{tab:method-provenance} expand as follows:
Here ``main'' denotes membership in the canonical comparison campaign and does
not necessarily imply inclusion in the two headline tables.
\begin{itemize}[leftmargin=1.5em]
    \item main: canonical 600\,s external-baseline comparison.
    \item abl: Fourier-layer IR ablation.
    \item res: resource-consequence check.
    \item corr: large-instance correctness extension.
    \item irr: irrational independent-Pauli unbounded witness suite.
    \item status: provenance-only configuration excluded from numerical
    quality comparison when its predicate fails.
    \item real: official real-instance frozen panel.
    \item appx: older frozen-control appendix panels.
\end{itemize}

The new executable protocol and provenance are in
\path{research/round6_major_revision/}: the witness declaration and exact-rank
verifier, \path{unbounded_witness_results.json},
\path{rational_cp_round6_results.json},
\path{public_phase_folding_results.json},
\path{tket_paulisimp_results.json}, per-cell logs, and the environment
snapshot. After global-phase alignment, the implemented check computes the
maximum elementwise error $\epsilon$ and requires
$\epsilon \leq \mathrm{atol}+\mathrm{rtol}\max_{i,j}|U_{ij}|$. Thus
$\mathrm{atol}=\mathrm{rtol}=10^{-9}$ gives an effective upper threshold of
approximately $2\times10^{-9}$ because unitary entries have magnitude at most
one. Recorded maximum errors increase with repetition count $r$, while every
row marked completed and valid satisfies this criterion. The irrational-witness
PyZX $10$k result passes at approximately $1.05\times10^{-9}$; the $20$k result
at approximately $2.68\times10^{-9}$ is a correctness failure. Timeout,
predicate-error, unsupported-conversion, and correctness-failure statuses are
kept distinct. The staq provenance records
public commit \texttt{a2acd39} (with the full hash in the artifact) and the
generic double-precision OpenQASM patch used for continuous-angle round trips.

\subsection{Certificate, Cut-Trace, and QRE Artifacts}

The exact canonicalizer and symbolic checker are
\path{certificates/canonicalize.py} and
\path{certificates/symbolic_checker.py}.  The frozen audit
\path{data/frozen/certificate_audit.json} contains 12 accepted/rejected dense
cross-check cases over widths one through six and 56 adversarial mutations.
The checker and dense evaluator use separate implementations: the former
propagates complete signed Clifford tableaus and exact rational-$\pi$ Pauli
coefficients, while the latter constructs full numerical matrices and a
projective distance.  Only \texttt{completed\_valid} records carrying an
independent certificate enter numerical means; unsupported, predicate-error,
and numerical-inconclusive records remain status-only.

The canonical W3 theorem-model reference compilers are under
\path{compiler/reference_streaming/} and are driven by
\path{scripts/run_tradeoff_sweep.py}.  After every input update, their
version-one event schema records \texttt{committed\_output\_bytes}, the five
cut fields from \eqref{eq:info-budget}, pass index, RSS, elapsed nanoseconds,
artifact paths, and digests.  The verifier reopens and decodes every field
slice and checks the committed-prefix hashes.  The complete 592-run campaign
is stored under \path{data/runs/w3-tradeoff-20260802-local/}; its registered
summary and freeze manifest are \path{data/frozen/tradeoff_summary.csv} and
\path{data/frozen/tradeoff_freeze_manifest.json}.  The older
\path{research/cut_trace/} bundle was produced by an earlier serializer and
lacks the per-event field files, so its 1,320 events and 24 summaries are
historical only and are not accepted by the version-one verifier.  The checked
three-cut codec fixture and its manifests remain in the project-root
\path{instrumentation/example_trace/}.

The W7 runner and allocator are the project-root
\path{scripts/run_qre_campaign.py} and
\path{scripts/allocate_synthesis_error.py}; their total-error schema, two QEC
profiles, and factory grid are under \path{qre/}.  The 288-row Parquet/CSV,
full angle records, analysis, figures, and manifests are under
\path{data/frozen/}, \path{data/runs/w7-fixed-total-error-20260803-local/}, and
\path{figures/}.  The earlier single-profile runner and its six-row result are
retained as historical artifacts under \path{research/fixed_budget_qre/} but
are not the W7 main result.  The grid synthesizer is built from
public staq commit
\texttt{a2acd39e60ed0e5f1978fa530df7e1e0c7cedf7a}; the disclosed artifact change
only replaces the nondeterministic candidate-search seed with 20260801.  The
runner omits staq's \texttt{--time} fast path because that path bypasses its
internal detail/check branch; it instead measures subprocess time externally
and requires \texttt{Check flag = 1} for every unique angle.  The QRE is the
explicit surface-code model family stated in the main text, not a vendor
service or an independent architecture implementation.

\subsection{Immutable Campaign Hierarchy}

The machine-readable index
\path{research/campaign_manifests/campaign_index.json} defines two base
profiles and five noninterchangeable campaigns.  \texttt{RP23} is the
pass-order probe, \texttt{HW24} is the canonical line-backend campaign,
\texttt{REAL24} is the official real-instance anti-regression panel,
\texttt{MR25} is the matched-representation matrix, and \texttt{NF25} is the
natural-kernel factorial ablation.  Each prefix is unique and must appear in a
reader-facing label before the source instance name.

Every campaign manifest records or explicitly marks unavailable the backend
topology, native gate set, seed, timeout policy, memory cap, tool versions,
artifact/freeze commit, bridge sequence, certificate policy, immutable config
hash, and frozen dataset hash.  Historical campaigns did not enforce a memory
cap or independent certificate; their manifests say
\texttt{not\_enforced\_historical} and ``diagnostic only'' rather than imputing
a modern policy.  In particular, \texttt{HW24} uses 120\,s for QPE/QAOA and
240\,s for Grover, whereas \texttt{RP23} uses 600\,s.  Those rows may share an
MQT source-instance string but cannot share a short label or be averaged.

The configuration hierarchy is single-parent: a campaign inherits only its
named base profile and then supplies explicit overrides.  Backend and config
hashes use canonical JSON.  The frozen source snapshots
\path{frozen_hardware_aware_rows.json} and
\path{frozen_real_panel_rows.json} live beside the manifests, so later changes
to a runner cannot silently change a displayed row.

\subsection{Claim--Evidence Matrix}

Every headline statement in the abstract, introduction, resource section, and
conclusion is registered with one primary theorem, proposition, table, figure,
or artifact test and an explicit scope.  The machine-readable source is
\path{research/claim_evidence/claim_evidence.json}; the compact reader map is
\cref{tab:claim-evidence}.

\begin{table*}[tbp]
\centering
\scriptsize
\begin{tabularx}{\textwidth}{@{}l l l X@{}}
\toprule
ID & claim location & primary evidence & authorized scope \\
\midrule
\input{research/claim_evidence/claim_evidence_rows.tex} \\
\bottomrule
\end{tabularx}
\caption{Claim--evidence matrix.  The full JSON also contains the exact claim
text and evidence class.  A claim outside the displayed scope is not made by
this paper.}
\label{tab:claim-evidence}
\end{table*}

%% file: research/claim_evidence/claim_evidence_rows.tex
C01 & abstract/introduction & \cref{def:compiler,eq:info-budget} & formal model only \\
C02 & abstract/conclusion & \cref{thm:output-counting} & declared family-relative and self-contained codes \\
C03 & abstract/conclusion & \cref{thm:main-tradeoff} & two-round masked-share stream \\
C04 & introduction/discussion & \cref{prop:hybrid-upper} & ordered-share family-relative codec \\
C05 & validation & \cref{fig:matched-interaction} & 1,080 frozen cells; 24 cells per representation--compiler pair \\
C06 & validation & \cref{fig:measured-tradeoff} & 592 formal-model reference-compiler runs \\
C07 & validation/workloads & \cref{tab:natural-factorial,fig:natural-factorial} & 9 families; $2^6$ factors; 3 scales and 3 seeds \\
C08 & resource consequences & \cref{tab:resource-consequence} & one disclosed surface-code model \\
C09 & compiler & \cref{prop:symbolic-soundness} & declared symbolic domain only \\
C10 & validation & \cref{fig:matched-interaction,sec:external-baselines} & native and unified configurations; bridge loss reported separately \\
C11 & application workloads & \cref{tab:real-instances} & REAL24 campaign \\
C12 & conclusion & \cref{tab:evidence-classes} & evidence-class separation \\
C13 & abstract/conclusion & \cref{thm:w1-unified-frontier} & round-major forward scans; all 2p-1 crossings charged \\
C14 & abstract/discussion & \cref{prop:w1-block-hybrid} & family-relative p-block codec; compact-output corner for every r; full curve for r=2

%% file: appendix/appendix_C_prxq.tex
\section{Supplementary Experiments}
\label{app:supplementary-experiments}

In this appendix, rows labeled ``baseline UCC'' are frozen-control rows:
they denote the upstream UCC v0.4.12 pre-semantic default path defined in the
provenance map of Appendix~B, not a separate semantic method.

\subsection{Exact-Semantics Witness Data}
\label{app:witness-data}

The following tables contain the full uniform-600\,s comparison data for
the exact-semantics witnesses discussed in
\cref{cor:exact-semantics-corollary} of the main text.

\begin{table*}[htbp]
\centering
\scriptsize
\resizebox{\textwidth}{!}{\begin{tabular}{cccccccc}
\toprule
Requested & Actual & semantic UCC & \texttt{qiskit opt3} & PyZX \texttt{full\_reduce} & TKET PauliSimp (rebased) & staq rotation folding & TKET GuidedPauliSimp \\
\midrule
4{,}000 & 3{,}998 & 42/24/12 & 9588/5593/4788 & 107/59/31 & 52/29/19 & 9188/4996/4788 & 13574/6392/4788 \\
10{,}000 & 9{,}998 & 42/24/12 & 23988/13993/11988 & timeout & 62/36/28 & 22988/12496/11988 & 33974/15992/11988 \\
20{,}000 & 19{,}998 & 42/24/12 & 47988/27993/23988 & timeout & 59/27/20 & 45988/24996/23988 & 67974/31992/23988 \\
50{,}000 & 49{,}998 & 42/24/12 & 119988/69993/59988 & timeout & 73/44/35 & 114988/62496/59988 & 169974/79992/59988 \\
100{,}000 & 99{,}998 & 42/24/12 & 239988/139993/119988 & timeout & 66/37/28 & 229988/124996/119988 & 339974/159992/119988 \\
\bottomrule
\end{tabular}}
\caption{Uniform-600\,s comparison on the rational CP finite-range
witness. Completed cells report gates/depth/cx; other entries are status
outcomes and are not scored as structural wins. The family is certified
only for $R < T = 131040$, with per-transducer tail constants. No
external row is a formal theorem instance.}
\label{tab:rational-cp-600}
\end{table*}

\begin{table*}[htbp]
\centering
\scriptsize
\resizebox{\textwidth}{!}{\begin{tabular}{cccccccc}
\toprule
Requested & Actual & semantic UCC & \texttt{qiskit opt3} & PyZX \texttt{full\_reduce} & TKET PauliSimp (rebased) & staq rotation folding & TKET GuidedPauliSimp \\
\midrule
4{,}000 & 3{,}998 & 27/18/6 & 4799/2801/2394 & 38/19/8 & 36/14/6 & 3402/2201/2394 & 8786/4000/2394 \\
10{,}000 & 9{,}998 & 27/18/6 & 12000/7001/5994 & 48/33/12 & 35/15/8 & 8502/5501/5994 & 21986/10000/5994 \\
20{,}000 & 19{,}998 & 27/18/6 & 24000/14001/11994 & correctness failure & 38/16/8 & 17002/11001/11994 & 43986/20000/11994 \\
50{,}000 & 49{,}998 & 27/18/6 & 60000/35001/29994 & timeout & 36/15/8 & 42502/27501/29994 & 109986/50000/29994 \\
100{,}000 & 99{,}998 & 27/18/6 & 120000/70001/59994 & timeout & 37/16/8 & 85002/55001/59994 & 219986/100000/59994 \\
\bottomrule
\end{tabular}}
\caption{Uniform-600\,s comparison on the irrational independent-Pauli
witness. Completed cells report gates/depth/cx; other entries are status
outcomes and are not scored as structural wins. The exact rank and
irrational-angle certificate, rather than this numerical table,
instantiate \cref{cor:exact-semantics-corollary}.}
\label{tab:irrational-600}
\end{table*}

\subsection{Fourier Witness Diagnostics}

The tables in this subsection support the main-text Fourier witness without
serving as additional theorem instances. The canonical theorem-linked CP
witness is the finite-range case in \cref{tab:rational-cp-600}; the
generalized topology/angle rows are robustness diagnostics, and the external
stress rows use the budgets recorded in their frozen source files.

\begin{table*}[htbp]
\centering
\scriptsize
\begin{tabularx}{\textwidth}{@{}lccccX@{}}
\toprule
Topology & Angles & Points & \texttt{qiskit} ok/TO & Semantic completed & Semantic output signatures \\
\midrule
chain CP & 3 & 12 & 9 / 3 & all 12 & 23--27 gates, 17--18 depth, 6 \texttt{cx} \\
ring CP & 3 & 12 & 9 / 3 & all 12 & 32 gates, 23 depth, 8 \texttt{cx} \\
sparse CP $0.5$ & 3 & 12 & 6 / 6 & all 12 & 25--27 gates, 11--12 depth, 6 \texttt{cx} \\
full-pair CP & 2 & 6 & 4 / 2 & all 6 & 42 gates, 24 depth, 12 \texttt{cx} \\
\midrule
Total points & -- & 42 & 28 / 14 & all rows & repetition-independent at fixed topology/angle \\
\bottomrule
\end{tabularx}
\caption{Generalized fixed-width Fourier witness suite. Each point is a paired run of semantic UCC and \texttt{qiskit opt3} on an $H D^r H$ circuit at $n=4$, with requested sizes up to $50$k gates. The rows are filtered from the recorded 284-row generalized suite by the $n=4$ topology/angle/size slice. Recorded \texttt{qiskit opt3} timeout cells in that source use legacy 60\,s or 90\,s budgets, so the table is a robustness diagnostic rather than a uniform-600\,s comparison.}
\label{tab:generalized-fourier-witness}
\end{table*}

\begin{table}[htbp]
\centering
\scriptsize
\setlength{\tabcolsep}{2pt}
\begin{tabular}{cccccc}
\toprule
Width ($n$) & Checks & Output Gates & Depth & CX Count & Seed Variance \\
\midrule
4 & 15 / 15 & 42 & 24 & 12 & 0 \\
5 & 15 / 15 & 65 & 32 & 20 & 0 \\
6 & 15 / 15 & 93 & 40 & 30 & 0 \\
\bottomrule
\end{tabular}
\caption{Correctness and seed-robustness certificate for the full-pair Fourier witness. The correctness experiment checks $3$ widths $\times$ $3$ requested sizes $\times$ $5$ angle seeds using Qiskit \texttt{Operator.equiv}; all $45$ configurations are equivalent to the original circuit up to global phase. The seed-robustness experiment contains $90$ rows including \texttt{qiskit opt3}; semantic UCC has zero output gate-count variance across the five tested nonresonant angle seeds at each width.}
\label{tab:fourier-correctness-seed}
\end{table}

\begin{table*}[tbp]
\centering
\scriptsize
\begin{tabularx}{\textwidth}{@{}>{\raggedright\arraybackslash}Xccc
  >{\raggedright\arraybackslash}X
  >{\raggedright\arraybackslash}X@{}}
\toprule
Method & Requested & Actual & Status & Gates/depth/\texttt{cx} & Certificate \\
\midrule
semantic UCC (Fourier-layer IR enabled) & 100{,}000 & 99{,}998 & completed & 42/24/12 & \texttt{Operator.equiv=True} \\
phase-polynomial reference & 100{,}000 & 99{,}998 & completed & 42/24/12 & \texttt{Operator.equiv=True} \\
\bottomrule
\end{tabularx}
\caption{Large-instance correctness extension on the canonical $n=4$, $r=9999$ Fourier witness. The check is a finite-instance Qiskit \texttt{Operator.equiv} certificate for the two completed compressed outputs only; it does not certify the generalized topology/angle suite or prove an algebraic identity for all $r$.}
\label{tab:large-correctness-extension}
\end{table*}

\begin{table*}[t]
\centering
\scriptsize
\begin{tabular}{@{}rrrrrr@{}}
\toprule
Width & Requested & semantic UCC & \texttt{qiskit opt3} & PyZX bridge & TKET FullPeephole \\
\midrule
5 & 4{,}000 & 65 & 11{,}711 & 14{,}640 & 30{,}369 \\
5 & 10{,}000 & 65 & 29{,}315 & 36{,}640 & timeout \\
6 & 4{,}000 & 93 & 12{,}108 & 15{,}321 & 33{,}881 \\
6 & 10{,}000 & 93 & 30{,}417 & 38{,}487 & timeout \\
\bottomrule
\end{tabular}
\caption{Representative external stress tests on the width-axis full-pair Fourier witness. Entries report output gates after the configured bridge pipeline and final target-basis compilation. PyZX completes all four runs but does not recover the bounded semantic forms $65$ and $93$; TKET completes the $4$k runs with larger outputs and times out at $10$k. The recorded budget in this stress file is 300\,s.}
\label{tab:width-axis-external-stress}
\end{table*}

\begin{table*}[tbp]
\centering
\scriptsize
\begin{tabular}{@{}lrrrr@{}}
\toprule
Topology & semantic UCC & \texttt{qiskit opt3} & PyZX bridge & TKET FullPeephole \\
\midrule
chain CP & 27 & $6{,}851$ / $17{,}136$ & $10{,}838$ / $27{,}121$ & $20{,}555$ / timeout \\
ring CP & 32 & $7{,}996$ / $19{,}996$ & $11{,}984$ / $29{,}984$ & $23{,}987$ / $59{,}987$ \\
sparse CP $0.5$ & 27 & $6{,}851$ / $17{,}136$ & $10{,}838$ / $27{,}121$ & timeout / timeout \\
full-pair CP & 42 & $10{,}783$ / $26{,}983$ & $13{,}574$ / $33{,}974$ & timeout / timeout \\
\bottomrule
\end{tabular}
\caption{Topology-axis external stress test at $n=4$ with nonresonant seed $0$. Each external-baseline cell reports output gates at requested sizes $4$k / $10$k. The deterministic sparse-CP draw uses the same number of CP edges as chain CP and is graph-isomorphic to a chain at $n=4$, which explains the matching \texttt{qiskit opt3} and PyZX gate totals for those two rows. The recorded budget in this stress file is 300\,s.}
\label{tab:topology-external-stress}
\end{table*}

\begin{table*}[t]
\centering
\scriptsize
\begin{tabularx}{\textwidth}{@{}>{\raggedright\arraybackslash}p{0.18\textwidth}
                >{\raggedright\arraybackslash}p{0.19\textwidth}
                >{\raggedright\arraybackslash}X
                >{\raggedright\arraybackslash}p{0.19\textwidth}@{}}
\toprule
Pipeline & Representation tested & Observed behavior in the disclosed suites & Interpretation \\
\midrule
semantic UCC (Fourier-layer IR enabled)
& pre-basis Fourier-layer term
& completed 5/5; $42$ gates, depth $24$, $12$ \texttt{cx} at every tested scale
& semantic side \\
phase-polynomial reference
& same Qiskit input; only commuting-diagonal coefficient aggregation
& completed 5/5; $42$ gates, depth $24$, $12$ \texttt{cx} at $4$k--$100$k
& authors' protocol-specific constructive control \\
\texttt{qiskit opt3}
& preset target-basis pipeline
& completed 5/5; grows from $9{,}588$ gates at $4$k to $239{,}988$ gates at $100$k
& strong completed baseline, but not bounded \\
PyZX configured bridge
& Qiskit--QASM--PyZX basic optimization--Qiskit
& completed 4/5; grows through $169{,}974$ gates at $50$k and times out at $100$k
& configured bridge does not recover bounded form \\
PyZX \texttt{full\_reduce}
& QASM to PyZX graph; \texttt{full\_reduce}; extract and lower
& rational witness: completes only at $4$k with $107$ gates, depth $59$, $31$
\texttt{cx}; timeout at $10$k--$100$k. Irrational witness: completes and
validates at $4$k and $10$k with $38$ and $48$ gates; correctness failure at
$20$k; timeout at $50$k and $100$k
& strong small-scale simplification; correctness failure remains distinct from timeout and error \\
TKET FullPeephole
& direct Qiskit--TKET bridge to FullPeepholeOptimise
& timeout on all five canonical sizes under the 600\,s budget
& no completed bounded recovery in this protocol \\
TKET PauliSimp (unrebased configuration)
& direct Qiskit--TKET bridge to PauliSimp
& predicate error on all five sizes for both witness families
& provenance-only configuration; not a structural loss \\
TKET PauliSimp (rebased)
& generic DecomposeBoxes, AutoRebase to \{CX,Rz,Rx\}, RemoveRedundancies,
PauliSimp, RemoveRedundancies
& completes and validates at all five scales: $52$--$73$ gates on the rational
witness and $35$--$38$ on the irrational witness
& generic aggregate reconstruction crosses the formal boundary \\
staq rotation folding
& generic target-basis QASM2 through public \texttt{fold\_rotations}
& completes and validates at all five scales, but grows from $9{,}188$ to
$229{,}988$ gates on the rational witness and from $3{,}402$ to $85{,}002$ on
the irrational witness
& completed public baseline; growing output is not an impossibility result \\
TKET GuidedPauliSimp
& direct Qiskit--TKET bridge to GuidedPauliSimp
& completed 5/5 on each witness; grows from $13{,}574$ to $339{,}974$ gates on
the rational witness and from $8{,}786$ to $219{,}986$ on the irrational witness
& successful pipeline, but not bounded \\
artifact UCC (Fourier-layer IR disabled)
& same artifact with Fourier-layer semantic path disabled
& in ablation/resource runs, completes at $4$k, $10$k, and $20$k with $9{,}588$, $23{,}988$, and $47{,}988$ gates; times out at $50$k and $100$k
& semantic-path causality check; completed rows materialize large outputs \\
\bottomrule
\end{tabularx}
\caption{External-pipeline status matrix and constructive controls for the
rational CP and irrational independent-Pauli witnesses. Completed outputs and
status outcomes are separated. The rebased and unrebased PauliSimp
configurations are distinct, and no row is treated as a formal transducer
instance.}
\label{tab:external-failure-matrix}
\end{table*}

\subsection{Controlled Fixed-Basis Results}

\subsubsection{\texorpdfstring{Repeated $\mathrm{QFT}+\mathrm{QFT}^{-1}$ ($100$k gates)}{Repeated QFT+inverse-QFT (100k gates)}}

\subsubsection{\texorpdfstring{Repeated $\mathrm{QFT}+\mathrm{QFT}$ control ($100$k gates)}{Repeated QFT+QFT control (100k gates)}}

\begin{table*}[!t]
\centering
\begin{tabular}{lrr rr}
\toprule
& \multicolumn{2}{c}{$\mathrm{QFT}+\mathrm{QFT}^{-1}$}
& \multicolumn{2}{c}{$\mathrm{QFT}+\mathrm{QFT}$} \\
\cmidrule(lr){2-3}\cmidrule(lr){4-5}
Method & Output Gates & Output Depth & Output Gates & Output Depth \\
\midrule
translation only & 400{,}000 & 142{,}500 & 400{,}000 & 140{,}000 \\
baseline UCC & 968{,}784 & 276{,}266 & 1{,}207{,}504 & 317{,}501 \\
semantic UCC & 0 & 0 & 347{,}500 & 125{,}000 \\
\bottomrule
\end{tabular}
\caption{Controlled fixed-basis results for the repeated inverse and
non-inverse QFT controls at 100k input gates.}
\label{tab:qft-inverse-fixed}
\label{tab:qft-control-fixed}
\end{table*}

These finite controls show that, for this configured family and the tested
pipelines, a large structural opportunity remains after controlling for the
final target basis.  They are not a separate theorem-level witness.

\subsubsection{Systematic \texorpdfstring{$\mathrm{QFT}+\mathrm{QFT}^{-1}$}{QFT+inverse-QFT} external scaling}

The inverse-QFT scaling control is not the primary theorem witness, because
\texttt{qiskit opt3} also finds the zero-gate circuit at smaller sizes. Its role
is to test scalability when both paths return the same zero-gate output. As shown in
\cref{tab:qft-inverse-external-scaling}, the semantic artifact returns the
zero-gate circuit at all tested sizes, is faster than \texttt{qiskit opt3}
wherever \texttt{qiskit opt3} completes, and avoids the \texttt{qiskit opt3}
timeouts at $50$k and $100$k. As an additional control, it also matches
\texttt{qiskit\_commutative\_inverse} in output quality and is faster at every
tested size.

\begin{table*}[htbp]
\centering
\small
\begin{tabular}{rrrrrr}
\toprule
Input gates & \multicolumn{2}{c}{\texttt{qiskit opt3}} & \multicolumn{2}{c}{\texttt{qiskit\_commutative\_inverse}} & semantic UCC \\
\cmidrule(lr){2-3}\cmidrule(lr){4-5}
 & Gates & Runtime & Gates & Runtime & Runtime \\
\midrule
4{,}000 & 0 & 1.052 s & 0 & 0.082 s & 0.063 s \\
10{,}000 & 0 & 6.090 s & 0 & 0.182 s & 0.124 s \\
20{,}000 & 0 & 24.596 s & 0 & 0.341 s & 0.279 s \\
50{,}000 & timeout & $>90$ s & 0 & 0.979 s & 0.829 s \\
100{,}000 & timeout & $>90$ s & 0 & 2.271 s & 1.651 s \\
\bottomrule
\end{tabular}
\caption{External scaling on repeated $\mathrm{QFT}+\mathrm{QFT}^{-1}$. The semantic UCC artifact returns zero gates at all tested sizes; the table reports runtime because output quality is identical to the zero-gate reference wherever the references complete. Timeout cells in this source use the recorded 90\,s budget.}
\label{tab:qft-inverse-external-scaling}
\end{table*}

\subsection{Why Simple Pass Reordering Is Not Enough}

To test whether the observed gains are mainly explained by a missing pass or
pass-ordering issue inside \texttt{UCCDefaults}, we ran a controlled comparison
between baseline UCC, several minimal \texttt{UCCDefaults}-style alternatives,
and the semantic artifact. The minimal alternatives were
\texttt{defaults\_preinverse\_only}, \texttt{defaults\_reordered\_only}, and
\texttt{defaults\_minimal\_bundle}; the all-to-all cases were normalized to the
same final target basis $B = [\texttt{cx},\texttt{rx},\texttt{ry},\texttt{rz},\texttt{h}]$,
and the backend-aware cases used fixed transpiler seed $12345$. This is a
separate reorder-probe experiment, so its runtime numbers should be read
independently from the frozen mainline tables below.  Its hardware rows have
reader labels \texttt{RP23-qpeexact-s12345-c89463b96} and
\texttt{RP23-qaoa-s12345-c89463b96}; the canonical hardware campaign below
instead uses the \texttt{HW24-...-c0f32747f} labels.  The differing labels are
different full experiment keys, not repeated measurements of one row.

\begin{table*}[t]
\centering
\scriptsize
\begin{tabular}{llrrrr}
\toprule
Case & Method & Gates & Depth & CX count & Runtime \\
\midrule
\multirow{6}{*}{\texttt{phase\_estimation\_real}}
  & baseline UCC & 471{,}819 & 328{,}271 & 58{,}491 & 75.538 s \\
  & \texttt{defaults\_preinverse\_only} & 485{,}780 & 332{,}924 & 82{,}925 & 80.045 s \\
  & \texttt{defaults\_reordered\_only} & 485{,}780 & 332{,}924 & 82{,}925 & 80.005 s \\
  & \texttt{defaults\_minimal\_bundle} & 485{,}780 & 332{,}924 & 82{,}925 & 80.084 s \\
  & semantic UCC & 166{,}805 & 119{,}204 & 58{,}491 & 1.989 s \\
  & \texttt{qiskit opt3} & 166{,}805 & 119{,}204 & 58{,}491 & 1.254 s \\
\midrule
\multirow{6}{*}{\texttt{qaoa\_ring\_100k}}
  & baseline UCC & 300{,}055 & 163{,}055 & 40{,}000 & 2.431 s \\
  & \texttt{defaults\_preinverse\_only} & 434{,}104 & 217{,}962 & 40{,}000 & 11.352 s \\
  & \texttt{defaults\_reordered\_only} & 445{,}216 & 217{,}962 & 40{,}000 & 10.120 s \\
  & \texttt{defaults\_minimal\_bundle} & 434{,}104 & 217{,}962 & 40{,}000 & 10.818 s \\
  & semantic UCC & 100{,}000 & 62{,}000 & 40{,}000 & 3.131 s \\
  & \texttt{qiskit opt3} & 100{,}000 & 62{,}000 & 40{,}000 & 0.307 s \\
\midrule
\input{research/experiment_registry/generated/reorder_probe_rows.tex}\\
\bottomrule
\end{tabular}
\caption{Controlled comparison of baseline UCC, minimal \texttt{UCCDefaults}-style alternatives, the semantic artifact, and \texttt{qiskit opt3}. Small default-pipeline fixes explain part of the improvement on some cases, but they do not account for the main gains on the real and structured workloads. The overlapping \texttt{phase\_estimation\_real} rows use the per-method output \texttt{cx} counts from the recorded real-instance data.  The keyed hardware rows are generated from the frozen registry; their \texttt{RP23} labels distinguish this reorder-probe configuration from the \texttt{HW24} campaign below.}
\label{tab:defaults-vs-full}
\end{table*}

Two qualifications are not visible in \cref{tab:defaults-vs-full} itself.
On \texttt{RP23-qaoa-s12345-c89463b96} the minimal bundle already matches
the semantic artifact, so the reorder-only explanation is adequate there.
On \texttt{RP23-qpeexact-s12345-c89463b96} it is not: the minimal bundle
gives $1{,}893$ gates, depth $576$, \texttt{cx} $1{,}311$, against
$1{,}593$, depth $420$, \texttt{cx} $772$ for the semantic artifact.

\subsection{Public Benchmark Suite Results}

To move beyond official library examples, we added two public benchmark sources: \texttt{MQT Bench} and \texttt{SupermarQ}.

\paragraph{MQT Bench.}
Representative results are:
\begin{itemize}[leftmargin=1.5em]
    \item \texttt{mqt\_qpeexact\_32}: \texttt{qiskit opt3} $= 1{,}891$ gates, baseline UCC $= 5{,}494$, semantic UCC $= 1{,}891$;
    \item \texttt{mqt\_qpeinexact\_24}: \texttt{qiskit opt3} $= 1{,}206$ gates, baseline UCC $= 3{,}682$, semantic UCC $= 1{,}206$;
    \item \texttt{mqt\_ae\_8}: \texttt{qiskit opt3} $= 207$ gates, baseline UCC $= 542$, semantic UCC $= 207$;
    \item \texttt{mqt\_draper\_qft\_adder\_32}: \texttt{qiskit opt3} $= 1{,}630$ gates, baseline UCC $= 5{,}278$, semantic UCC $= 1{,}630$;
    \item \texttt{mqt\_qaoa\_32}: \texttt{qiskit opt3} $= 1{,}422$ gates, baseline UCC $= 6{,}310$, semantic UCC $= 1{,}422$;
    \item \texttt{mqt\_grover\_20}: \texttt{qiskit opt3} $= 3{,}848{,}810$ gates in $38.043$ s, baseline UCC times out at $>120$ s, and semantic UCC reaches the same output quality in $39.123$ s.
\end{itemize}
Thus, on \texttt{MQT Bench}, the semantic artifact repeatedly recovers
\texttt{qiskit opt3}-level output quality and clearly improves over baseline
UCC.  The added \texttt{qpeinexact}, \texttt{ae}, and
\texttt{draper\_qft\_adder} cases are especially relevant because they extend
the phase-ladder/Fourier-layer story beyond one exact-QPE benchmark into
inexact QPE, amplitude estimation, and QFT-based arithmetic. Runtime still
remains weaker on the hardest Grover-style case.

Additional mirrored/conjugation positive cases make the mirrored-shell
heuristic more concrete. On the official \texttt{grover\_real} instance,
baseline UCC returns
$287{,}047$ gates while the semantic artifact and \texttt{qiskit opt3} both return
$79{,}029$. On \texttt{MQT Bench} Grover instances, the same pattern persists
across scale: \texttt{mqt\_grover\_8} gives $18{,}992 \rightarrow 5{,}170$,
\texttt{mqt\_grover\_12} gives $259{,}797 \rightarrow 71{,}706$, and
\texttt{mqt\_grover\_16} gives $2{,}112{,}285 \rightarrow 581{,}098$, where in
each case the semantic artifact records the same gate count as
\texttt{qiskit opt3} while clearly
improving on baseline UCC.  The large fixed-basis mirrored Grover benchmark shows
the same parity pattern: baseline UCC gives $4{,}122{,}006$ gates, while
the semantic artifact and \texttt{qiskit opt3} both give $1{,}158{,}011$. These cases make the
mirrored/conjugation line less dependent on the routed timeout-recovery
interpretation of \texttt{hw\_mqt\_grover\_20}.

A dedicated mirrored/conjugation scaling sweep makes this support more
systematic. On public \texttt{MQT Bench} Grover instances, the semantic artifact
records the same gate count as \texttt{qiskit opt3} at sizes $8$, $12$, and
$16$, reducing
baseline UCC from $18{,}992$ to $5{,}170$ gates, from $259{,}797$ to $71{,}706$
gates, and from $2{,}112{,}285$ to $581{,}098$ gates, respectively. On the
fixed-basis \texttt{grover\_mirrored} repeated-shell scaling sweep, the same
parity pattern holds at $10$k, $20$k, $50$k, and $100$k input gates:
baseline UCC gives $412{,}206$, $824{,}406$, $2{,}061{,}006$, and
$4{,}122{,}006$ output gates, while the semantic artifact and \texttt{qiskit opt3}
give $115{,}811$, $231{,}611$, $579{,}011$, and $1{,}158{,}011$ gates. Thus
this second structural family targeted by the artifact has empirical positive
cases in both public Grover scaling and large repeated-shell fixed-basis
scaling.

\paragraph{SupermarQ.}
Representative results are:
\begin{itemize}[leftmargin=1.5em]
    \item \texttt{supermarq\_mermin\_bell\_8}: baseline UCC $= 386$ gates, depth $135$; semantic UCC $= 76$ gates, depth $44$; \texttt{qiskit opt3} $= 76$ gates, depth $44$;
    \item \texttt{supermarq\_qaoa\_vanilla\_12}: baseline UCC $= 947$ gates, depth $225$; semantic UCC $= 222$ gates, depth $80$; \texttt{qiskit opt3} $= 222$ gates, depth $80$;
    \item \texttt{supermarq\_hamiltonian\_sim\_8}: baseline UCC $= 71$ gates, depth $36$; semantic UCC $= 51$ gates, depth $24$; \texttt{qiskit opt3} $= 51$ gates, depth $24$.
\end{itemize}
This second public benchmark source provides additional finite anti-regression
checks on nontrivial families and also exposes a limitation: the semantic
artifact is not uniformly better on every benchmark family.

\subsection{Hardware-Aware Results}

We added a backend-aware comparison using a fixed $20$-qubit bidirectional line backend with native operations \texttt{u}, \texttt{sx}, \texttt{p}, \texttt{cx}, \texttt{measure}, and \texttt{id}. The benchmark source is \texttt{MQT Bench}.  The reader-facing rows are \texttt{HW24-qpeexact-s12345-c0f32747f}, \texttt{HW24-qaoa-s12345-c0f32747f}, and \texttt{HW24-grover-s12345-c0f32747f}; the full source instance remains part of each registry primary key.

\paragraph{\texttt{HW24-qpeexact-s12345-c0f32747f}.}
\begin{table*}[!t]
\centering
\begin{minipage}[t]{0.48\textwidth}
\centering
\textbf{(a) \texttt{HW24-qpeexact-s12345-c0f32747f}}\\[2pt]
\resizebox{\linewidth}{!}{\begin{tabular}{lrrrr}
  \toprule
  Method & Output Gates & Output Depth & CX Count & Runtime \\
  \midrule
  \input{research/experiment_registry/generated/hardware_qpe_rows.tex}\\
  \bottomrule
  \end{tabular}}
\end{minipage}\hfill
\begin{minipage}[t]{0.48\textwidth}
\centering
\textbf{(b) \texttt{HW24-qaoa-s12345-c0f32747f}}\\[2pt]
\resizebox{\linewidth}{!}{\begin{tabular}{lrrrr}
  \toprule
  Method & Output Gates & Output Depth & CX Count & Runtime \\
  \midrule
  \input{research/experiment_registry/generated/hardware_qaoa_rows.tex}\\
  \bottomrule
  \end{tabular}}
\end{minipage}
\caption{Canonical hardware-aware results for the QPE and QAOA reader labels,
generated from the frozen keyed dataset.}
\label{tab:hw-qpe}
\label{tab:hw-qaoa}
\end{table*}

This is a clean fixed-seed parity case: the semantic artifact matches
\texttt{qiskit opt3} in total gates, depth, and \texttt{cx} count, while still
remaining much better than baseline UCC on all structural metrics. The tradeoff
is runtime rather than output quality.

\paragraph{\texttt{HW24-qaoa-s12345-c0f32747f}.}
This is the clearest hardware-aware external-baseline improvement in the paper:
the semantic artifact improves on \texttt{qiskit opt3} in total gates, depth,
and \texttt{cx} count, while also improving on baseline UCC in total gates,
depth, and \texttt{cx} count.

\paragraph{\texttt{HW24-grover-s12345-c0f32747f}.}
\texttt{qiskit opt3} and baseline UCC still time out under this setting,
but the semantic artifact returns a result in 1.999\,s via a stronger
mirrored/self-inverse repeated-run backend fallback. The resulting circuit is
still extremely large, but the selector lowers total gates, depth, and
\texttt{cx} simultaneously (to $6{,}467{,}857$ total gates, depth
$4{,}094{,}196$, and \texttt{cx} count $3{,}045{,}048$), so we treat this as a
robustness result rather than a main competitive routed-quality claim.

\subsection{Stability and Seed Sensitivity}

We measured repeated-run stability on the frozen artifact in two ways:
\begin{enumerate}[leftmargin=1.5em]
    \item five repeated runs for the official real instances and the canonical hardware-aware seed,
    \item a five-seed sweep for the backend-aware \texttt{MQT Bench} cases.
\end{enumerate}

The main conclusions are:
\begin{itemize}[leftmargin=1.5em]
    \item the official real-instance outputs are deterministic across repeated runs,
    \item the hardware-aware outputs are deterministic for each fixed seed,
    \item \texttt{hw\_mqt\_qaoa\_20} remains better than \texttt{qiskit opt3} in total gates, depth, and \texttt{cx} count across all five tested seeds $0$, $1$, $42$, $12345$, and $54321$,
    \item \texttt{hw\_mqt\_qpeexact\_20} is now a fixed-seed parity case rather than a strict improvement.
\end{itemize}
Thus, the hardware-aware claim is a stable workload-family result across a small
fixed seed portfolio, rather than a single favorable seed.

\begin{figure*}[t]
\centering
\IfFileExists{figures/scaling_plot.png}{\includegraphics[width=\textwidth]{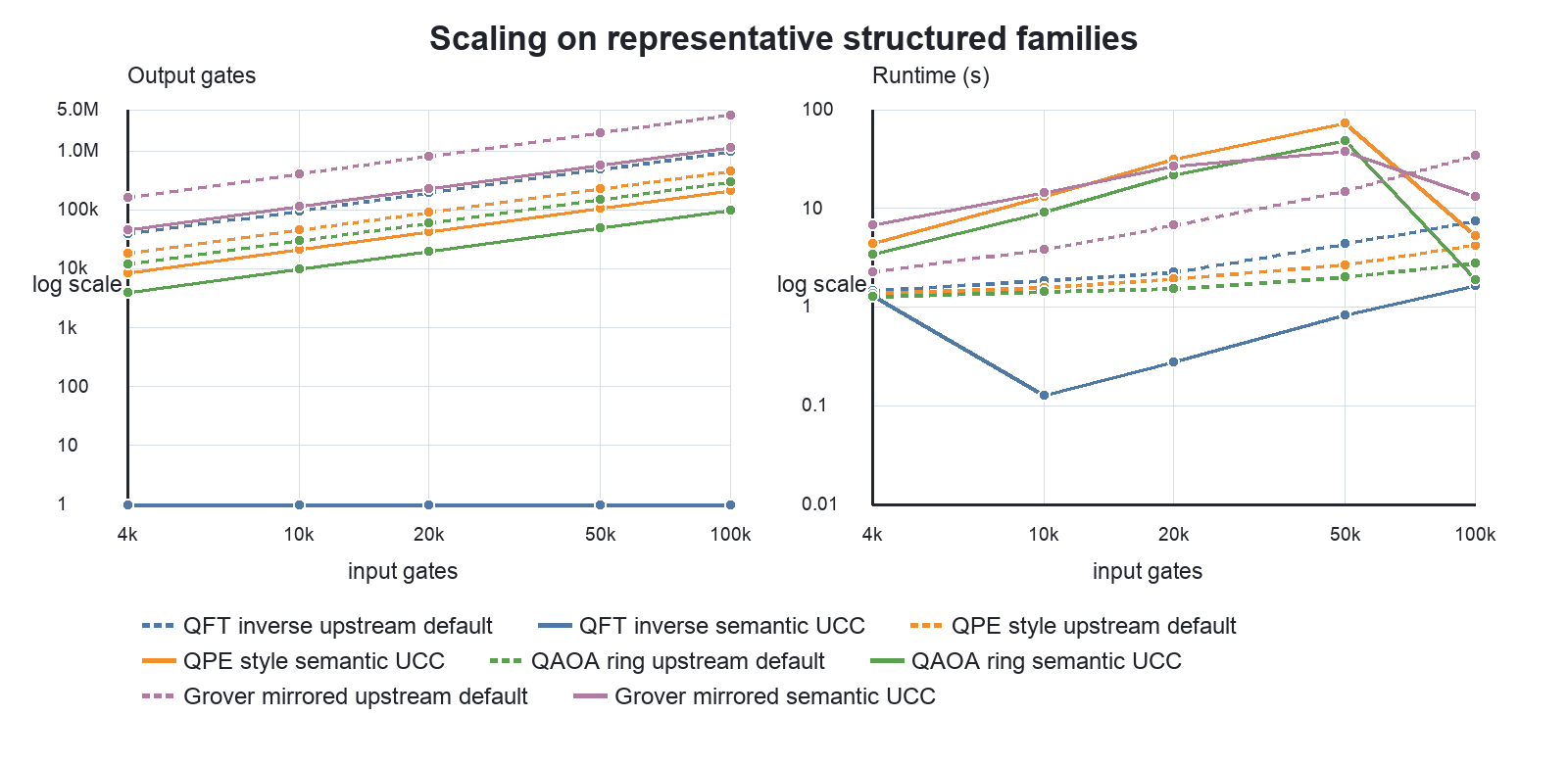}}{\fbox{\parbox{0.9\textwidth}{\centering\vspace{2em}Missing file: \texttt{figures/scaling\_plot.png}.\vspace{2em}}}}
\caption{Scaling behavior on representative structured families. Left: output gate count versus input gate count for upstream default (baseline UCC) and the semantic UCC artifact. Right: runtime versus input gate count. The semantic path preserves clear quality advantages across scale, while runtime behavior remains family-dependent.}
\label{fig:scaling}
\end{figure*}

\subsection{Scaling and Ablation}

The main scaling trends are summarized in \cref{fig:scaling}. The size sweep
covers $4$k, $10$k, $20$k, $50$k, and $100$k inputs for QFT-inverse,
QPE-style, QAOA-ring, and mirrored-Grover families; the QAOA-ring block uses
20 qubits. The most important
observations are stable quality behavior across scale and improving
repeated-family runtime once the semantic representation and dispatch controller are
enabled. The method is most effective on exact or highly repeated structure:
\texttt{qft\_inverse} remains maximally strong across all tested sizes,
\texttt{qpe\_style} remains a strong positive fixed-basis case,
\texttt{qaoa\_ring} consistently avoids UCC inflation, and the mirrored-Grover
family continues to match Qiskit quality at $100$k while remaining
more expensive. Direct spot checks on the semantic artifact give
\texttt{qpe\_style} $=82{,}200$ gates at $50$k in $2.504$ s and $164{,}400$
gates at $100$k in $4.636$ s, reflecting the phase-ladder plus
repeated-run controller improvements.

The non-monotonic semantic runtime between the $20$k and $50$k points is a
dispatch effect visible in the implementation, not an inferred hardware
effect. The 20,000-gate full-preset limit
\path{_MAX_FULL_PRESET_REFERENCE_SIZE} allows the 19,998-gate input to enter
the full preset-reference comparison,
whereas the 49,998-gate input exceeds that threshold and uses the cheaper
repeated-prefix/hierarchical route. The corresponding canonical witness
runtimes are 28.314~s and 20.934~s. This threshold explains the local drop;
it is not evidence of asymptotically decreasing compilation time.

The ablation picture is narrow. The dominant practical component is
conservative candidate control on recoverability-sensitive families, most
clearly on the 10,000-input-gate \texttt{qaoa\_ring} row, where the full configuration returns
$10{,}000$ gates while the no-candidate-selection variant returns $30{,}055$.
By contrast, \texttt{qft\_inverse}, \texttt{qft\_control},
\texttt{qpe\_style}, and \texttt{grover\_mirrored} are less sensitive to the
removal of individual structural rules. This reinforces the paper's current
framing: the main contribution is a bounded semantic representation and
dispatch controller that preserve recoverability before flat materialization,
rather than a single local identity.

%% file: research/experiment_registry/generated/reorder_probe_rows.tex
\multirow{6}{*}{\shortstack{\texttt{RP23-qpeexact}\\\texttt{s12345-c89463b96}}}  & baseline UCC & 2{,}077 & 626 & 1{,}262 & 0.126 s \\
  & \texttt{defaults\_preinverse\_only} & 1{,}893 & 576 & 1{,}311 & 0.292 s \\
  & \texttt{defaults\_reordered\_only} & 1{,}900 & 582 & 1{,}312 & 0.289 s \\
  & \texttt{defaults\_minimal\_bundle} & 1{,}893 & 576 & 1{,}311 & 0.287 s \\
  & semantic UCC & 1{,}593 & 420 & 772 & 0.296 s \\
  & \texttt{qiskit opt3} & 1{,}572 & 407 & 812 & 0.033 s \\
\midrule
\multirow{6}{*}{\shortstack{\texttt{RP23-qaoa}\\\texttt{s12345-c89463b96}}}  & baseline UCC & 2{,}088 & 522 & 1{,}439 & 0.098 s \\
  & \texttt{defaults\_preinverse\_only} & 2{,}050 & 487 & 1{,}458 & 0.270 s \\
  & \texttt{defaults\_reordered\_only} & 2{,}050 & 487 & 1{,}458 & 0.269 s \\
  & \texttt{defaults\_minimal\_bundle} & 2{,}050 & 487 & 1{,}458 & 0.271 s \\
  & semantic UCC & 2{,}050 & 487 & 1{,}458 & 0.278 s \\
  & \texttt{qiskit opt3} & 2{,}091 & 532 & 1{,}530 & 0.048 s

%% file: research/experiment_registry/generated/hardware_qpe_rows.tex
\texttt{qiskit opt3} & 1{,}572 & 407 & 812 & 0.056 s \\
baseline UCC & 1{,}976 & 648 & 1{,}350 & 0.174 s \\
semantic UCC & 1{,}572 & 407 & 812 & 0.235 s

%% file: research/experiment_registry/generated/hardware_qaoa_rows.tex
\texttt{qiskit opt3} & 2{,}091 & 532 & 1{,}530 & 0.078 s \\
baseline UCC & 2{,}057 & 513 & 1{,}503 & 0.141 s \\
semantic UCC & 2{,}050 & 487 & 1{,}458 & 0.112 s